\documentclass[sigconf, nonacm, anonymous=false, review=false, screen=true,authorversion=true]{acmart}
\newcommand\vldbdoi{10.14778/3476249.3476306}
\newcommand\vldbpages{2599-2612}
\newcommand\vldbvolume{14}
\newcommand\vldbissue{11}
\newcommand\vldbyear{2021}
\newcommand\vldbauthors{\authors}
\newcommand\vldbtitle{\shorttitle} 
\newcommand\vldbavailabilityurl{https://github.com/northeastern-datalab/anyk-code}
\newcommand\vldbpagestyle{plain}

\usepackage[a-1b]{pdfx}   %

\usepackage{amsmath}
\usepackage{mathtools}
\usepackage{graphicx} %
\usepackage{url}
\usepackage{enumitem} %

\usepackage{stmaryrd} %

\usepackage{upgreek} %
\usepackage{bbding} %

\usepackage{listings} %
\usepackage{tabularx} %
\usepackage{booktabs} %
\usepackage{multirow}
\usepackage{hhline} %
\usepackage{diagbox}
\usepackage{pgfplotstable} %

\usepackage{xcolor,colortbl}    %

\usepackage{tikz} %
\usetikzlibrary{matrix}
\usetikzlibrary{calc}
\usetikzlibrary{math}
\usetikzlibrary{arrows.meta,positioning}
\usetikzlibrary{intersections, pgfplots.fillbetween}

\usepackage{easybmat} %

\usepackage{adjustbox}
\def\eox{\unskip\kern 10pt{\unitlength1pt\linethickness{.4pt}$\diamondsuit${}}} %

\usepackage{rotating}		%

\usepackage{xspace} %

\usepackage{subcaption} %
\newcommand{\hide}[1]{}

\usepackage[ruled,noend,linesnumbered]{algorithm2e} %

\usepackage{setspace} %

\DontPrintSemicolon     %
\SetNlSty{}{}{}                %
\SetAlgoInsideSkip{smallskip}   %

\SetAlFnt{\small}			%
\SetAlCapFnt{\small}		%
\SetAlCapNameFnt{\small}

\hypersetup{                                                            %
	pdfpagemode=UseNone,               %
    pdfstartview=Fit,
    pdfpagelayout=SinglePage,       %
    bookmarks=false,
    bookmarksnumbered=true,
    bookmarksopen=false,             %
    pdftex=true,
    unicode,
    colorlinks=true,       %
    linkcolor=blue,          %
    citecolor=cyan,  %
    filecolor=magenta,      %
     urlcolor=blue           %
    plainpages=false,           %
     hypertexnames=true,         %
     pdfpagelabels=true,         %
     hyperindex=true,
     naturalnames=false     %
}

\usepackage[capitalise,nameinlink]{cleveref} %

\usepackage{aliascnt}  	

\newtheorem{theorem}{Theorem} %

\newaliascnt{corollary}{theorem}

\aliascntresetthe{corollary}

\newaliascnt{example}{theorem}
\newtheorem{example}[example]{Example}
\aliascntresetthe{example}

\newaliascnt{definition}{theorem}
\newtheorem{definition}[definition]{Definition}
\aliascntresetthe{definition}

\newaliascnt{proposition}{theorem}

\aliascntresetthe{proposition}

\newaliascnt{lemma}{theorem}
\newtheorem{lemma}[lemma]{Lemma}
\aliascntresetthe{lemma}

\newaliascnt{conjecture}{theorem}

\aliascntresetthe{conjecture}

\newtheorem{questionW}{Question}
\newtheorem{resultW}{Result}

{\end{itshape}
}
\AtEndPreamble{%
    \hypersetup{colorlinks,
      linkcolor=purple,
      citecolor=blue,
      urlcolor=ACMDarkBlue,
      filecolor=ACMDarkBlue}}

\usepackage{tcolorbox}		%
\tcbuselibrary{breakable,skins}		%

\tcbset{examplestyle/.style={
		enhanced jigsaw,	%
		colback=blue!10,	%
		colframe=blue!10,	%
		arc=0mm,
		boxrule=0pt,		%
		left=1mm,
		right=1mm,
		left skip= 0mm,  %
		right skip= 0mm, %
		top=1pt,		%
		bottom=1pt,		%
		breakable,		%
		parbox = false,		%
		before={\par\pagebreak[0]\vspace{1mm}\parindent=0pt},		
		after={\par\pagebreak[0]\vspace{1mm}\parindent=0pt},				
		bottomrule = 0mm,
		boxsep = 0mm,					%
		topsep at break=0pt,			%
		bottomsep at break=0pt,			%
		pad at break=0mm,
		pad before break=1mm,		
		pad after break=1mm,		
		bottomrule at break=0mm,
		toprule at break=0mm,		
		}}

\usepackage{footnote}

\tcolorboxenvironment{example}{examplestyle}
\newcommand{\resultbox}[1]{
\begin{tcolorbox}[
	enhanced jigsaw,		%
	colback=red!5,
	colframe=red!75!black,	
	arc=0mm,
	left skip=-1mm,
	right skip=-1mm,	
	left=0mm,
	topsep at break=1mm,			%
	right=0mm,
	top=0mm,
	bottom=0mm,		%
	breakable,		%
	parbox = false		%
]
\emph{#1}
\end{tcolorbox}
}

\makeatletter
\DeclareRobustCommand*\uell{\mathpalette\@uell\relax}
\newcommand*\@uell[2]{
  \setbox0=\hbox{$#1\ell$}
  \setbox1=\hbox{\rotatebox{10}{$#1\ell$}}
  \dimen0=\wd0 \advance\dimen0 by -\wd1 \divide\dimen0 by 2
  \mathord{\lower 0.1ex \hbox{\kern\dimen0\unhbox1\kern\dimen0}}
}
\makeatother

\usepackage{marginnote}

\newcommand{\introparagraph}[1]{\textbf{#1.}} %

\renewcommand{\epsilon}{\varepsilon} %

\newcommand{\datarule}{{\,:\!\!-\,}} %
\renewcommand{\vec}[1]{\boldsymbol{\mathbf{#1}}}
\newcommand{\R}{{\mathbb{R}}} %
\renewcommand{\O}{{\mathcal{O}}} %
\newcommand{\smallO}{{\mathit{o}}} %

\usepackage{scalerel}
\usepackage{tikz}
\usetikzlibrary{svg.path}

\definecolor{orcidlogocol}{HTML}{A6CE39}
\tikzset{
  orcidlogo/.pic={
    \fill[orcidlogocol] svg{M256,128c0,70.7-57.3,128-128,128C57.3,256,0,198.7,0,128C0,57.3,57.3,0,128,0C198.7,0,256,57.3,256,128z};
    \fill[white] svg{M86.3,186.2H70.9V79.1h15.4v48.4V186.2z}
                 svg{M108.9,79.1h41.6c39.6,0,57,28.3,57,53.6c0,27.5-21.5,53.6-56.8,53.6h-41.8V79.1z M124.3,172.4h24.5c34.9,0,42.9-26.5,42.9-39.7c0-21.5-13.7-39.7-43.7-39.7h-23.7V172.4z}
                 svg{M88.7,56.8c0,5.5-4.5,10.1-10.1,10.1c-5.6,0-10.1-4.6-10.1-10.1c0-5.6,4.5-10.1,10.1-10.1C84.2,46.7,88.7,51.3,88.7,56.8z};
  }
}

\RequirePackage{etoolbox}
\DeclareRobustCommand\orcidicon[1]{\href{https://orcid.org/#1}{\mbox{\scalerel*{
\begin{tikzpicture}[yscale=-1,transform shape]
\pic{orcidlogo};
\end{tikzpicture}
}{|}}}}

\usepackage{array}
\usepackage{makecell}

\def\polylog{\operatorname{polylog}}

\newcommand{\TLFG}{TLFG\xspace}
\newcommand{\TLFGs}{TLFGs\xspace}
\newcommand{\BATCH}{\textsc{Batch}\xspace}
\newcommand{\QUADEQUI}{\textsc{QuadEqui}\xspace}
\newcommand{\TT}{\ensuremath{\mathrm{TT}}}
\newcommand{\MEM}{\ensuremath{\mathrm{MEM}}}
\newcommand{\IneqMultiFun}{partIneqMulti}
\newcommand{\IneqBinaryFun}{partIneqBinary}
\newcommand{\BandFun}{bandToIneq}
\newcommand{\NextCond}{nextPredicate}
\newcommand{\BINPART}{\textsc{Binary Partitioning}\xspace}
\newcommand{\MULTIPART}{\textsc{Multiway Partitioning}\xspace}
\newcommand{\SHAREDRAN}{\textsc{Shared Ranges}\xspace}

\newcommand{\dpgraph}{enumeration graph\xspace}

\newcommand{\SYSX}{\textsc{System X}\xspace}
\newcommand{\PSQL}{\textsc{PSQL}\xspace}
\newcommand{\OURS}{\textsc{Factorized}\xspace}
\newcommand{\LAZY}{\textsc{Lazy}\xspace}

\newcommand{\Reddit}{\textsc{RedditTitles}\xspace}
\newcommand{\Birds}{\textsc{OceaniaBirds}\xspace}

\newcommand{\cc}{olive}
\newcommand{\algocomment}[1]{\textcolor{\cc}{{//#1}}}

\newcommand{\Prep}{\mathcal{P}}
\newcommand{\Space}{\mathcal{S}}

\newcommand{\partitions}{\rho}
\newcommand{\out}{|\mathrm{out}|}
\newcommand{\depth}{d}

\newcommand{\DBMSs}{DBMSs\xspace}

\newcommand*\circled[1]{\tikz[baseline=(char.base)]{
            \node[shape=circle,draw,inner sep=0.8pt] (char) {#1};}}	%

\begin{document}

\title{Beyond Equi-joins: Ranking, Enumeration and Factorization}

\author{Nikolaos Tziavelis}
\affiliation{%
    \orcidicon{0000-0001-8342-2177}
    \institution{Northeastern University}
    \city{Boston}
    \state{Massachusetts}
    \country{USA}
}
\email{tziavelis.n@northeastern.edu}

\author{Wolfgang Gatterbauer}
\affiliation{%
    \orcidicon{0000-0002-9614-0504}
    \institution{Northeastern University}
    \city{Boston}
    \state{Massachusetts}
    \country{USA}
}
\email{w.gatterbauer@northeastern.edu}

\author{Mirek Riedewald}
\affiliation{%
    \orcidicon{0000-0002-6102-7472}
    \institution{Northeastern University}
    \city{Boston}
    \state{Massachusetts}
    \country{USA}
}
\email{m.riedewald@northeastern.edu}

\begin{abstract}
We study theta-joins in general and join predicates with
conjunctions and disjunctions of inequalities in particular,
focusing on \emph{ranked enumeration}
where the answers are returned incrementally in an order dictated by a
given ranking function. 
Our approach achieves strong time and space complexity
properties: with $n$ denoting the number of tuples in the database, we guarantee
for acyclic full join queries with 
inequality conditions that for
\emph{every} value of $k$, the $k$ top-ranked answers are returned in
$\O(n \polylog n + k \log k)$ time.
This is within a polylogarithmic factor of 
$\O(n + k \log k)$,
i.e., the best known complexity for equi-joins,
and even of $\O(n + k)$, i.e., 
the time it takes to look at the input and
return $k$ answers in any order.
Our guarantees extend to join queries with selections and many types of
projections
(namely those called ``free-connex'' queries and those that use bag semantics).
Remarkably, they hold even when the number of join results is 
$n^\ell$ for a join of $\ell$ relations.
The key ingredient is a novel $\O(n \polylog n)$-size 
\emph{factorized representation of the query output}, 
which is constructed on-the-fly for a given query and database.
In addition to providing the first non-trivial theoretical guarantees
beyond equi-joins, we show in an experimental study that our ranked-enumeration
approach is also memory-efficient and fast in practice, beating the running time
of state-of-the-art database systems by orders of magnitude.
\end{abstract}

\maketitle

\pagestyle{\vldbpagestyle}
\begingroup\small\noindent\raggedright\textbf{PVLDB Reference Format:}\\
\vldbauthors. \vldbtitle. PVLDB, \vldbvolume(\vldbissue): \vldbpages, \vldbyear.\\
\href{https://doi.org/\vldbdoi}{doi:\vldbdoi}
\endgroup
\begingroup
\renewcommand\thefootnote{}\footnote{\noindent
This work is licensed under the Creative Commons BY-NC-ND 4.0 International License. Visit \url{https://creativecommons.org/licenses/by-nc-nd/4.0/} to view a copy of this license. For any use beyond those covered by this license, obtain permission by emailing \href{mailto:info@vldb.org}{info@vldb.org}. Copyright is held by the owner/author(s). Publication rights licensed to the VLDB Endowment. \\
\raggedright Proceedings of the VLDB Endowment, Vol. \vldbvolume, No. \vldbissue\ %
ISSN 2150-8097. \\
\href{https://doi.org/\vldbdoi}{doi:\vldbdoi} \\
}\addtocounter{footnote}{-1}\endgroup
\ifdefempty{\vldbavailabilityurl}{}{
\vspace{.3cm}
\begingroup\small\noindent\raggedright\textbf{PVLDB Artifact Availability:}\\
The source code, data, and/or other artifacts have been made available at \url{\vldbavailabilityurl}.
\endgroup
}

\section{Introduction}
\label{sec:intro}

Join processing is one of the most fundamental topics in database research,
with recent work aiming at strong asymptotic guarantees
\cite{ngo2018worst,Khamis:2016:JVG:3014437.2967101,navarro19wco,Ngo:2014:SSB:2590989.2590991}.
Work on constant-delay (unranked) enumeration 
\cite{bagan07constenum,DBLP:journals/sigmod/Segoufin15,idris19dynamic,DBLP:conf/pods/CarmeliK19}
strives to pre-process the database for a given query on-the-fly so that
the first answer is returned in linear time (in database size), followed by
all other answers with constant delay (i.e., independent of database size) 
between them. 
Together, linear pre-processing and constant delay 
guarantee that all answers are returned in time linear in
input and output size, which is optimal.

\introparagraph{Ranked enumeration}
Ranked enumeration~\cite{tziavelis20tutorial} 
generalizes the heavily studied \emph{top-$k$} 
paradigm~\cite{fagin03,ilyas08survey} 
by continuously returning join answers in
ranking order. 
This enables the output consumer
to select the cut-off
$k$ on-the-fly while observing the answers.
For top-$k$, the value of $k$ must be chosen in advance,
before seeing any query answer.
Unfortunately, non-trivial complexity guarantees of previous
top-$k$ techniques, including the celebrated
Threshold Algorithm \cite{fagin03}, are limited to the
``middleware'' cost model, which only accounts for the number of distinct
data items accessed~\cite{tziavelis20tutorial}.
While some of those top-$k$ algorithms can be applied to joins with
general predicates, they do not provide
non-trivial guarantees in the standard RAM model of computation, 
and their
time complexity for a join of $\ell$ relations can be $\O(n^\ell)$.

The goal of this paper is to design
\emph{ranked-enumeration algorithms for general theta joins
with strong space and time guarantees in the standard RAM model of computation}.
Tight upper complexity bounds are essential for ensuring predictable performance,
no matter the given database instance (e.g., in terms of data skew)
or the query's total output size.
Notice that it already takes 
$\O(n + k)$ time to simply look at $n$ input
tuples as well as create and return $k$ output tuples. Since polylogarithmic
factors are generally considered small or even negligible for asymptotic
analysis \cite{khamis17panda,Cormen:2009dp}, we aim for time bounds
that are within such polylogarithmic factors of 
$\O(n+k)$.
At the same time, we want space complexity to be reasonable;
e.g., for small $k$ to be within a polylogarithmic factor of $\O(n)$,
which is the required space to hold the input.

While state-of-the-art commercial and open-source \DBMSs
do not yet support
ranked enumeration, it is worth taking a closer look at their implementation
of top-$k$ join queries.
(Here $k$ is specified in a SQL
clause like FETCH FIRST or LIMIT.)
While we tried a large variety of inputs, indexes on the input relations,
join queries, and values of $k$, 
the optimizer 
of PostgreSQL and two other widely used commercial \DBMSs
always
chose to execute the join before applying the ranking and top-$k$
condition on the join results.\footnote{
For non-trivial
ranking functions, or when the attributes used for joining differ from
those used for ranking, the DBMS cannot determine if a
subset of the join output so far produced 
already contains all
$k$ top-ranked answers.
This applies to general theta joins as well as equi-joins.}
This implies that their overall time complexity
to return even the top-$1$ result
cannot be better than the worst-case join output size, which
can be
$\O(n^\ell)$
for a join of $\ell$ relations.

\introparagraph{Beyond equi-joins}
Recent work on ranked
enumeration~\cite{tziavelis20vldb,tziavelis20tutorial,deep21,yang2018any,YangRLG18:anyKexploreDB,ding21progressive}
achieves much stronger worst-case guarantees, but only considers \emph{equi-joins}.
However, big-data analysis often also requires other join conditions
\cite{khayyat17ineq,li20band,enderle04interval,dewitt91band} such as
\emph{inequalities} (e.g., \texttt{S.age < T.age}),
\emph{non-equalities} (e.g., \texttt{S.id $\neq$ T.id}),
and \emph{band} predicates (e.g., \texttt{|S.time - T.time| < $\epsilon$}).
For these joins, two major challenges must be addressed.
First, the join itself must be computed efficiently in the presence of
complex conditions, possibly consisting of conjunctions and disjunctions of
such predicates. Second, to avoid having to produce the entire output,
\emph{ranking has to be pushed deep into the join itself}.

\begin{example}
A concrete application of 
ranked enumeration for inequality joins
concerns graph-based approaches for detecting ``lateral movement'' between infected computers in a network \cite{liu18latte}.
By modeling computers as nodes and connections as timestamped edges,
these approaches search for anomalous access patterns
that take the form of paths
(or more general subgraphs)
ranked by the probability of occurrence according to historical data.
The inequalities arise from a time constraint:
the timestamps of two consecutive edges need to be in ascending order.
Concretely, consider the relation \texttt{G(From,To,Time,Prob)}.
Valid 2-hop paths can be computed with a self-join 
(where $G_1, G_2$ are aliases of $G$)
where the join condition is an equality $\texttt{G}_1.\texttt{To} = \texttt{G}_2.\texttt{From}$
and an inequality $\texttt{G}_1.\texttt{Time} < \texttt{G}_2.\texttt{Time}$,
while the score of a path is $\texttt{G}_1.\texttt{Prob} \cdot \texttt{G}_2.\texttt{Prob}$.
Existing approaches are severely limited computationally in terms of the length of the pattern,
since the number of paths in a graph can be extremely large.
Thus, they usually resort to a search over very small paths (e.g., only $2$-hop). 
With the techniques developed in this paper, patterns of much larger size can be retrieved efficiently in ranked order
without considering all 
possible instantiations of the pattern.
\end{example}

\introparagraph{Main contributions}
We provide the first comprehensive study on ranked enumeration for joins with
conditions other than equality, notably 
general theta-joins and
conjunctions and disjunctions of inequalities and equalities.
While such joins are 
expensive to compute~\cite{khayyat17ineq,li20band},
we show that for many of them
the top-ranked answers can \emph{always} be found in time complexity
that only slightly exceeds the complexity of sorting the input.
This is remarkable, given that the input may be heavily skewed and the output size
of a 
join of $\ell$ relations is $\O(n^\ell)$.
We achieve this with a carefully designed factorized representation
of the join \emph{output} that can be constructed in relatively small
time and space. 
Then the ranking function determines the traversal
order on this representation.

Recall that ranked-enumeration algorithms must continuously output answer tuples
in order 
and the goal is to achieve non-trivial complexity guarantees
no matter at which value of $k$ the algorithm is stopped. Hence we express
algorithm complexity as a function of $k$: 
$\TT(k)$ and $\MEM(k)$ denote the algorithm's time and space complexity,
respectively, until the moment it returns the $k$-th answer in ranking order.
Our main contributions (see also \Cref{tab:summary}) are:

(1) We generalize an 
    equi-join-specific ranked-enumeration
    construction~\cite{tziavelis20vldb} 
    by representing the overall join structure as a tree of joining relations 
    and then introducing a join-condition-sensitive abstraction between each pair of adjacent relations in the tree.
    For the latter, we propose the ``\emph{Tuple-Level Factorization Graph}''
    (\TLFG, \Cref{sec:framework}), a novel factorized representation
    for any theta-join between two relations, and show how its size and
    depth affect the complexity of ranked enumeration. Interestingly,
    some \TLFGs 
	can be used to transform a given theta-join to an equi-join,
    a property we leverage for ranked enumeration for \emph{cyclic} join
    queries.

(2) For join conditions that are a DNF 
of inequalities (\Cref{sec:inequalities}),
we propose concrete \TLFGs 
with space and construction-time
complexity $\O(n \polylog n)$.
Using them for acyclic joins, our algorithm guarantees
$\TT(k) = \O(n \polylog n + k \log k)$, which is within
a polylogarithmic factor of the equi-join case, where
$\TT(k) = \O(n + k \log k)$ \cite{tziavelis20vldb}, and even the lower bound of
$\O(n+k)$.
    
(3) Our experiments (\Cref{sec:exp}) on synthetic and real datasets show
orders-of-magnitude improvements over highly optimized top-$k$ implementations
in 
state-of-the-art \DBMSs, as well as over an idealized competitor
that is not charged for any join-related cost.

Due to space constraints, formal proofs and several details of improvements to our core 
techniques (\cref{sec:improvements})
are in 
the full version of this paper \cite{tziavelis21full}.
Our project website contains more information including source code:
\url{https://northeastern-datalab.github.io/anyk/}.

\newcolumntype{M}[1]{>{\centering\arraybackslash}m{#1}}
\begin{figure}[!tb]
\footnotesize
\renewcommand{\tabcolsep}{0.5mm}
\begin{center}
\begin{tabular}{M{1.9cm}|M{2.5cm}|M{1.7cm}|M{1.7cm}}
	\toprule
	\textbf{Join Condition} 		& \textbf{Example} & \textbf{Time $\Prep(n)$} & \textbf{Space $\Space(n)$}\\ \hline
	($C$) Theta & booleanUDF(\texttt{S.A}, \texttt{T.C}) & $\O(n^2)$ & $\O(n^2)$\\
	\hline
	($C1$) Inequality  		& $\texttt{S.A} < \texttt{T.B}$ & & \\
	\cline{1-2}
	($C2$) Non-equality 	& $\texttt{S.A} \neq \texttt{T.B}$ & $\O(n \log n)$ & $\O(n \log\log n)$\\
	\cline{1-2}
	($C3$) Band 			& $|\texttt{S.A} - \texttt{T.B}| < \epsilon$ & & \\
	\hline
	($C4$) DNF of $(C1),(C2),(C3)$ 	
	    & $(\texttt{S.A} \!<\! \texttt{T.B} \wedge \texttt{S.A} \!<\! \texttt{T.C})$ $\vee (\texttt{S.A} \!\neq\! \texttt{T.D})$
		& $\O(n \polylog n)$ & $\O(n \polylog n)$\\
	\bottomrule
\end{tabular} 
\caption{Preprocessing time $\Prep(n)$ and space complexity $\Space(n)$
of our approach for various
join conditions. 
Our novel factorized representation allows
ranked enumeration to return the $k$ top-ranked results in time
(``Time-To'')
$\TT(k) = \O(\Prep(n) + k \log k)$, using $\MEM(k) = \O(\Space(n) + k)$ space.
}
\label{tab:summary}
\end{center}
\end{figure}

\section{Preliminaries}
\label{sec:preliminaries}

\subsection{Queries}
\label{sec:queries}

\noindent
Let $[m]$ denote the set of integers $\{1, \ldots, m\}$.
A \emph{theta-join} query in Datalog notation is a formula of the type
\begin{equation*}
Q(\vec Z) \datarule R_1(\vec{X}_1), \ldots, R_\ell(\vec{X}_\ell), \; \theta_1(\vec{Y}_1), \ldots, \theta_q(\vec{Y}_q)
\end{equation*}
where $R_i$ are {relational symbols},
$\vec{X}_i$ are lists of {variables} (or {attributes}),
$\vec{Z}, \vec{Y}_i$ are subsets of 
$\vec{X} = \bigcup \vec{X}_i$,
$i \in [\ell]$, $j \in [q]$, and
$\theta_j$ are Boolean formulas called \emph{join predicates}.
The terms $R_i(\vec{X}_i)$ are called the {atoms} of the query.
Equality predicates are encoded by repeat occurrences of the same variable
in different atoms; all other join predicates are encoded in the corresponding
$\theta_j$. If no predicates $\theta_j$ are present, then $Q$ is an \emph{equi-join}.
The size $|Q|$ of the query is equal to the number of symbols in the formula.

\introparagraph{Query semantics}
Join queries are evaluated over a database 
that associates with each $R_i$ a finite relation (or table) that
draws values from a domain that we assume to be $\R$ for
simplicity.\footnote{Our approach naturally extends to other domains such as
strings or vectors, as long as the corresponding join predicates are
well-defined and computable in $\O(1)$ for a pair of input tuples.}
Without loss of generality, we assume that relational symbols in different
atoms are distinct since self-joins can be handled with linear overhead
by copying a relation to a new one.
The maximum number of tuples in an input relation is denoted by $n$.
We write $R.A$ for an attribute $A$ of relation $R$
and $r.A$ 
for the value of $A$ in tuple $r \in R_i$.
The semantics of a theta-join query is to
($i$) create the Cartesian product of the
$\ell$ relations,
($ii$) select the tuples that
satisfy the
equi-join conditions and $\theta_j$ predicates, 
and ($iii$) project on the $\vec{Z}$ attributes. 
Consequently, each individual query answer
can be represented as a combination of joining input tuples,
one from each table $R_i$.

\introparagraph{Projections}
In this paper, we focus on \emph{full} queries, i.e., join queries without
projections 
($\vec{Z} = \vec{X}$).
While our approach can handle projections by applying them in the end,
the strong asymptotic $\TT(k)$ guarantees may not hold any more.
The reason is that a projection could map multiple distinct output tuples
to the same projected answer. In the strict relational model where relations
are sets, those ``duplicates'' would have to be eliminated, creating larger
gaps between consecutive answers returned to the user.
Fortunately, our strong guarantees still hold for \emph{arbitrary
projections} in the presence of bag semantics, which
is what \DBMSs use when the SQL query has a SELECT clause instead
of SELECT DISTINCT.
Even for set semantics and SELECT DISTINCT queries, it is straightforward
to extend our strong guarantees to non-full queries that are
\emph{free-connex}~\cite{bagan07constenum,baron16acyclic,Berkholz20tutorial,idris20dynamic_theta}.

\introparagraph{Join trees for equi-joins}
An equi-join query is (alpha-)\emph{acyclic} \cite{graham80gyo,yu79gyo,tarjan84acyclic}
if it admits a join tree.
A \emph{join tree} is a tree with the atoms (relations) as the nodes where for every attribute $A$ appearing in an atom, all nodes containing $A$ form a
connected subtree.
The GYO reduction \cite{yu79gyo} computes such a join tree for equi-joins.

\introparagraph{Atomic join predicates}
We define the following types of predicates between attributes $S.A$ and $T.B$:
an \emph{inequality} is 
$S.A < T.B$, $S.A > T.B$, $S.A \leq T.B$, or $S.A \geq T.B$,
a \emph{non-equality} is 
$S.A \neq T.B$ and 
a \emph{band} is 
$|S.A - T.B| < \epsilon$ for some $\epsilon > 0$.
Our approach also supports numerical expressions over input tuples,
e.g., $f(S.A_1, S.A_2, \ldots) < g(T.B_1, T.B_2, \ldots)$,
with $f$ and $g$ arbitrary $\O(1)$-time computable functions
that map to $\R$.
The join predicates $\theta_j$ are built with conjunctions and disjunctions of such atomic predicates.
We assume there are no predicates on individual relations 
since they can be removed in linear time by filtering the corresponding
input tables.

\subsection{Ranked Enumeration}
\label{sec:ranked_enum}

\emph{Ranked enumeration} \cite{tziavelis20tutorial}
returns 
distinct join answers one-at-a-time,
in the order dictated by a given ranking function on the output tuples.
Since this paradigm generalizes top-$k$
(top-$k$ for ``any $k$'' value, or ``anytime top-$k$''), 
it is also called any-$k$ \cite{yang2018any,tziavelis20vldb}.
An obvious solution is to compute the entire join output,
and then either batch-sort it or insert it into a heap data structure.
Our goal is to find more efficient solutions for appropriate ranking
functions.

For simplicity, in this paper we only discuss ranking by 
increasing
\emph{sum-of-weights},
where each input tuple has a real-valued weight and the weight of an output
tuple is the sum of the weights of the 
input tuples that were
joined to derive it. Ranked enumeration returns the join answers in
increasing order of output-tuple weight.
It is straightforward to generalize our approach to any ranking function
that can be interpreted as
a \emph{selective dioid} \cite{tziavelis20vldb}.
Intuitively, a selective dioid \cite{GondranMinoux:2008:Semirings}
is a semiring that also establishes a total
order on the domain. 
It has two operators 
($\operatorname{min}$ and $+$ for sum-of-weights) 
where one 
\emph{distributes} over the other
($+$ distributes over $\operatorname{min}$).
These structures include even less obvious cases such as lexicographic
ordering by relation attributes. 

\subsection{Complexity Measures}
\label{sec:complexity}

We consider in-memory computation and analyze all algorithms in the standard
Random Access Machine (RAM)
model with uniform cost measure. Following common practice, we treat query size
$|Q|$---intuitively, the length of the SQL string---as a constant.
This corresponds to the classic notion of
\emph{data complexity} \cite{DBLP:conf/stoc/Vardi82}, where 
one is interested
in scalability in the size of the input data, and not of the query
(because users do not write arbitrarily large queries).

In line with previous work~\cite{berkholz19submodular,GottlobGLS:2016,carmeli20random},
we assume that it is possible to create in linear time an index that supports tuple lookups in constant time. 
In practice, hashing achieves those guarantees in an expected, amortized sense.
We include all index construction times and index sizes in our analysis.

For the time complexity of enumeration algorithms, 
we measure the time until the $k^\textrm{th}$ result 
is returned ($\TT(k)$) for all values of $k$.
In the full version \cite{tziavelis21full},
we further discuss the relationship of $\TT(k)$ to enumeration delay as complexity measures.
Since we do not assume any given indexes, a trivial lower bound
is $\TT(k) = \O(n + k)$: the time to inspect each input tuple at least
once and to return $k$ output tuples.
\emph{Our algorithms achieve that lower bound up to a polylogarithmic factor}.
For space complexity, we use $\MEM(k)$ to denote the required memory
until the $k^\textrm{th}$ result is returned.

\section{Graph Framework for Joins}
\label{sec:framework}

\begin{figure*}[tb]
\centering
\begin{subfigure}[t]{.3\linewidth}
    \centering
    \includegraphics[height=5cm]{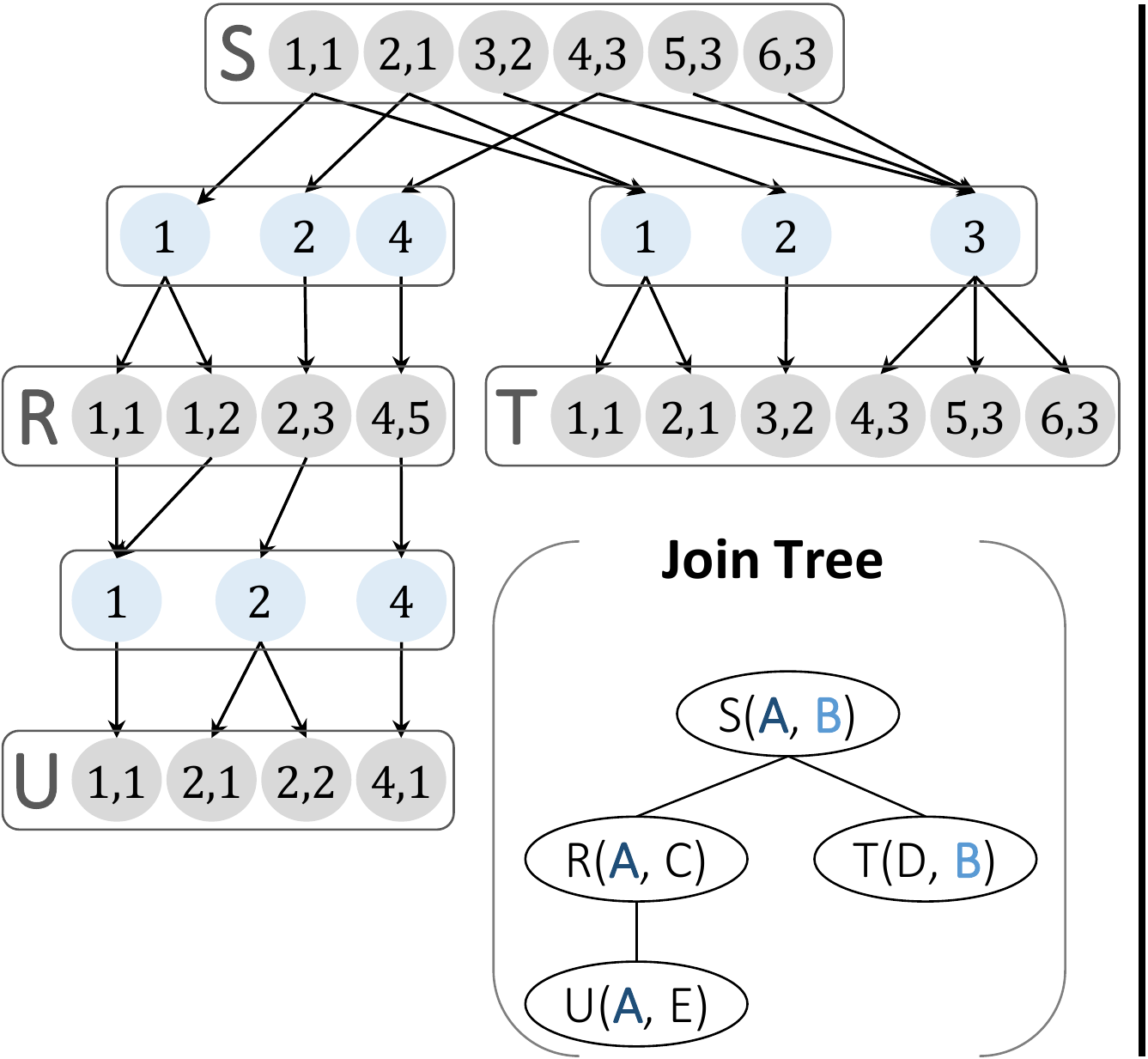}
    \caption{Equi-join \dpgraph \cite{tziavelis20vldb}.}
    \label{fig:overview_equi}
    \end{subfigure}%
\hfill
\begin{subfigure}[t]{.47\linewidth}
    \centering
    \includegraphics[height=5cm]{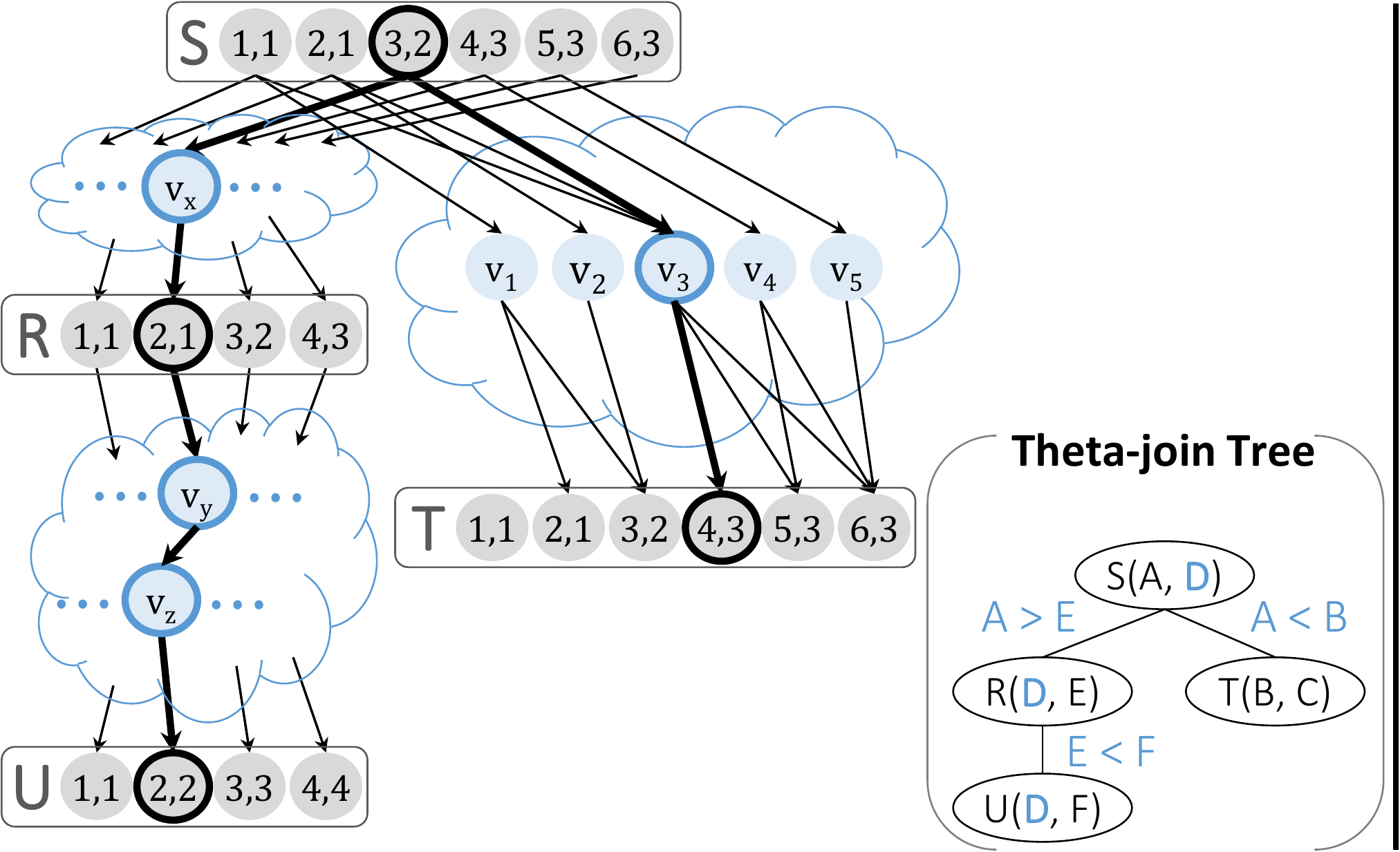}
    \caption{Theta-join \dpgraph and abstraction proposed in this paper.}
    \label{fig:overview_abstract}
\end{subfigure}
\hfill
\begin{subfigure}[t]{.22\linewidth}
    \centering
    \includegraphics[height=5cm]{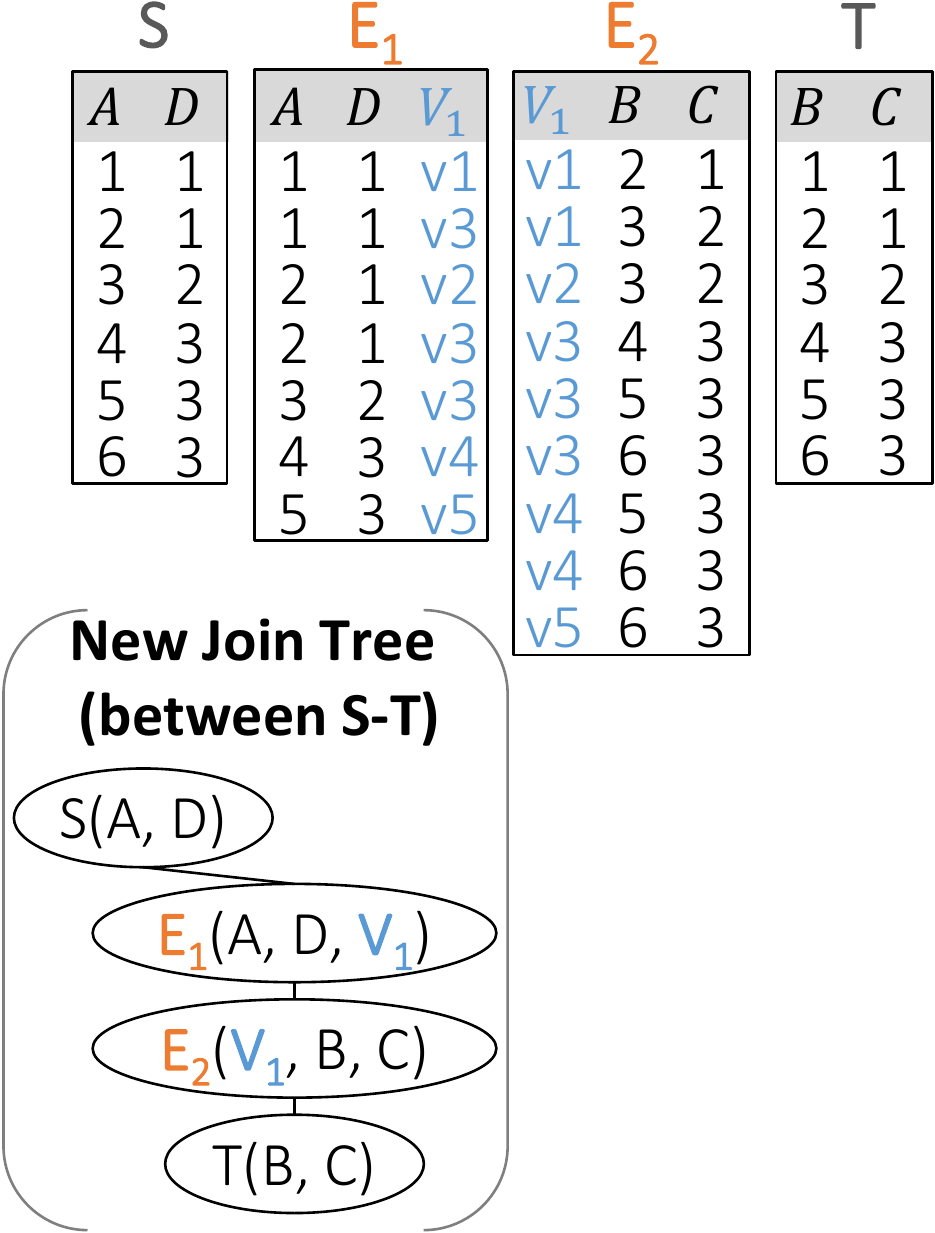}
    \caption{Reduction to equi-join.}
    \label{fig:overview_tables}
\end{subfigure}
\caption{Overview of our approach. We generalize the equi-join-specific construction to theta-joins by introducing an abstraction (blue clouds) that factorizes binary joins.
Some factorizations can also be used to reduce theta-joins to equi-joins.
}
\label{fig:overview}
\end{figure*}

We summarize our recent work on ranked enumeration for equi-joins, then
show our novel generalization to theta-joins.

\subsection{Previous Work: Any-$k$ for Equi-joins}
\label{sec:prev_work}

\emph{Any-$k$} algorithms~\cite{tziavelis20vldb} for \emph{acyclic equi-joins}
reduce ranked enumeration
to the problem of finding the $k^\textrm{th}$-lightest trees 
in a layered DAG, which we call the \emph{\dpgraph}.
Its structure depends on the join tree of the given query;
an example is depicted in \cref{fig:overview_equi}. 
The \dpgraph is a layered DAG in the sense that we associate it with a particular topological sort:
(1) Conceptually, each node is labeled with a layer ID
(not shown in the figure to avoid clutter).
A layer is a set of nodes that share the same layer ID (depicted with rounded rectangles). 
(2) Each edge is directed, going from lower to higher layer ID.
(3) All tuples from an input relation appear as (black-shaded) nodes in the same layer, called a \emph{relation layer}.
Each relation layer has a unique ID and for each join-tree edge $(S, T)$, $S$ has a lower layer ID
than $T$. 
(4) If and only if two relations are adjacent in the join tree, then
their layers are connected via a \emph{connection layer} that contains (blue-shaded) nodes representing
their join-attribute values.
(5) The edges from a relation layer to a connection layer connect the tuples with their corresponding join-attribute values and vice-versa.

The \dpgraph is constructed on-the-fly and bottom-up, according
to a join tree of the query (starting from $U$
and $T$ in the example). This phase essentially performs a bottom-up
semi-join reduction that also creates the edges and join-attribute-value
nodes. 
A \emph{tree solution} is a tree that starts from the root layer and contains exactly 1 node from each relation layer.
By construction, every tree solution corresponds
to a query answer, and vice versa.

The any-$k$ algorithm then goes through two phases on the \dpgraph.
The first is a Dynamic Programming computation, 
where every graph node records for each of its outgoing edges
the lowest weight among all subtrees that contain $1$ node from each relation layer below.
The minimum-subtree and input-tuple weights are not shown in \Cref{fig:overview_equi}
to avoid clutter. For instance, the outgoing edge for $R$-node $(2,3)$ would store
the smaller of the weights of $U$-tuples $(2,1)$ and $(2,2)$.
Similarly, the left edge from $S$-node $(2,1)$ would store the sum of the weight
of $R$-tuple $(2,3)$ and the minimum subtree weight from $R$-node $(2,3)$.
The minimum-subtree weight for a node's outgoing edge is obtained at a
constant cost by pushing the minimum weight 
over all outgoing edges
up to the node's parent.
Afterwards, enumeration is done in a second phase,
where the \dpgraph is traversed
top-down (from $S$ in the
example), with the traversal order determined by the layer IDs and
minimum-subtree weights on a node's outgoing edges.
The size of the \dpgraph and its number of layers determine space and time
complexity of the any-$k$ algorithm. The following lemma summarizes the main
result from our previous work~\cite{tziavelis20vldb}. We restate it here in terms of
data complexity (where query size $\ell$ is a constant) and using $\lambda$
for the number of layers.\footnote{Due to the specific
construction for equi-joins \cite{tziavelis20vldb}, there $\lambda$ was linear
in query size $\ell$ and hence $\ell$ and $\lambda$ were used interchangeably.
In our generalization this may not be the case, therefore we use the more precise
parameter $\lambda$ here.}

\begin{lemma}[\cite{tziavelis20vldb}]\label{lem:acyclicEquiJoinComplexity}
Given an \dpgraph with $|E|$ edges and $\lambda$ layers, 
ranked enumeration 
of the $k$-lightest tree solutions
can be performed with
$\TT(k) = \O(|E| + k \log k + k \lambda)$ and
$\MEM(k) = \O(|E| + k \lambda)$.
\end{lemma}

To extend the any-$k$ framework beyond equi-joins, we 
generalize first the definition of a join tree
and then the \dpgraph with an abstraction that is sensitive to the join conditions.

\subsection{Theta-Join Tree}
\label{sec:join_tree_graph}

The join tree 
is essential for generating the \dpgraph. 
In contrast to
equi-joins,
for general join conditions there is no
established methodology for how to define or find a join tree.
We generalize the join tree definition as follows:

\begin{definition}[Theta-join Tree]
A theta-join tree for a theta-join query $Q$ is a join tree for the equi-join $Q'$ that has all the $\theta_j$ predicates of $Q$ removed,
and every $\theta_j$ is assigned to an edge $(S, T)$ of the tree
such that $S$ and $T$ contain all the attributes referenced in $\theta_j$.
\end{definition}

We call a theta-join query \emph{acyclic} if it admits a theta-join tree.
In the theta-join tree, edge $(S,T)$ represents the join $S \bowtie_\theta T$, where join condition
$\theta$ is the conjunction of all predicates $\theta_j$ 
assigned to the edge, as well as the equality predicates
$S.A = T.A$ for every attribute $A$ that appears in both $S$ and $T$.

\begin{example}
\label{ex:query}
Consider
$Q(A, B, C, D, E, F) \datarule R(D, E), S(A, D)$, $T(B, C), U(D, F), (A < B), (A > E), (E < F)$.\footnote{SELECT * FROM R, S, T, U WHERE\\\phantom{plhldr}R.D = S.D AND R.D = U.D AND S.A < T.B AND
S.A > R.E AND R.E < U.F}
This query is acyclic since we can construct the theta-join tree shown in
\cref{fig:overview_abstract}.
Notice that all nodes containing attribute $D$ are connected
and each inequality is assigned to an edge whose adjacent nodes together
contain all referenced attributes.
For example, $A < B$ is assigned to $(S, T)$ ($S$ contains $A$ and $T$
contains $B$). The join-tree edges represent join predicates
$\theta_1 = S.A < T.B$ (edge $(S,T)$),
$\theta_2 = S.A > R.E \wedge S.D = T.D$ (edge $(S,R)$), and
$\theta_3 = R.E < U.F \wedge R.D = U.D$ (edge $(R,U)$).
\end{example}

We can construct the theta-join tree 
by first removing
all $\theta_j$ predicates from the given query $Q$, turning it into an
equi-join $Q'$. Then an algorithm like the GYO reduction can be used to find
a join tree for $Q'$. For the query in \Cref{ex:query}, this join tree looks
like the one in \Cref{fig:overview_abstract}, but without the edge labels.
Finally, we attempt to add each
$\theta_j$ predicate to a join-tree edge: $\theta_j$ can be assigned to any edge
where the two adjacent nodes contain all the attributes referenced in it.
Note that there may exist different join trees for $Q'$,
and we may have to try all possible options to obtain a theta-join tree.
Fortunately, this computation depends only on the query, 
thus takes $\O(1)$ space and time in data complexity.
If either the GYO algorithm fails to find a join tree for $Q'$ or 
no join tree allows us to assign the 
$\theta_j$ predicates to tree edges,
then the query
is \emph{cyclic} and can be handled as discussed in \Cref{sec:cycles}.
We discuss next how to create the \dpgraph for a given theta-join tree.

\subsection{Factorized Join Representation}
\label{sec:factorized}

By relying on a join tree similar in structure to the equi-join case, we can
establish a similar layered structure for the \dpgraph.
In particular, each input relation appears in a separate layer
and each join-tree edge is mapped to a subgraph implementing
the join condition between the corresponding relation layers.
This is visualized by the blue clouds in \Cref{fig:overview_abstract}.
In contrast to the equi-joins, we allow more general connection layers,
possibly a single layer with a more complex connection pattern
(like the $S$-to-$T$ connection in the example)
or even multiple layers (like the connection between
$R$-node $(2,1)$ and $U$-node $(2,2)$).

To be able to apply our any-k algorithms \cite{tziavelis20vldb} to this
generalized \dpgraph we must ensure that 
(1) each ``blue cloud'' can be mapped to a layered graph and 
(2) each tree solution corresponds to a join answer, and vice versa
(like the one highlighted in \Cref{fig:overview_abstract} which corresponds to joining input tuples $s=(3,2)$, $t=(4,3)$, $r=(2,1)$,
and $u=(2,2)$).
For (2) it is sufficient to ensure for each adjacent parent-child
pair of relations in the theta-join tree that there exists a path from a node
in the parent-relation layer to a node in the child-relation layer iff
the corresponding input tuples join. In the example, there is a path
from $S$-node $(3,2)$ via $v_3$ to $T$-node $(4,3)$, because the two tuples
satisfy $A=3 < B=4$. Similarly, since $s'=(5,3)$ and $t=(4,3)$ violate $A<B$,
there is no path from the former to the latter.
For (1), it is sufficient to ensure that the ``blue cloud'' is a DAG
with parent-relation nodes only having edges going into the cloud, while
all child-relation edges must point out of the cloud.
We formalize these properties with the notion of a
\emph{Tuple-Level Factorization Graph} (\TLFG).

\begin{definition}[\TLFG]
A \emph{Tuple-Level Factorization Graph} of a theta-join 
$S \bowtie_\theta T$ of relation $S$, called the source, and $T$, called the
target, is a directed acyclic graph $G(V, E)$ where:
\begin{enumerate}
    \item $V$ contains a distinct source node $v_s$ for each tuple $s \in S$, 
    a distinct target node $v_t$ for each tuple $t \in T$, and
    possibly other intermediate nodes,

    \item each source node $v_s$ has only outgoing edges and each target
    node $v_t$ has only incoming edges, and

    \item for each $s \in S, t \in T$, there exists a path from $v_s$ to $v_t$
    in $G$ if and only if $s$ and $t$ satisfy join condition $\theta$.
\end{enumerate}
\end{definition}
The \emph{size} of a \TLFG $G(V, E)$ is $|V| + |E|$ and 
its \emph{depth} $d$ is the maximum length of any path in $G$.
The graphs depicted in \cref{fig:Equality_all} and \cref{fig:Equality_grouping} are valid \TLFGs for equi-joins.

It is easy to see that any \TLFG is a layered graph: 
Assign w.l.o.g.\ layer ID 0 to all source nodes $v_s$;
each intermediate node $v$ is assigned layer ID $i$,
where $i$ is the length of the longest path (measured in number of edges) from
any source node to $v$. Here $i$ is well-defined due to the \TLFG's acyclicity.
All target-relation nodes are assigned to layer $d$, which is the maximum
layer ID assigned to any intermediate node, plus 1.
In the example in \Cref{fig:Inequality_sharing}, node $v_3$ is in layer 3,
because the longest path from any $S$-node to $v_3$
has 3 edges (from $(1,1)$ in the example). All $T$-nodes are in layer 6.

Since the entire generalized \dpgraph consists of
$\ell$ relation layers
and $\ell - 1$ \TLFGs (one for each edge of the
theta-join tree),
using \Cref{lem:acyclicEquiJoinComplexity} we can show:

\begin{theorem}\label{thm:genericComplexity}
Given a theta-join $Q$ of $\ell = \O(1)$ relations,
a theta-join tree,
and the corresponding \dpgraph $G_Q$,
where for each edge of the theta-join tree
the corresponding \TLFG has $\O(|E|)$ size and $\O(d)$ depth, then
ranked enumeration of
the $k$-lightest tree solutions
can be performed with
$\TT(k) = \O(|E| + k \log k + k d)$ and
$\MEM(k) = \O(|E| + k d)$.
\end{theorem}

The theorem states that worst-case size and depth of the \TLFG %
determine
the time and space complexity of enumerating 
the theta-join answers in weight order.
Hence the main challenge is to encode join condition with the smallest
and most shallow \TLFG possible.

\begin{figure}[t]
\centering
\includegraphics[width=.75\linewidth]{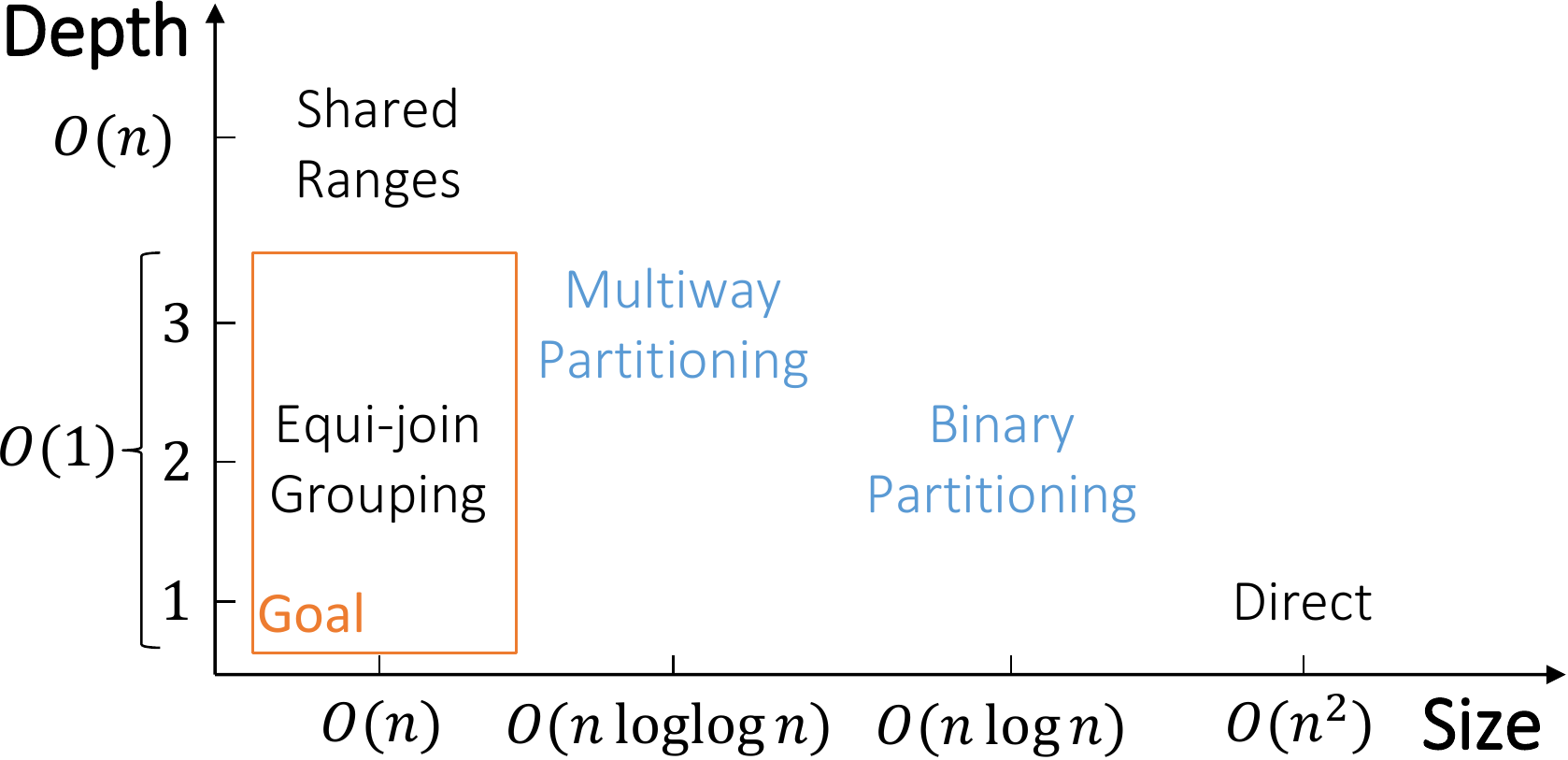}
\caption{We propose 4 different \TLFGs for a single inequality.
These trade off size with depth and 2 of them (in blue) achieve the equi-join guarantee up to a logarithmic factor.}
\label{fig:sizeTradeoff}
\end{figure}

\introparagraph{Direct \TLFGs}
For any theta-join, a naive way to construct a \TLFG is to directly connect
each source node with all the target nodes it joins with.
This results in $|E| = \O(n^2)$ and $d = 1$, thus
$\TT(k) = \O(n^2 + k \log k)$ and $\MEM(k) = \O(n^2 + k)$, respectively.
Hence even the top-ranked result requires quadratic 
time and space.
To improve this complexity, we must find a \TLFG with a smaller
number of edges, while keeping the depth low.
Our results are summarized in \Cref{fig:sizeTradeoff}, with details
discussed in later sections.

\introparagraph{Output duplicates}
A subtle issue with \Cref{thm:genericComplexity} is that two non-isomorphic tree solutions
of the \dpgraph may contain the exact same input tuples 
(the relation-layer nodes), causing duplicate query answers.
This happens if and only if
a \TLFG has multiple paths between the same source and destination node.
While one would like to avoid this, it may not be possible to find a
\TLFG that is both efficient in terms of size and depth, 
and also free of duplicate paths. Among the inequality conditions studied in
this paper, this only happens for disjunctions (\cref{sec:disjunctions}).

Since duplicate join answers must be removed, the time to return the $k$
top-ranked answers may increase. Fortunately, for our disjunction construction
it is easy to show that the number of duplicates per output tuple is $\O(1)$,
i.e., it does not depend on input size $n$. This implies that we can
filter the duplicates on-the-fly without increasing the complexity of
$\TT(k)$ (or $\MEM(k)$, for that matter): We maintain the top-$k$ join answers returned
so far in a lookup structure and, before outputting the next join answer,
we check in $\O(1)$ time if the same output had been returned before.\footnote{As 
an optimization, 
we can clear this lookup structure
whenever the weight of an answer is greater than the previous,
since all duplicates share the same weight. While this does not impact 
worst-case complexity, it can greatly reduce computation cost in practice
whenever output tuples have diverse sum-of-weight values.}

To prove that the number of duplicates per join answer is independent of input size,
it is sufficient to show that for each \TLFG the maximum number of paths from
any source node $v_s$ to any target node $v_t$, which we will call
the \emph{duplication factor}, is independent of input size.
We show this to be the case for the only \TLFG construction that could
introduce duplicate paths: disjunctions (\cref{sec:disjunctions}).
A duplicate-free \TLFG has a duplication factor equal to 1 (which is the case for most \TLFGs we discuss).

\subsection{Theta-join to Equi-join Reduction}
\label{sec:to_equi}

The factorized representation of the output of a theta-join as an \dpgraph
(using \TLFGs to connect adjacent relation layers) 
enables a novel reduction
from complex theta-joins to equi-joins.

\begin{theorem}\label{thm:thetaToEqui}
Let $G = (V,E)$ be a \TLFG of depth $d$ for a theta-join $S \bowtie_\theta T$
of relations $S$, $T$
and $X$ be the union of their attributes. 
For $0 < i \le d$,
let $E_i$ be the set of edges from layer $i-1$ to $i$. If $E = \bigcup_i E_i$,
i.e., every edge connects nodes in adjacent layers, then
$S \bowtie_\theta T = \pi_{X}( S \bowtie E_1 \bowtie\cdots\bowtie E_d \bowtie T$) where $\pi_X$ is an $X$-projection.
\end{theorem}

Intuitively, the theorem states that if no edge in the \TLFG skips a layer,
then the theta-join $S \bowtie_\theta T$ can equivalently be computed as an
equi-join between $S$, $T$, and $d$ \emph{auxiliary} relations.
Each of those relations is the set of edges between adjacent layers of the \TLFG.

The theorem is easy to prove by construction, which we explain using the example
in \Cref{fig:overview_abstract}.
Consider the \TLFG for $S$ and $T$ and notice that all edges are between adjacent layers
and $d = 2$. In \Cref{fig:overview_tables}, the first
tuple $(1, 1, v_1) \in E_1$ represents the edge
from $S$-node $(1,1)$ to intermediate node $v_1$. (The tuple is obtained as the
Cartesian product of the edge's endpoints.) Similarly, the first tuple in $E_2$
represents the edge from $v_1$ to $T$-node $(2,1)$. It is easy to verify that
$S(A, D) \bowtie_{A<B} T(B, C) = \pi_{ADBC} (S \bowtie E_1 \bowtie E_2 \bowtie T)$.
The corresponding branch of the join tree is shown in \Cref{fig:overview_tables}.
Compared to the theta-join tree in \Cref{fig:overview_abstract}, the inequality
condition disappeared from the edge and is replaced by new nodes
$E_1(A, D, V_1)$ and $E_2(V_1, B, C)$.

\introparagraph{\QUADEQUI for direct \TLFGs}
Recall that any theta-join $S \bowtie_\theta T$ between relations of size $\O(n)$
can be represented by a 1-layer
\TLFG that directly connects the joining $S$- and $T$-nodes. Since this \TLFG
satisfies the condition of \Cref{thm:thetaToEqui}, it can be reduced
to equi-join $S \bowtie E \bowtie T$, where $|E| = \O(n^2)$.
We refer to the algorithm that first applies this construction to each
edge of the theta-join tree (and thus reducing the entire theta-join query between
$\ell$ relations to an equi-join) and then uses the equi-join ranked-enumeration
algorithm \cite{tziavelis20vldb} as \QUADEQUI.

Below we will show that better constructions with smaller auxiliary relations $E_i$
can be found for any join condition that is a DNF of inequalities. 
In particular,
such joins can be expressed as $S \bowtie E_1 \bowtie E_2 \bowtie T$ where $E_1, E_2$ are of size $\O(n \polylog n)$. 
\Cref{fig:overview_tables} shows a concrete instance.
However, note that not all \TLFGs satisfy the condition of \cref{thm:thetaToEqui}.
For example, \cref{fig:Inequality_sharing} shows a \TLFG which cannot be reduced to an equi-join with our theorem.

\begin{figure*}[t]
\centering
\begin{subfigure}[t]{.17\linewidth}
    \centering
    \includegraphics[height=4.2cm]{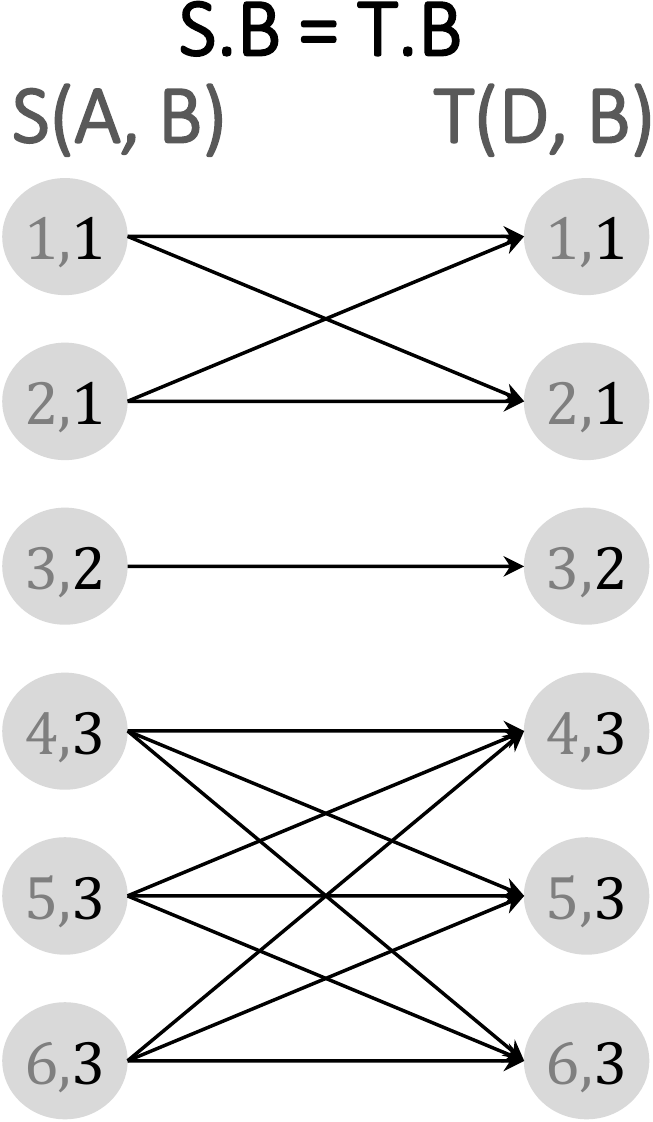}
    \caption{Equality: naive construction with edges between all joining pairs. $\O(n^2)$ size, $\O(1)$ depth.}
    \label{fig:Equality_all}
\end{subfigure}%
\hfill
\begin{subfigure}[t]{.17\linewidth}
    \centering
    \includegraphics[height=4.2cm]{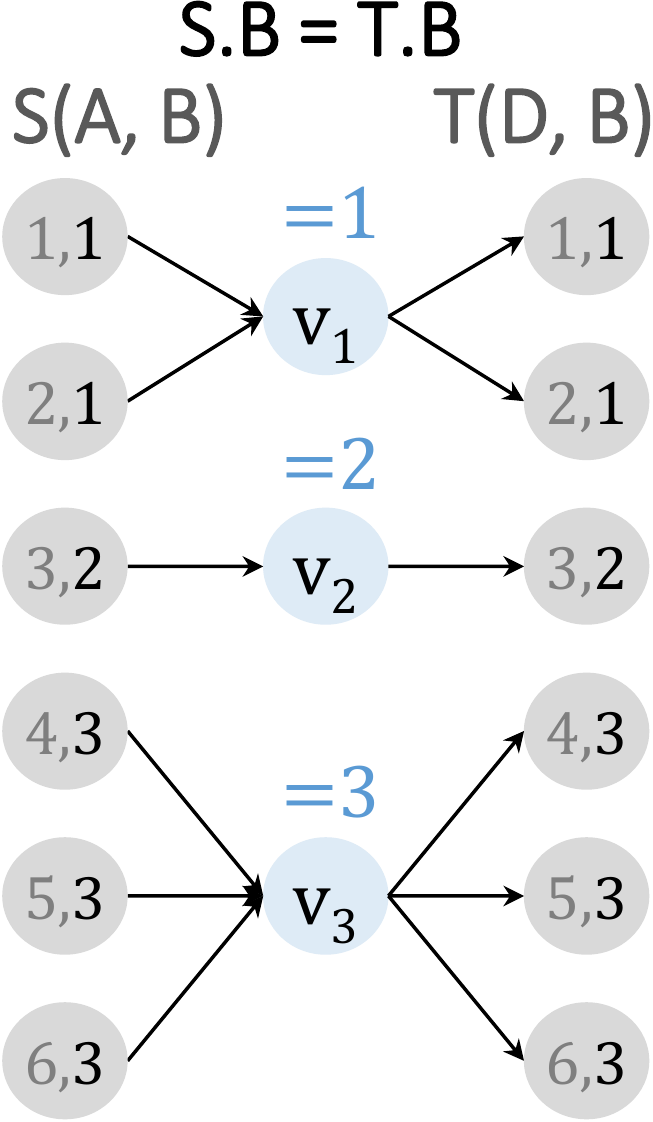}
    \caption{Equality: grouping tuples with common join values together. $\O(n)$ size, $\O(1)$ depth.}
    \label{fig:Equality_grouping}
\end{subfigure}%
\hfill
\begin{subfigure}[t]{.17\linewidth}
    \centering
    \includegraphics[height=4.2cm]{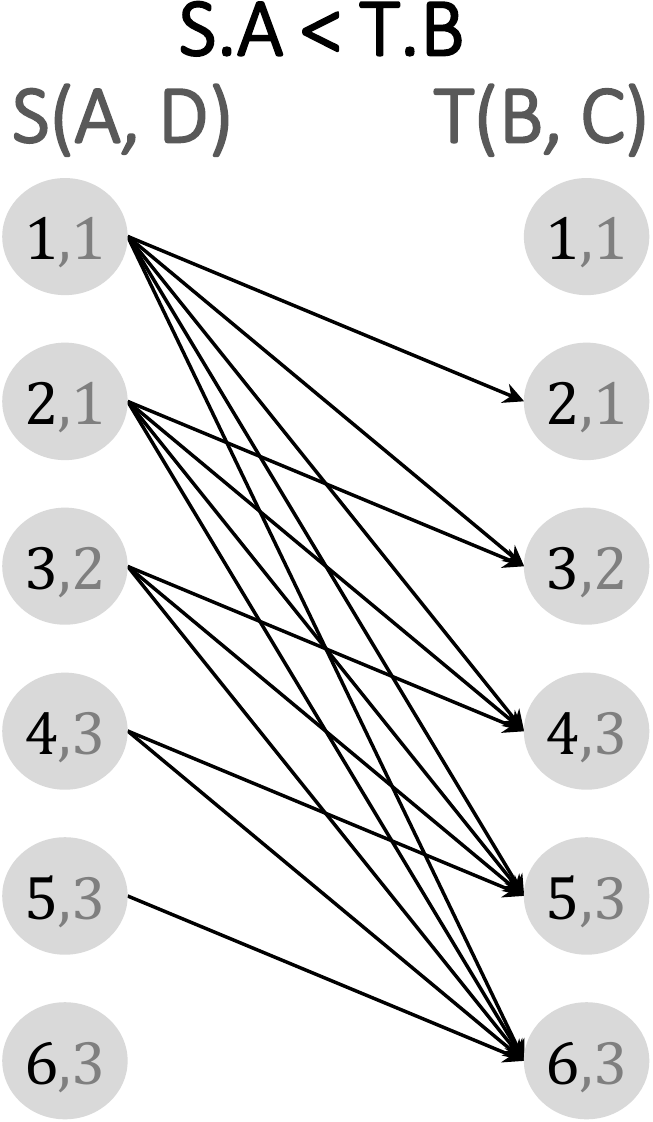}
    \caption{Inequality: naive construction with edges between all joining pairs. $\O(n^2)$ size, $\O(1)$ depth.}
    \label{fig:Inequality_all}
\end{subfigure}%
\hfill
\begin{subfigure}[t]{.17\linewidth}
    \centering
    \includegraphics[height=4.2cm]{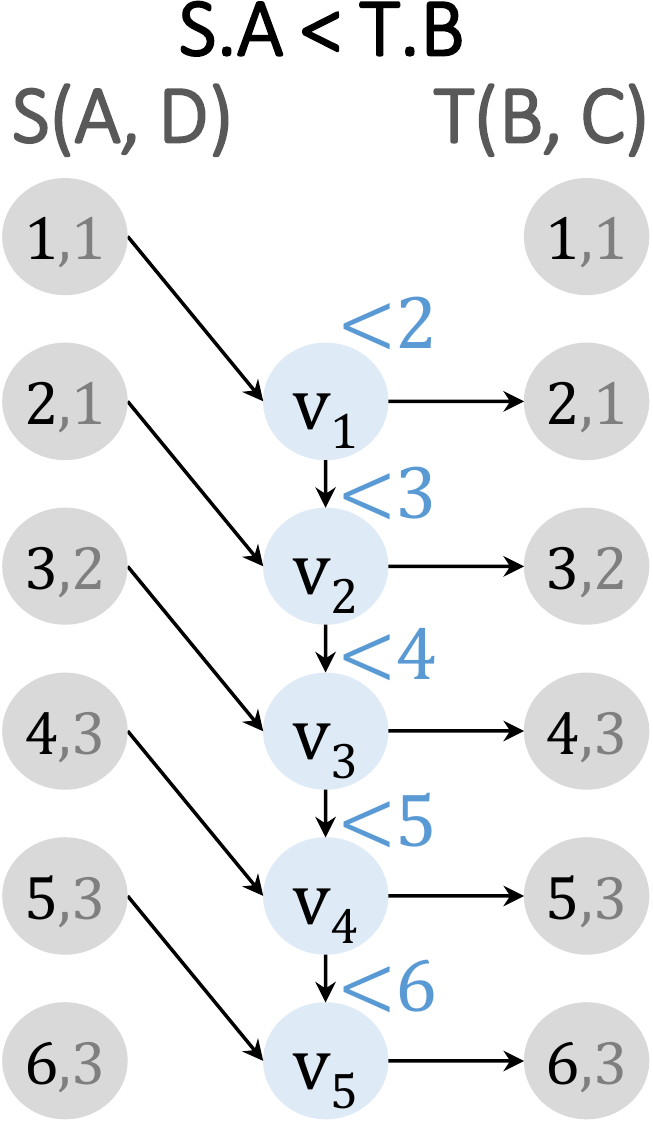}
    \caption{Inequality: shared ranges. Middle nodes indicate a range. $\O(n)$ size, $\O(n)$ depth.}
    \label{fig:Inequality_sharing}
\end{subfigure}%
\hfill
\begin{subfigure}[t]{.2\linewidth}
    \centering
    \includegraphics[height=4.2cm]{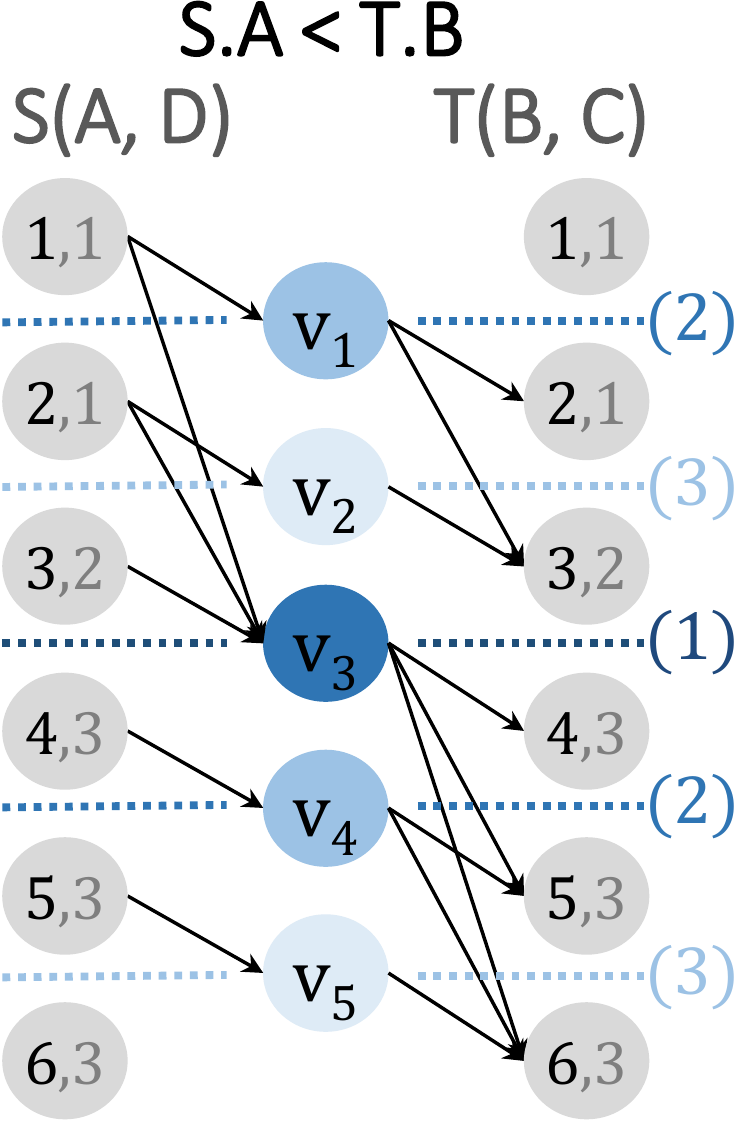}
    \caption{Inequality: binary partitioning. Dotted lines indicate partitioning steps. $\O(n \log n)$ size, $\O(1)$ depth.}
    \label{fig:Inequality_partitioning}
\end{subfigure}%
\caption{Factorization of Equality and Inequality conditions with our \TLFGs. 
The S and T node labels indicate the values of the joining attributes.
All \TLFGs shown here have $\O(1)$ depth.}
\label{fig:TLFGs}
\end{figure*}

\section{Factorization of Inequalities}
\label{sec:inequalities}

We now show how to construct \TLFGs of size
$\O(n \polylog n)$ and depth $\O(1)$ when the
join condition $\theta$ in a join $S \bowtie_\theta T$
is a DNF\footnote{Converting
an arbitrary formula to DNF may increase query size exponentially. This
does not affect data complexity, because query size is still a constant.} of inequalities (and equalities).
Starting with a single inequality, we then generalize
to conjunctions and finally to DNF.
Non-equalities and bands will be discussed in \Cref{sec:improvements}.

\subsection{Single Inequality Condition}
\label{sec:inequality}

Efficient \TLFGs for equi-joins exploit that equality conditions
group input tuples into \emph{disjoint} equivalence classes (\cref{fig:Equality_grouping}).
For inequalities, this is generally not possible and therefore
we need a different approach to leverage their
structural properties (see \cref{fig:Inequality_all}).

\introparagraph{Binary partitioning}
Our \emph{binary-partitioning} based \TLFG is inspired by
quicksort \cite{hoare62quicksort}.
Consider condition $S.A < T.B$ and a \emph{pivot} value $v$.
We partition relations $S$ and $T$ s.t.\
$s.A < v$ for $s \in S_1$ and $s.A \geq v$ for $s \in S_2$,
and similarly $t.B < v$ for $t \in T_1$ and $t.B \geq v$ for $t \in T_2$.
This guarantees that all $A$-values in $S_1$ are strictly less than
all $B$-values in $T_2$. Instead of representing this with
$|S_1|\cdot |T_2|$ direct edges $(s_i \in S_1,t_j \in T_2)$,
we introduce an intermediate ``pivot node'' $v$ and use only
$|S_1|+|T_2|$
edges $(s_i \in S_1, v)$ and $(v, t_j \in T_2)$.

Then we continue \emph{recursively} with the remaining partition
pairs $(S_1, T_1)$ and $(S_2, T_2)$.
(Note that $(S_2, T_1)$ cannot contain joining tuples by construction.)
Each recursive step will create a new intermediate node
connecting a set of source and target nodes,
therefore the \TLFG has depth $2$.

As the pivot, we use the median of the \emph{distinct} join-attribute
values appearing in the tuples in both input partitions.
E.g., for multiset $\{1, 1, 1, 1, 2, 3, 3\}$ the set of distinct values
is $\{1, 2, 3\}$ and hence the median is 2.
This pivot is easy to find in $\O(n)$ time if the relations have
been sorted on the join attributes beforehand.
Since each partition step cuts the number of distinct values per
partition in half, it takes $\O(\log n)$ steps until we reach the
base case where all input tuples in a partition share the same
join-attribute value and the recursion terminates.
Overall, the algorithm takes time $\O(n \log n)$ to construct
a \TLFG of size $\O(n \log n)$ and depth $2$.
It is easy to see that there is exactly one path from
each source to joining target node, hence the \TLFG is
\emph{duplicate-free}.

\begin{example}
\Cref{fig:Inequality_partitioning} illustrates the approach, 
with dotted lines showing how the relations are partitioned.
Initially, we create partitions containing the values $\{1, 2, 3\}$ and $\{4, 5, 6\}$ respectively.
The source nodes containing $A$-values of the first partition are connected to target nodes containing $B$-values of the second partition via the intermediate node $v_3$.
The first partition is then recursively split into $\{1\}$ and $\{2, 3\}$.
Even though these new partitions are uneven with $2$ and $4$ nodes respectively, 
they contain roughly the same number of distinct values (plus or minus one).
\end{example}

\introparagraph{Other inequality types}
The construction for greater-than ($>$) is symmetric, connecting
$S_2$ to $T_1$ instead of $S_1$ to $T_2$.
For $\leq$ and $\geq$, we only need to modify handling of the base
case of the recursion: instead of simply returning from the last
call (when all tuples in a partition have the same join-attribute
value), the algorithm connects the corresponding source and target
nodes via an intermediate node (like for equality predicates). 

\begin{lemma}
\label{lem:inequality_binary}
Let $\theta$ be an inequality predicate
for relations $S, T$ of total size $n$.
A duplicate-free \TLFG of $S \bowtie_\theta T$
of size $\O(n \log n)$ and depth $2$ can be constructed in $\O(n \log n)$ time.
\end{lemma}

\subsection{Conjunctions}
\label{sec:conjunctions}

\TLFG construction for conjunctions can be integrated elegantly
into the recursive binary partitioning.

\begin{example}
\label{ex:conjunction}
Consider join condition $S.A < T.C \wedge S.B > T.D$
for relations $S(A, B), T(C, D)$ as shown in \cref{fig:Conjunction_1}.
The algorithm initially considers the first inequality $S.A < T.C$,
splitting the relations into $S_1$, $T_1$, $S_2$, $T_2$
as per the binary partitioning method (see \cref{sec:inequality}).
All pairs $(s_i \in S_1, t_j \in T_2)$ satisfy
$S.A < T.C$, but not all of them satisfy the other conjunct
$S.B > T.D$. To correctly connect the source and target nodes, we
therefore run the same binary partitioning algorithm on input partitions
$S_1$ and $T_2$, but now with predicate $S.B > T.D$ as illustrated by
the diagonal blue edge in \cref{fig:Conjunction_1};
the resulting graph structure
is shown in \cref{fig:Conjunction_2}.
For the remaining partition pairs $(S_1, T_1)$ and $(S_2, T_2)$,
the recursive call still needs to enforce both conjuncts as illustrated
by the orange edges in \cref{fig:Conjunction_1}.
\end{example}

\begin{figure}[t]
\centering
\begin{subfigure}[t]{.51\linewidth}
    \centering
    \includegraphics[height=4cm]{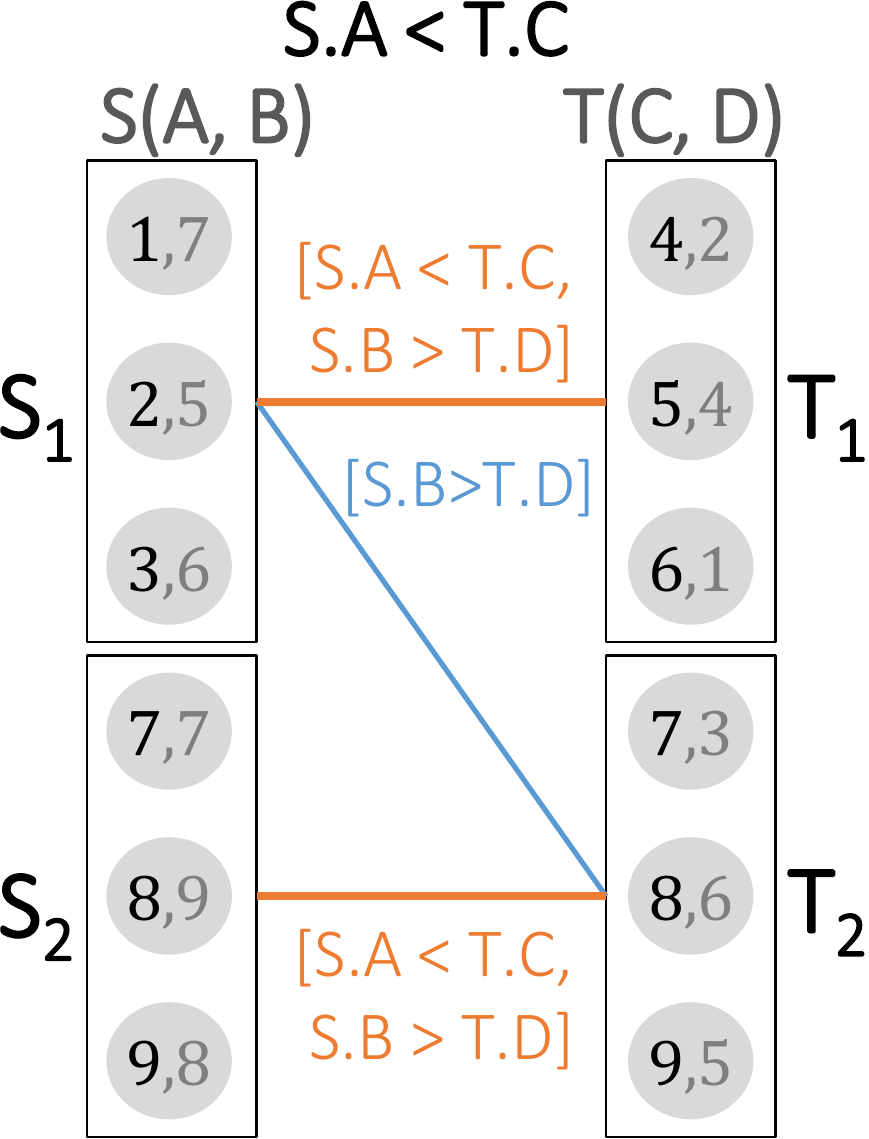}
    \caption{Binary partitioning and recursions.}
    \label{fig:Conjunction_1}
\end{subfigure}%
\hfill
\begin{subfigure}[t]{.42\linewidth}
    \centering
    \includegraphics[height=4cm]{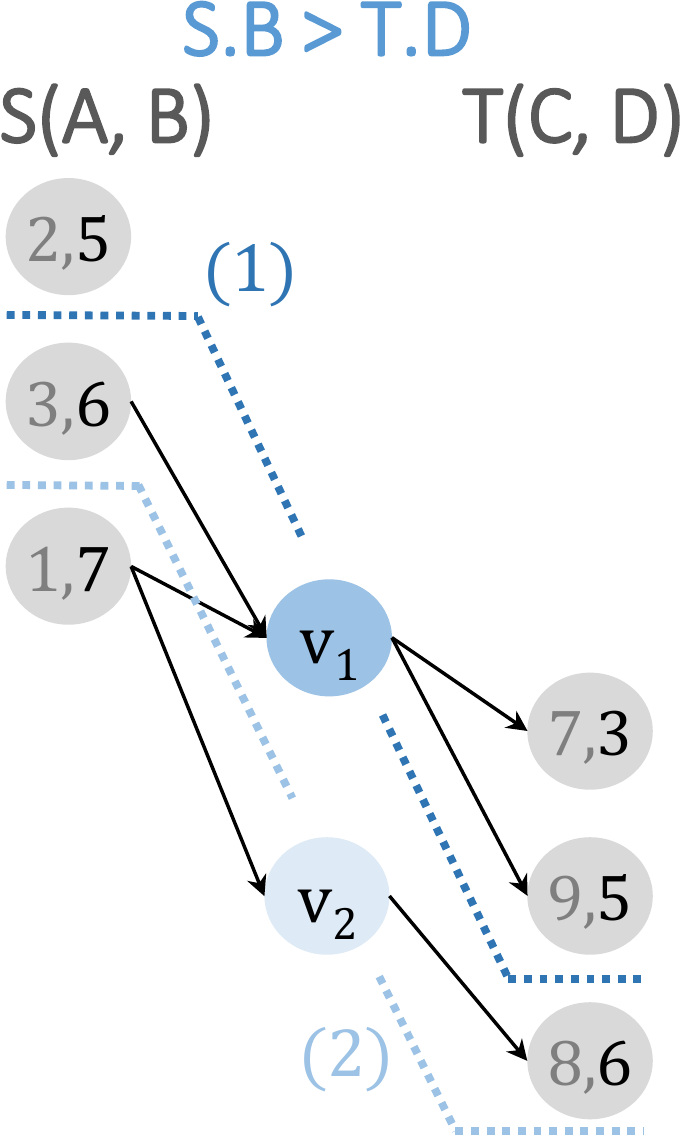}
    \caption{Handling the next predicate.}
    \label{fig:Conjunction_2}
\end{subfigure}
\caption{\Cref{ex:conjunction}: Steps of the conjunction algorithm for two inequality predicates on $S(A, B), T(C, D)$. 
Node labels depict $A, B$ values (left) or $C, D$ values (right).}
\label{fig:Conjunction}
\end{figure}

\introparagraph{Strict inequalities}
The example generalizes in a straightforward way to the conjunction of
any number of strict inequalities as shown in \cref{alg:conj}.
We note that the order in which the predicates are handled does not impact the asymptotic analysis,
but in practice, handling the most selective predicates first is bound to give better performance.
Whenever two partitions are guaranteed to satisfy a conjunct, that
conjunct is removed from consideration in the next recursive call
(\cref{alg_line_conj:next}).
An intermediate node for the pivot and the corresponding edges connecting
it to source and target nodes are only added to the \TLFG when
no predicates remain
(\crefrange{alg_line_conj:conn1}{alg_line_conj:conn3}).
Overall, we perform two recursions simultaneously.
In one direction, we make recursive calls on smaller partitions of the data
and the same set of predicates
 (\cref{alg_line_conj:recursion_horizontal_1,alg_line_conj:recursion_horizontal_2}).
In the other direction, when the current predicate is satisfied
for a partition pair, $\mathrm{\NextCond}()$ is called with one less
predicate (\cref{alg_line_conj:next}).
The recursion stops either when we are left with $1$ join value
(base case for binary partitioning)
or we exhaust the predicate list
(base case for conjunction).
Finally, notice that each time a new predicate is processed by a
recursive call, the join-attribute values in the corresponding partitions
are sorted according to the new attributes (\cref{alg_line_conj:sort})
to find the pivot.

\begin{algorithm}[tb]
\setstretch{0.85}   %
\SetAlgoLined
\LinesNumbered
\SetKwProg{myproc}{Procedure}{}{}
\SetKwFunction{nextc}{\NextCond}
\SetKwFunction{ineq}{\IneqBinaryFun}
\SetKwFunction{band}{\BandFun}
\textbf{Input}: Relations $S, T$, nodes $v_s, v_t$ for $s \in S, t \in T$, \\\phantom{Inputt: }Conjunction $\theta = \bigwedge_{i=1}^{p} \theta_i$, where $\theta_1 = S.A < T.B$\\
\textbf{Output}: A \TLFG of the join $S \bowtie_\theta T$\;
Call \nextc($S, T, \theta$)\;

\myproc{\nextc{$S, T, (S.A < T.B) \land \bigwedge_{i=2}^{p} \theta_i$}}
{
    $S', T'$ = $S, T$ sorted by attributes $A$ and $B$, respectively \label{alg_line_conj:sort}\;
    \ineq($S'$, $T'$, $(S.A < T.B) \land \bigwedge_{i=2}^{p} \theta_i$)\;
}

\SetKwFunction{distinct}{vals}
\myproc{\ineq{$S, T, (S.A < T.B) \land \bigwedge_{i=2}^{p} \theta_i$}}{
$\delta =$ \distinct{S.A $\cup$ T.B} \algocomment{  Number of distinct A, B values}\;
\lIf{$\delta == 1$}{
\KwRet \algocomment{  Base case for binary partitioning} \label{alg_line_conj:singleValue}
}
Partition $S, T$ into $(S_1, S_2), (T_1, T_2)$ with median
distinct value as pivot\;

\If{$p==1$}
{
            \algocomment{  Base case for \#predicates: connect $S_1$ to $T_2$}\;
            Materialize intermediate node $x$ \label{alg_line_conj:conn1}\;
            \lForEach{$s$ in $S_1$}{
                    Create edge $v_s \longrightarrow x$ \label{alg_line_conj:conn2}
            }
            \lForEach{$t$ in $T_2$}{
                    Create edge $x \longrightarrow v_t$ \label{alg_line_conj:conn3}
            }  
}
\Else
{
    \algocomment{  Check $S_1 \rightarrow T_2$ against the rest of the predicates}\;
    \nextc($S_1$, $T_2$, $\bigwedge_{i=2}^{p} \theta_i$) \label{alg_line_conj:next}\;
}
\algocomment{  Recursive calls on horizontal partitions, same predicates}\;
\ineq($S_1$, $T_1$, $(S.A < T.B) \land \bigwedge_{i=2}^{p} \theta_i$) \label{alg_line_conj:recursion_horizontal_1}\;
\ineq($S_2$, $T_2$, $(S.A < T.B) \land \bigwedge_{i=2}^{p} \theta_i$) \label{alg_line_conj:recursion_horizontal_2}\;
}
\caption{Factorizing a conjunction of $p$ strict inequalities}
\label{alg:conj}
\end{algorithm}

\introparagraph{Non-strict inequalities}
Like for a single predicate, we only need to modify handling of the base
case when all join-attribute values in a partition are the same.
While a strict inequality is not satisfied and thus no edges are added
to the \TLFG, the non-strict one is satisfied for all pairs of source and
target nodes in the partition.
Hence instead of exiting the recursive call (\cref{alg_line_conj:singleValue}),
the partition pair is treated like the $(S_1, T_2)$ case
(\crefrange{alg_line_conj:conn1}{alg_line_conj:next}).

\introparagraph{Equalities}
If the conjunction contains both equality and inequality predicates,
then we reduce the problem to an inequality-only conjunction by first
partitioning the inputs into equivalence classes according
to all equality predicates (see \cref{fig:Equality_grouping}).
Then the inequality-only algorithm introduced above is executed on each
of these partitions. Since the equality-based partitioning takes linear time
and space, complexity is determined by the inequality predicates.

\begin{lemma}
\label{lem:conjunction}
Let $\theta$ be a conjunction of $p$ inequality and any number of 
equality predicates
for relations $S, T$ of total size $n$.
A duplicate-free \TLFG of $S \bowtie_\theta T$
of size $\O(n \log^p n)$ and depth $2$ can be constructed in $\O(n \log^p n)$ time.
\end{lemma}

\subsection{Disjunctions}
\label{sec:disjunctions}

Given a join condition that can be expressed as a disjunction
$P = \bigvee_{i} P_i$ where $G_i$ is the \TLFG for $P_i$, we construct the
\TLFG $G$ for $P$ by simply ``unioning'' the $G_i$, i.e., $G$'s set of nodes
and edges are the unions of the node and edge sets of the $G_i$, respectively.
Note that the auxiliary ``pivot'' nodes added by the binary partitioning algorithm
to the $G_i$ are all distinct. Hence if there is a path from source $s$ to target
$t$ in $j$ of the individual $G_i$, then there are exactly $j$ different
paths from $s$ to
$t$ in $G$. This creates duplicate join results when traversing $G$ during the
enumeration phase. Fortunately, since the number of ``duplicate'' paths depends
only on the number of terms in $P$ and hence query size (not input size),
the number of duplicates per join output tuple is constant.

\begin{lemma}
\label{lem:disjunctions}
Let $\theta$ be a disjunction of predicates 
$\theta_1, \ldots, \theta_p$
for relations $S, T$.
If for each $\theta_i, i \in [p]$ we can construct a 
duplicate-free \TLFG of $S \bowtie_{\theta_i} T$
of size $\O(\mathcal{S}_i)$ and depth $\depth_i$ in $\O(\mathcal{T}_i)$ time,
then we can construct a \TLFG of $S \bowtie_{\theta} T$
of size $\O(\sum_i \mathcal{S}_i)$ and depth $\max_i \depth_i$ in $\O(\sum_i \O(\mathcal{T}_i))$ time.
The duplication factor of the latter is at most $p$.
\end{lemma}

We can now factorize any DNF of equality and inequality predicates by applying
the conjunction construction to each conjunct, and then the union construction
for their disjunction.

\section{Improvements and Extensions}
\label{sec:improvements}

We propose improvements that lead to our main result: strong worst-case
guarantees for $\TT(k)$ and $\MEM(k)$ for acyclic join queries with
inequalities, which we then extend to cyclic joins.

\subsection{Improved Factorization Methods}
\label{sec:fact_methods}

We explore how to reduce the size of the \TLFG for inequalities. 

\introparagraph{Multiway partitioning}
When the join predicate on an edge of the theta-join tree is a simple
inequality like $S.A < T.B$, we can split the set of input tuples
into $\O(\sqrt{n})$ partitions
per step---instead of 2 partitions for binary partitioning (\Cref{sec:inequality})---hence the name \emph{multiway partitioning}.
This results in a smaller \TLFG
of size  $\O(n \log\log n)$ (vs. $\O(n \log n)$ for binary partitioning)
and depth $3$ (vs. $2$). 
Unfortunately, it is unclear how to
generalize this idea to a conjunction of inequalities.

\introparagraph{Shared ranges}
A simple inequality can be encoded even more compactly with $\O(n)$ edges
by exploiting the transitivity of ``$<$'' as illustrated
in \Cref{fig:Inequality_sharing}. Intuitively, our \emph{shared ranges}
method creates a hierarchy of intermediate nodes, each one representing
a range of values. Each range is entirely contained in all those that
are higher in the hierarchy, thus we connect the intermediate nodes in a chain.
The resulting \TLFG has size and depth $\O(n)$.
The latter causes a high delay between consecutive join answers.
From \Cref{thm:genericComplexity} and the fact that we need to sort to construct the \TLFG, we obtain
$\TT(k) = \O(n \log n + n + k \log k + k n) = \O(n \log n + k n)$ and
$\MEM(k) = \O(n + k n) = \O(k n)$. Compared to binary partitioning's
$\O(n \log n + k \log k)$ and $\O(n \log n + k)$
(\cref{thm:genericComplexity}, \cref{lem:inequality_binary}),
respectively, space complexity is reduced by about a factor $\log n$,
and without affecting time
complexity, only for small $k$, i.e., $k = \smallO(\log n)$.
For larger $k = \Omega(n)$ both space and time complexity are worse
by (almost) a factor $n$. (Recall that
$k = \O(n^\ell)$ for a join of $\ell$ relations.)
Moreover, like for multiway partitioning, it is not clear how to
generalize this construction to conjunctions of inequalities.

\introparagraph{Non-Equality and Band Predicates}
A non-equality predicate can be expressed as a disjunction of 2 inequalities;
a band predicate as a conjunction of 2 inequalities. Hence both can be handled
by the techniques discussed in \Cref{sec:inequalities}, at the cost of
increasing query size by up to a constant factor. 
This can be avoided by a specialized construction that leverages the structure
of these predicates. It is similar to the binary partitioning for an inequality
(and hence omitted due to space constraints)
and achieves the same size and depth guarantees for the \TLFG.

\subsection{Putting Everything Together}

Using multiway partitioning and the specialized techniques for non-equality and
band predicates yields:

\begin{lemma}
\label{lem:tlfg_improved}
Let $\theta$ be a simple inequality, non-equality, or band predicate
for relations $S, T$ of size $\O(n)$.
A duplicate-free \TLFG for $S \bowtie_\theta T$
of size $\O(n \log\log n)$ and depth $3$ can be constructed in $\O(n \log n)$ time.
\end{lemma}

Applying the approach for a DNF of inequalities (\Cref{sec:inequalities}),
but using the specialized \TLFGs for non-equality and band predicates
and multiway partitioning for the base case of the conjunction
construction (when only one predicate remains), we obtain:

\begin{theorem}[Main Result]
\label{th:enumeration_improved}
Let $Q$ be a full acyclic theta-join query over a database $D$ of size $n$
where all the join conditions are DNF formulas of
equality, inequality, non-equality, and band predicates.
Let $p$ be the maximum number of predicates, excluding equalities,
in a conjunction of a DNF on any edge of the theta-join tree.
Ranked enumeration of the answers to $Q$ over $D$ can
be performed with $\TT(k) = \O(n \log^p n + k \log k)$.
The space requirement
is $\MEM(k) = \O(n \log^{p-1} n \cdot \log\log n+ k)$.
\end{theorem}

\subsection{Cyclic Queries}
\label{sec:cycles}

So far, we have focused only on acyclic queries,
but our techniques are also applicable to cyclic queries with some modifications.
Recall that acyclic queries admit a theta-join tree,
which is found by assigning predicates to the edges of a join tree.
If this procedure fails, we can handle the query as follows:

\introparagraph{Post-processing filter}
An common practical solution for cyclic queries is to ignore some predicates
during join processing, then apply them as a filter on the output.
Specifically, we can remove $\theta_j$ predicates and equality conditions
encoded by the same variable names until the query admits a theta-join tree,
then apply our technique to the resulting acyclic query,
and finally use the removed predicates as a filter.
While this approach is simple to implement,
it can suffer from large intermediate results.
In the worst case, all answers to the acyclic join except the last one may be discarded,
giving us $\TT(k) = \O(n^\ell \log n)$ for an $\ell$-relation cyclic join.

\introparagraph{Transformation to equi-join}
An alternative approach with non-trivial guarantees is to apply our equi-join transformation to the cyclic query,
and then use existing algorithms for ranked enumeration of cyclic equi-joins \cite{tziavelis20vldb}.
We deal with the case where each $\theta_j$ predicate is covered by at most 2 input relations; the general case is left for future work.
To handle that case, we add edges to the join tree as needed (creating a cyclic \emph{theta-join graph})
and assign predicates to covering edges.
To achieve the equi-join transformation, 
we consider all pairs of connected relations in the join graph,
build a \TLFG according to the join condition,
and then materialize 
relations ``in the middle'' as illustrated in \cref{sec:to_equi}.
The resulting query contains only equality predicates, hence is a cyclic equi-join.
Ranked enumeration for cyclic equi-joins is possible with guarantees that depend on a width measure of the query \cite{tziavelis20vldb}.

\begin{example}[Inequality Cycle]
The following triangle query variant
joins three relations with inequalities in a cyclic way:
$Q(A, B, C, D, E, F)$ $\datarule R(A, B), S(C, D), T(E, F), (B < C), (D < E), (F < A)$.
Notice that there is no way to organize the relations in a tree with the inequalities over parent-child pairs. 
However, if we remove the last inequality $(F < A)$,
the query becomes acyclic and a generalized join tree can be constructed.
Thus, we can apply our techniques on that query and filter the answers with the selection condition $(F < A)$.

Alternatively, we can factorize the pairs of relations using our \TLFGs, 
to obtain a cyclic equi-join.
If we use binary partitioning, 
this introduces three new attributes $V_1, V_2, V_3$
and six new $\O(n \log n)$-size relations:
$E_1(A, B, V_1)$, $E_2(V_1, C, D)$, $E_3(C, D, V_2)$, $E_4(V_2, E, F)$, $E_5(E, F, V_3)$, $E_6(V_3, A, B)$.
The transformed query can be shown to have a submodular width
\cite{Marx:2013:THP:2555516.2535926,khamis17panda}
of $5 / 3$,
making ranked enumeration possible with
$\TT(k) = \O((n \log n)^{5/3} + k \log k)$.
\end{example}

\section{Experiments}
\label{sec:exp}

\begin{figure*}[t]
\centering
\footnotesize
\renewcommand{\tabcolsep}{1.3pt}
\begin{tabular}{|l|l|}
\hline
  \textbf{Query} & \textbf{Ranking} \\
 \hline
 $Q_{S1}(\ldots) \datarule S_1(A_1, A_2), S_2(A_3, A_4),
 \ldots, S_\ell(A_{2\ell-1}, A_{2\ell}), 
 (A_{2i} < A_{2i+1})$
 & $\min(W_1 + W_2 + \ldots)$ \\
 
 \hline
 $Q_{S2}(\ldots) \datarule S_1(A_1, A_2), S_2(A_3, A_4), 
 \ldots,
 S_\ell(A_{2\ell-1}, A_{2\ell}), 
 (|A_{2i} - A_{2i+1}| < 50), (A_{2i-1} \neq A_{2i+2})$
 & $\min(W_1 + W_2 + \ldots)$ \\
 
 \hline
 $Q_{T}(\ldots) \datarule \textrm{Item}(O_1, PK_1, SK, L_1, Q_1, P_1, S_1, C_1, R_1), \textrm{Item}(O_2, PK_2, SK, L_2, Q_2, P_2, S_2, C_2, R_2), \ldots,
 (Q_i < Q_{i+1}), (S_i < S_{i+1})$
 & $\min(P_1 + P_2 + \ldots)$ \\
 
 \hline
 $Q_{TD}(\ldots) \datarule \textrm{Item}(O_1, PK_1, SK, L_1, Q_1, P_1, S_1, C_1, R_1), \textrm{Item}(O_2, PK_2, SK, L_2, Q_2, P_2, S_2, C_2, R_2), \ldots,
 (Q_i < Q_{i+1}), (S_i < S_{i+1} \vee C_i < C_{i+1} \vee R_i < R_{i+1})$
 & $\min(P_1 + P_2 + \ldots)$ \\
 
 \hline
 $Q_{R1}(\ldots) \datarule \textrm{Reddit}(N_1, N_2, T_1, S_1, R_1), \textrm{Reddit}(N_2, N_3, T_2, S_2, R_2), \ldots,
 (T_i < T_{i+1})$
 & $\min(S_1 + S_2 + \ldots)$ \\
 
 \hline
 $Q_{R2}(\ldots) \datarule \textrm{Reddit}(N_1, N_2, T_1, S_1, R_1), \textrm{Reddit}(N_2, N_3, T_2, S_2, R_2), \ldots,
 (T_i < T_{i+1}), (S_i > S_{i+1})$
 & $\max(R_1 + R_2 + \ldots)$ \\
 \hline
 $Q_{B}(\ldots) \datarule \textrm{Birds}(I_1, LA_1, LO_1, C_1), 
 \textrm{Birds}(I_2, LA_2, LO_2, C_2)
 (|LA_1 - LA_2| < \epsilon), (|LO_1 - LO_2| < \epsilon)$
 & $\max(C_1 + C_2)$ \\
 \hline
\end{tabular} 
\caption{Queries used in our experiments expressed in Datalog. The head always contains all body variables (no projections).
Length $\ell$ of queries range from $2$ to $10$. Indices $i$ range from $1$ to $\ell-1$.
}
\label{tab:queries}
\end{figure*}

We demonstrate the superiority of our approach for ranked enumeration 
against existing \DBMSs, and even idealized competitors that receive the join output
``for free" as an (unordered) array.

\introparagraph{Algorithms}
We compare 5 algorithms:
\circled{1} \OURS is our proposed approach.
\circled{2} \QUADEQUI is an idealized version of the fairly straightforward reduction to
equi-joins described in \cref{sec:to_equi},
which for each edge $(S,T)$ of the theta-join tree uses the direct \TLFG
(no intermediate nodes) to convert $S \bowtie_\theta T$ to equi-join
$S \bowtie E \bowtie T$ via the edge set $E$ of the \TLFG. Then previous ranked-enumeration
techniques for equi-joins \cite{tziavelis20vldb} can be applied directly.
To avoid any concerns regarding the choice of technique for generating $E$, we
provide it ``for free.'' Hence the algorithm is not charged
for essentially executing theta-joins between all pairs of adjacent relations in the
theta-join tree, meaning the \QUADEQUI numbers reported here represent a
\emph{lower bound} of real-world running time.
\circled{3} \BATCH is an idealized version of the approach taken by state-of-the-art \DBMSs.
It computes the entire join output and puts it into a heap for ranked enumeration.
To avoid concerns about the most efficient join implementation, we give \BATCH the
entire join output ``for free'' as an in-memory array. It simply needs to read those
output tuples (instead of having to execute the actual join) to rank them,
therefore the numbers reported constitute a \emph{lower bound} of real-world running time.
We note that for a join of only $\ell=2$ relations, there is no difference between
\QUADEQUI and \BATCH since they both receive all the query results;
we thus omit \QUADEQUI for binary joins.
\circled{4} \PSQL is the open-source PostgreSQL system. 
\circled{5} \SYSX is a commercial database system that is highly optimized for in-memory computation.

We also compare our factorization methods
\circled{1a} \BINPART, \circled{1b} \MULTIPART, and \circled{1c} \SHAREDRAN
against each other. Recall that the latter two can only be applied to single-inequality
type join conditions. Unless specified otherwise, \OURS is set to \circled{1b} \MULTIPART
for the single-predicate cases and \circled{1a} \BINPART for all others.

\introparagraph{Data}
\circled{S} Our synthetic data generator creates relations $S_i(A_i, A_{i+1}, W_i), \, i \!\geq\! 1$ 
by drawing $A_i, A_{i+1}$ from integers in $[0 \ldots 10^4 - 1]$ 
uniformly at random with replacement, discarding duplicate tuples.
The weights $W_i$ are real numbers drawn from $[0, 10^4)$.
\circled{T} We also use the LINEITEM relation of the TPC-H benchmark \cite{tpch}, keeping the schema
\texttt{Item(OrderKey, PartKey, Suppkey, LineNumber, Quantity, Price, ShipDate, CommitDate, ReceiptDate)}.

\circled{R} For real data, we use a temporal graph \Reddit~\cite{kumar18reddit}
whose $286k$ edges represent posts from a source community to a target community
identified by a hyperlink in the post title.
The schema is \texttt{Reddit(From, To, Timestamp, Sentiment, Readability)}.
\circled{B} \Birds~\cite{birdsOceania} reports bird observations from
Oceania with schema \texttt{Birds(ID, Latitude, Longitude, Count)}.
We keep only the $452k$ observations with a non-empty \texttt{Count} attribute.

\introparagraph{Queries}
We test queries with various join conditions and sizes.
\Cref{tab:queries} gives the Datalog notation and the ranking function.
Some of the queries have the number of relations $\ell$ as a parameter;
for those we only write the join conditions between the
$i^\textrm{th}$ and $(i+1)^\textrm{st}$ relations,
with the rest similarly organized in a chain.
In the full version \cite{tziavelis21full} we give the equivalent SQL queries.

On our synthetic data, $Q_{S1}$ is a single inequality join,
while $Q_{S2}$ has a more complicated join condition that is a conjunction of a band and a non-equality.
On TPC-H, $Q_T$ finds a sequence of items sold by the same supplier with the quantity increasing over time, ranked by the price.
To test disjunctions, we run query $Q_{TD}$, which puts the increasing time
constraint on either of the three possible dates.
Query $Q_{R1}$ computes temporal paths~\cite{huanhuan14temporal} on \Reddit, 
and ranks them by a measure of sentiment such that sequences of negative posts are retrieved first.
Query $Q_{R2}$ uses instead the sentiment in the join condition,
keeping only paths along which the negative sentiment increases.
For ranking, we use readability to focus on posts of higher quality.
Last, $Q_{B}$ is a spatial band join on \Birds that finds pairs of high-count
bird sightings that are close based on proximity.

\introparagraph{Details}
Our algorithms are implemented in Java 8 and executed on an Intel Xeon E5-2643 CPU
running Ubuntu Linux. Queries execute in memory on a Java VM with 100GB of RAM.
If that is exceeded, we report an Out-Of-Memory (OOM) error.
The any-k algorithm used by \OURS and \QUADEQUI is \LAZY \cite{chang15enumeration,tziavelis20vldb}
which was found to outperform others in previous work.
The version of PostgreSQL is 9.5.25.
We set its parameters such that 
it is optimized for main-memory execution and 
system overhead related to logging or concurrency is minimized,
as it is standard in the literature \cite{tziavelis20vldb,bakibayev12fdb}.
To enable input caching for \PSQL and \SYSX, each execution is performed twice and we only measure the second one.
Additionally, we create B-tree or hash indexes for each attribute of the input relations,
while our methods do not receive these indexes. 
Even though the task is ranked enumeration, we still give the database systems a LIMIT clause whenever we measure a specific $\TT(k)$,
and thus allow them to leverage the $k$ value.
All data points we show are the median of 5 measurements.
We timeout any execution that does not finish within $2$ hours.

\subsection{Comparison Against Alternatives}

We will show that our approach has a significant advantage over the competition
when the size of the output is sufficiently large.
We test three distinct scenarios for which large output can occur:
($1$) the size of the database grows,
($2$) the length of the query increases,
and ($3$) the parameter of a band join increases.

\resultbox{%
\introparagraph{Summary}
\circled{1} \OURS is superior when the total output size is large,
even when compared against a lower bound of the running time of the other methods.
\circled{2} \QUADEQUI and 
\circled{3} \BATCH require significantly more memory and are infeasible for many queries.
\circled{4} \PSQL and \circled{5} \SYSX, similarly to \BATCH, 
must produce the entire output, which is very costly.
While \SYSX is clearly faster than \PSQL,
it can be several orders of magnitude slower than our \OURS,
and is outperformed across all tested queries.
}

\begin{figure*}[t]
    \centering
    \begin{subfigure}{\linewidth}
        \centering
        \includegraphics[width=0.73\linewidth]{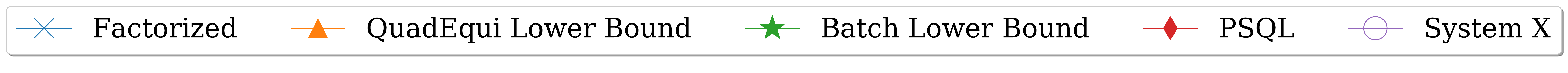}
    \end{subfigure}
    \vspace{-3mm}
    
    \begin{subfigure}{0.25\linewidth}
        \centering
        \includegraphics[height=2.8cm]{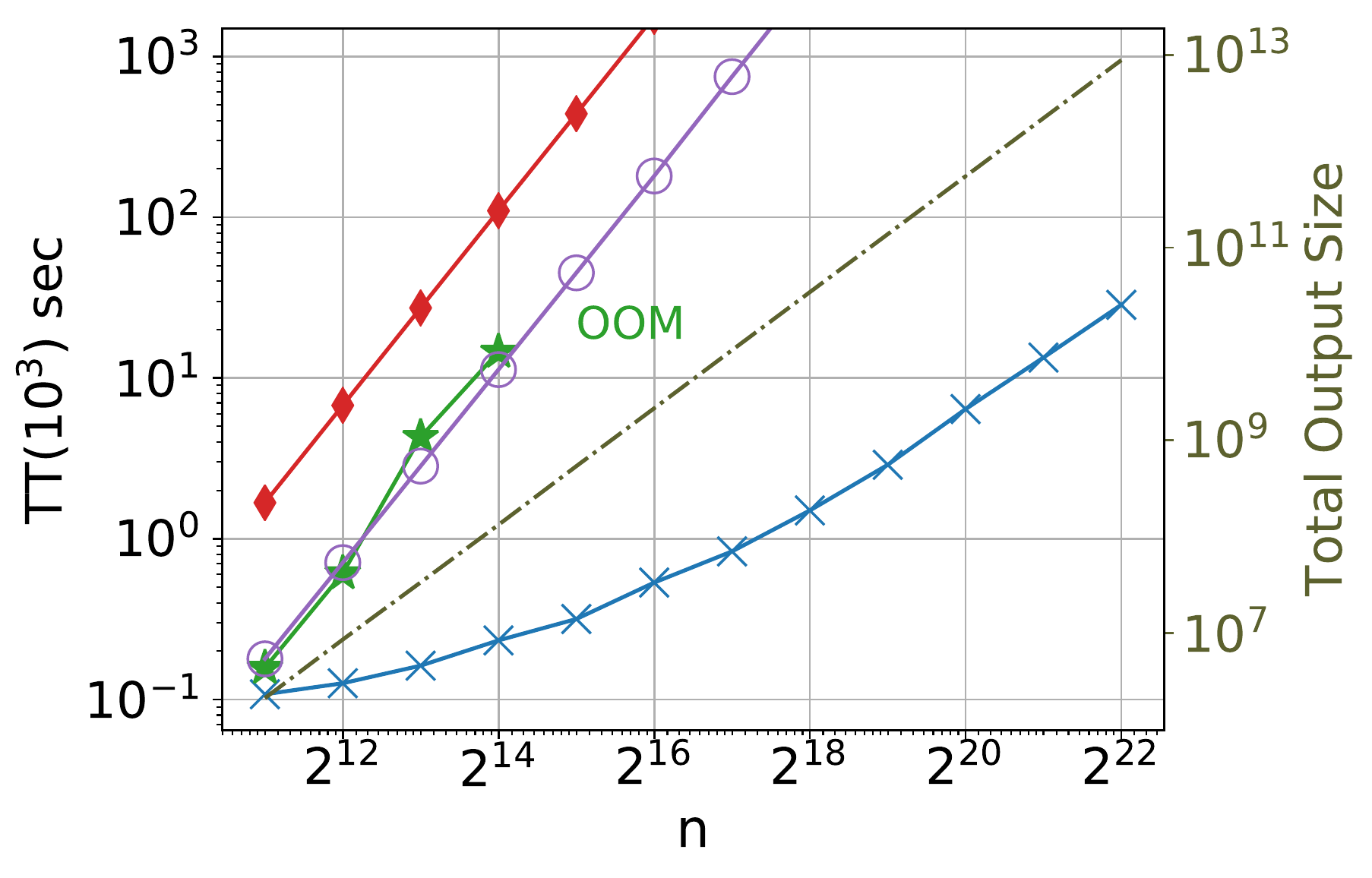}
        \caption{Query $Q_{S1}$, $\ell = 2$.}
		\label{exp:syn_q1_l2}
    \end{subfigure}%
    \hfill
    \begin{subfigure}{0.25\linewidth}
        \centering
        \includegraphics[height=2.8cm]{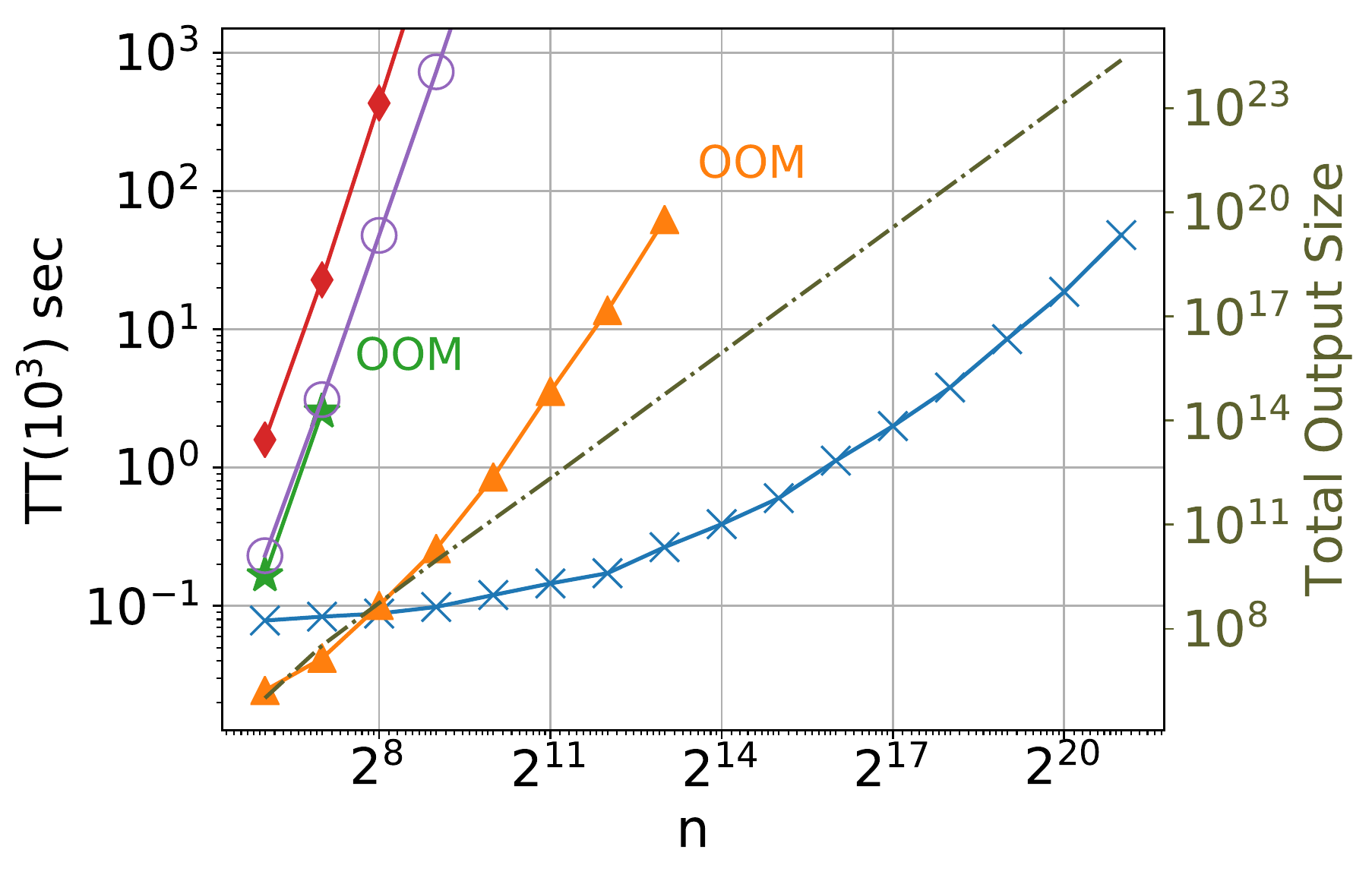}
        \caption{Query $Q_{S1}$, $\ell = 4$.}
		\label{exp:syn_q1_l4}
    \end{subfigure}%
    \hfill
    \begin{subfigure}{0.25\linewidth}
        \centering
        \includegraphics[height=2.8cm]{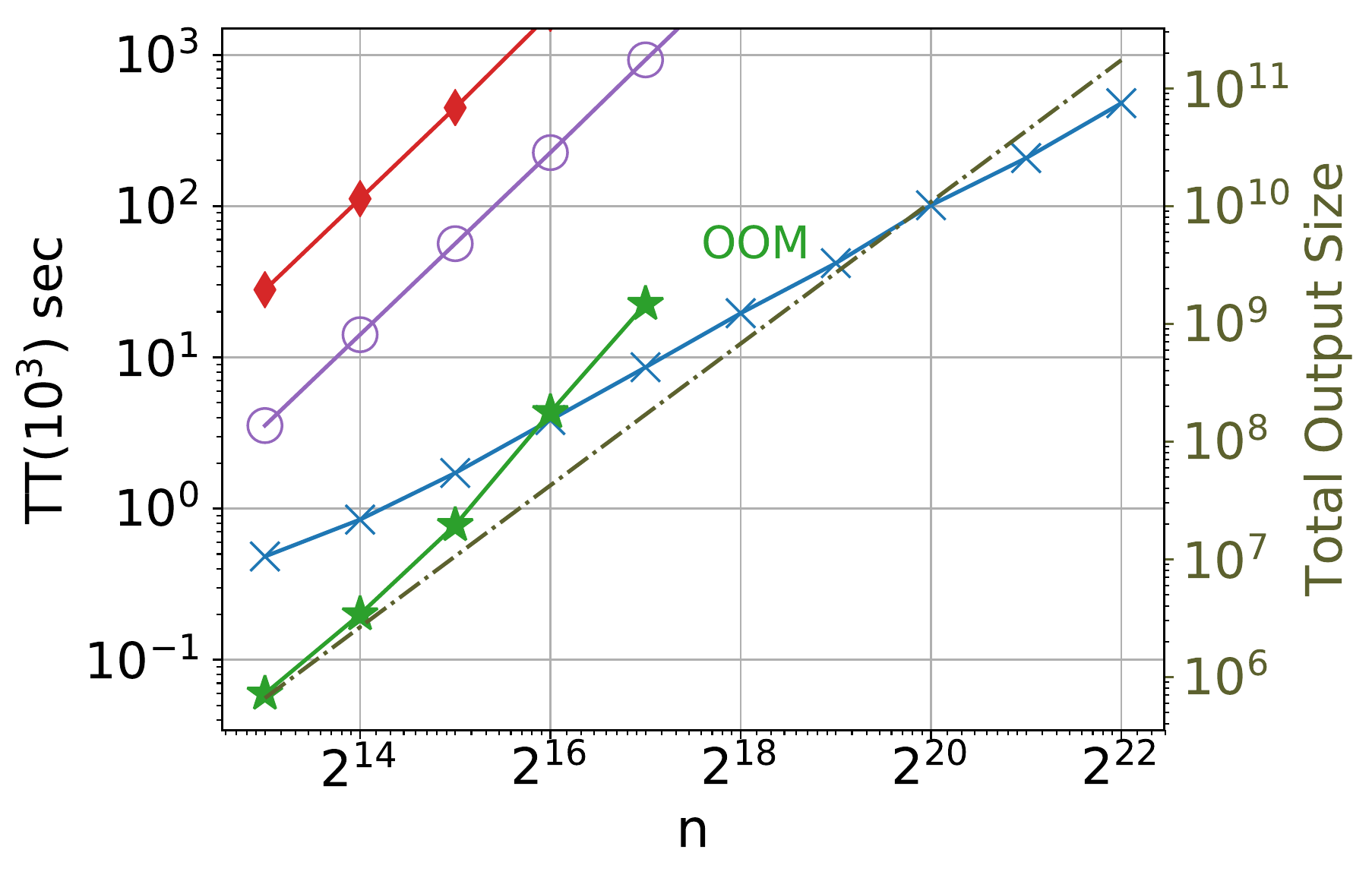}
        \caption{Query $Q_{S2}$, $\ell = 2$.}
		\label{exp:syn_q2_l2}
    \end{subfigure}%
    \hfill
    \begin{subfigure}{0.25\linewidth}
        \centering
        \includegraphics[height=2.8cm]{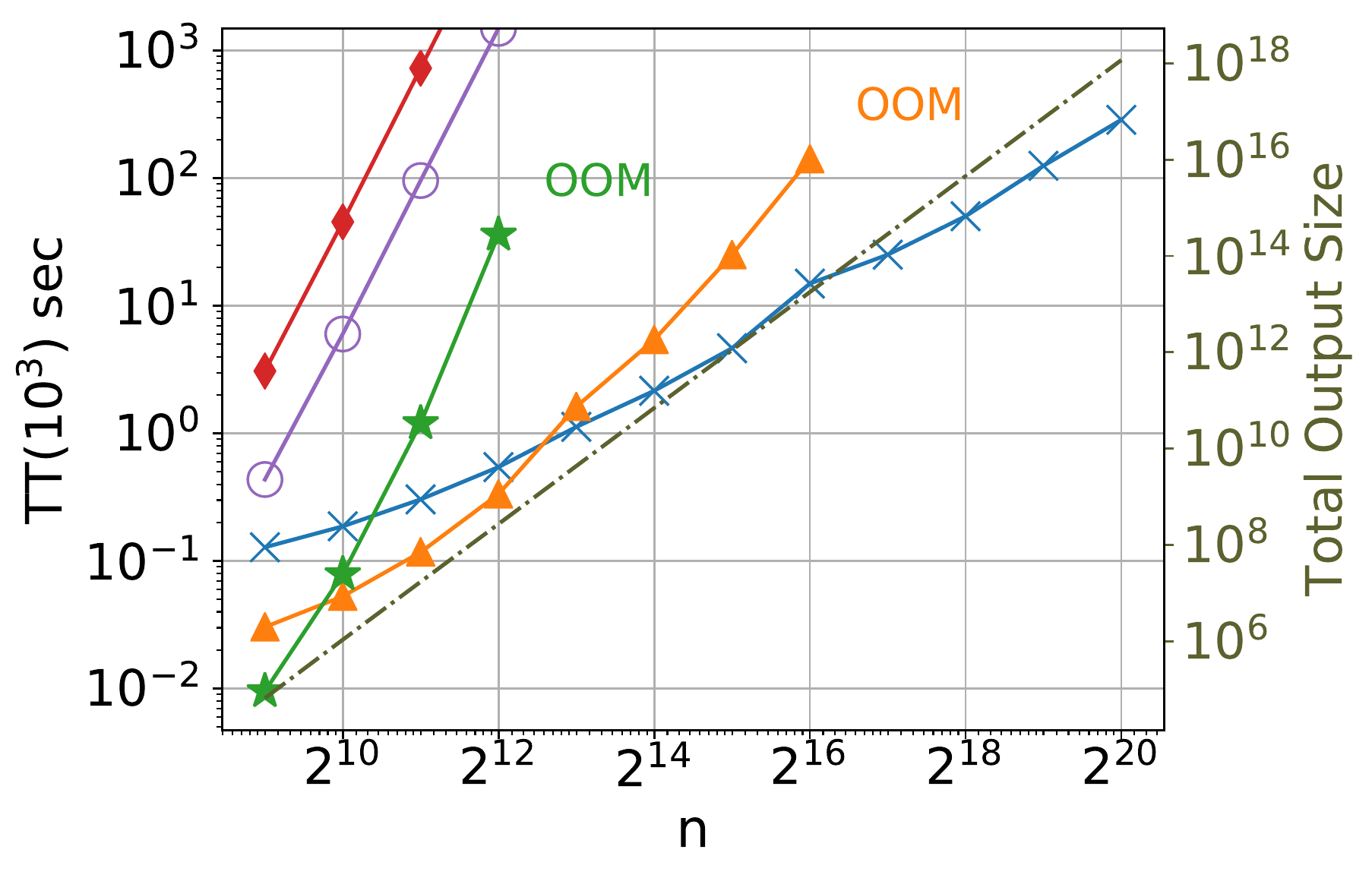}
        \caption{Query $Q_{S2}$, $\ell = 4$.}
		\label{exp:syn_q2_l4}
    \end{subfigure}
    \vspace{-1mm}
    
    \caption{\Cref{subsec:datasize}: Synthetic data with a growing database size $n$. 
    While all four alternative methods either run out of memory (``OOM'')
	or exceed a reasonable running time,
    our method scales quasilinearly ($\O(n \polylog n)$) with $n$.
    }
    \label{exp:syn}
\end{figure*}

\begin{figure}[t]
    \centering   
    \begin{subfigure}{0.5\linewidth}
        \centering
        \includegraphics[height=2.8cm]{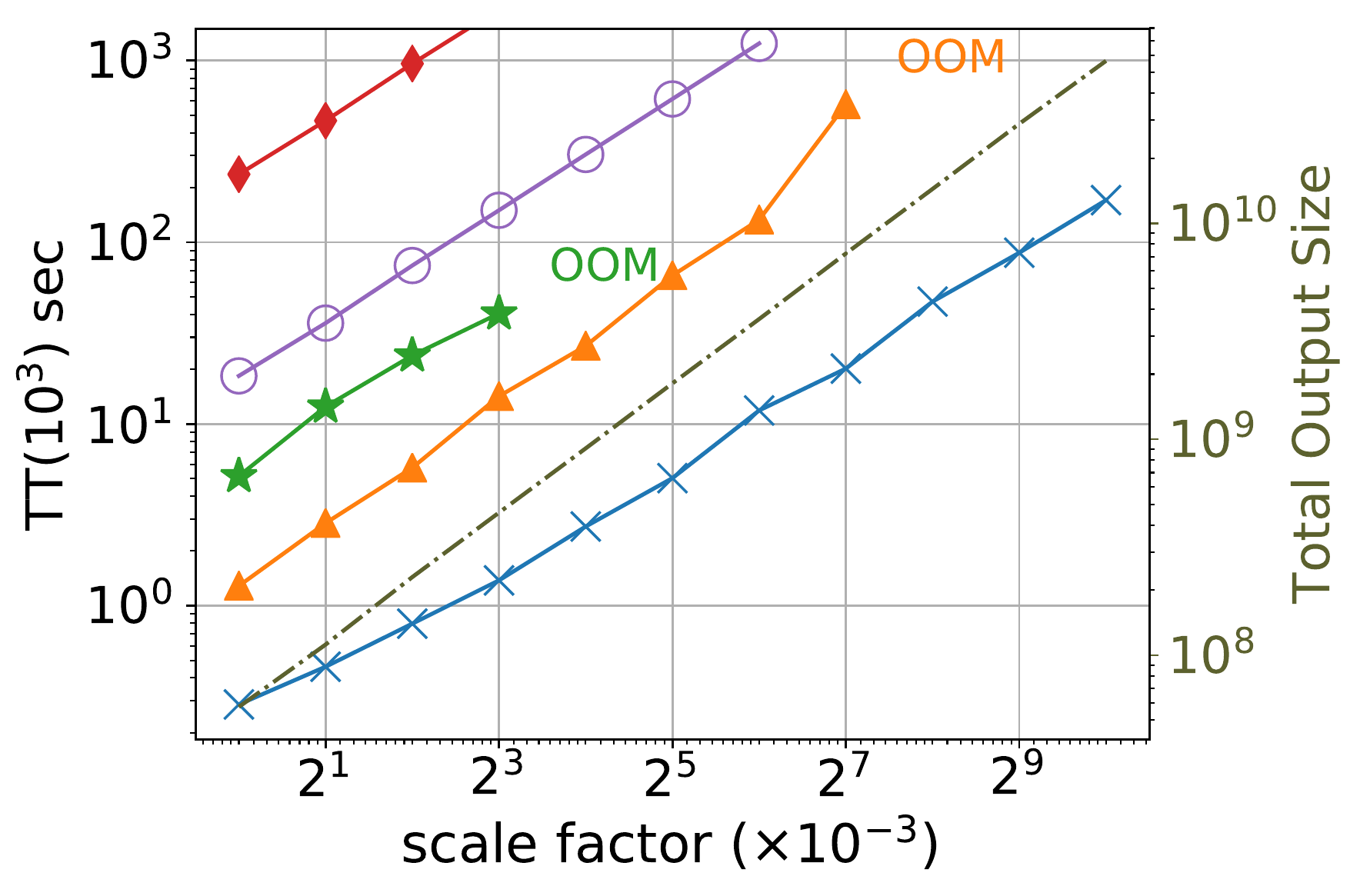}
        \caption{Query $Q_{T}$, $\ell = 3$.}
		\label{exp:syn_qt_l3}
    \end{subfigure}%
    \hfill
    \begin{subfigure}{0.5\linewidth}
        \centering
        \includegraphics[height=2.8cm]{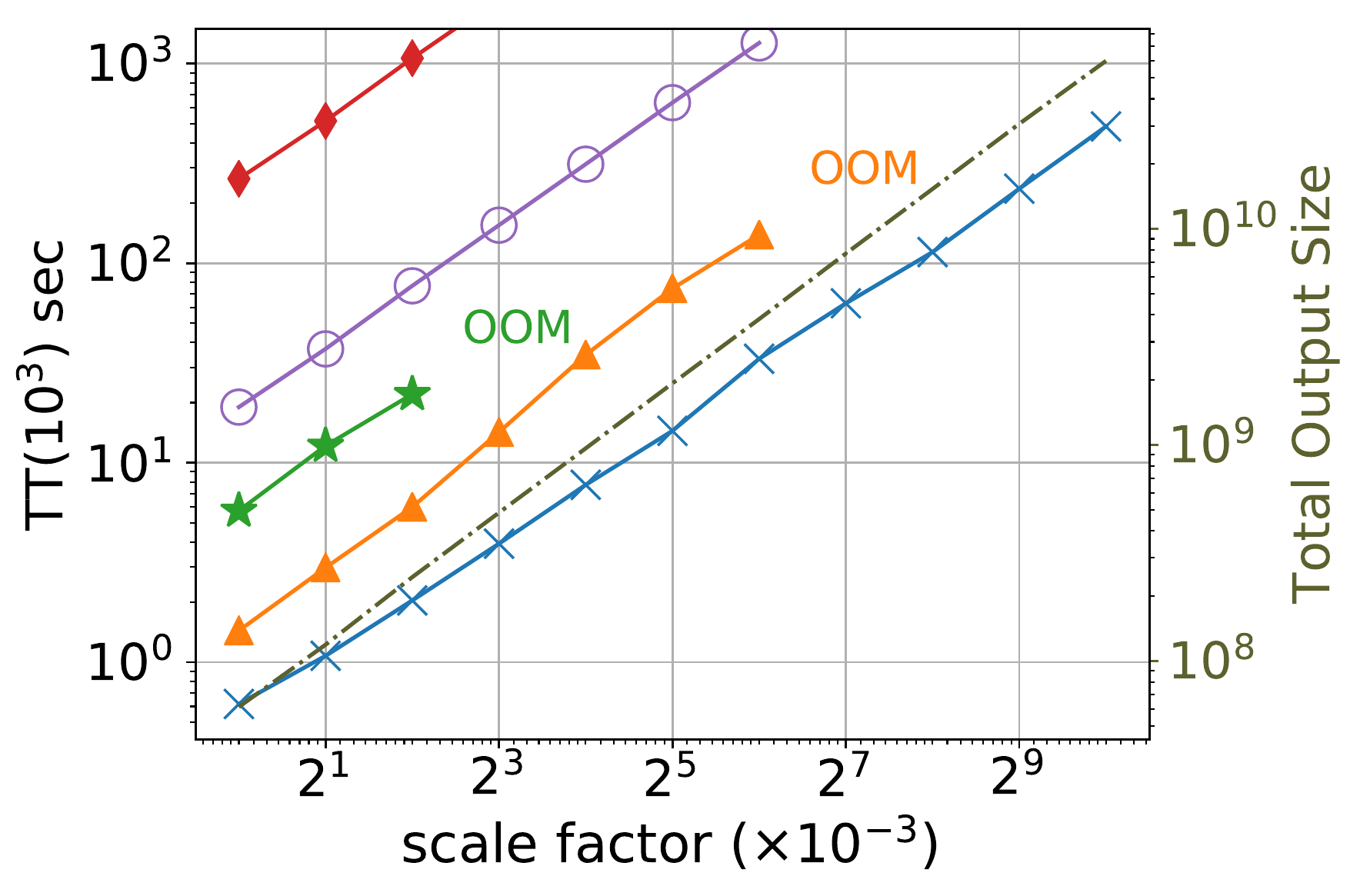}
        \caption{Query $Q_{TD}$, $\ell = 3$.}
		\label{exp:syn_qtd_l3}
    \end{subfigure}%
    \caption{\Cref{subsec:datasize}: TPC-H data with increasing scale factor. Disjunctions do not affect the scaling of our algorithm.}
    \label{exp:tpch}
\end{figure}

\subsubsection{Effect of Data Size}\label{subsec:datasize}
We run queries $Q_{S1}, Q_{S2}$ for different input sizes $n$ and two distinct query lengths.
\Cref{exp:syn} depicts the time to return the top $k = 10^3$ results.
We also plot how the size of the output grows with increasing $n$ on a secondary y-axis.
Even though \QUADEQUI and \BATCH are given precomputed join results for free
and do not even have to resolve complicated join predicates,
they still require a large amount of memory to store those.
Thus, they quickly run out of memory even for relatively small inputs (\Cref{exp:syn_q1_l4}).
\PSQL does not face a memory problem because it can resort to secondary storage, 
yet becomes unacceptably slow.
The in-memory optimized \SYSX is $10$ times faster than \PSQL,
but still follows the same trend because it is materializing the entire output.
In contrast, our \OURS approach scales smoothly across all tests
and requires much less memory.
For instance, in \Cref{exp:syn_q1_l4} \QUADEQUI fails after $8k$ input size, 
while we can easily handle $2M$. 
For very small input, the idealized methods \QUADEQUI and \BATCH are sometimes faster,
but their real running time would be much higher if
join computation was accounted for.
$Q_{S2}$ has more join predicates and thus smaller output size
(\Cref{exp:syn_q2_l2,exp:syn_q2_l4}).
Our advantage is smaller in this case,
yet still significant for large $n$.

We similarly run queries $Q_T$ (\Cref{exp:syn_qt_l3}) and $Q_{TD}$ (\Cref{exp:syn_qtd_l3}) for $\ell=3$
with an increasing scale factor 
(which determines data size).
Here, the equi-join condition on the supplier severely limits the blowup of the output compared to the input.
Still, \OURS is again superior.
Disjunctions in $Q_{TD}$ increase the running time of our technique only slightly by a small constant factor.

\begin{figure*}[t]
    \centering
    \begin{subfigure}{\linewidth}
        \centering
        \includegraphics[width=0.82\linewidth]{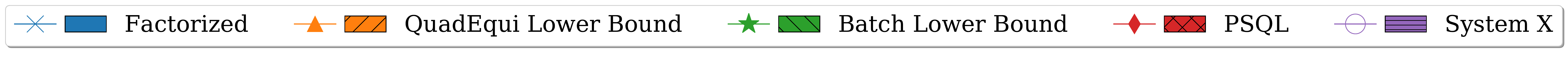}
    \end{subfigure}
    \vspace{-3mm}
    
    \begin{subfigure}[t]{0.24\linewidth}
        \centering
        \includegraphics[height=3cm]{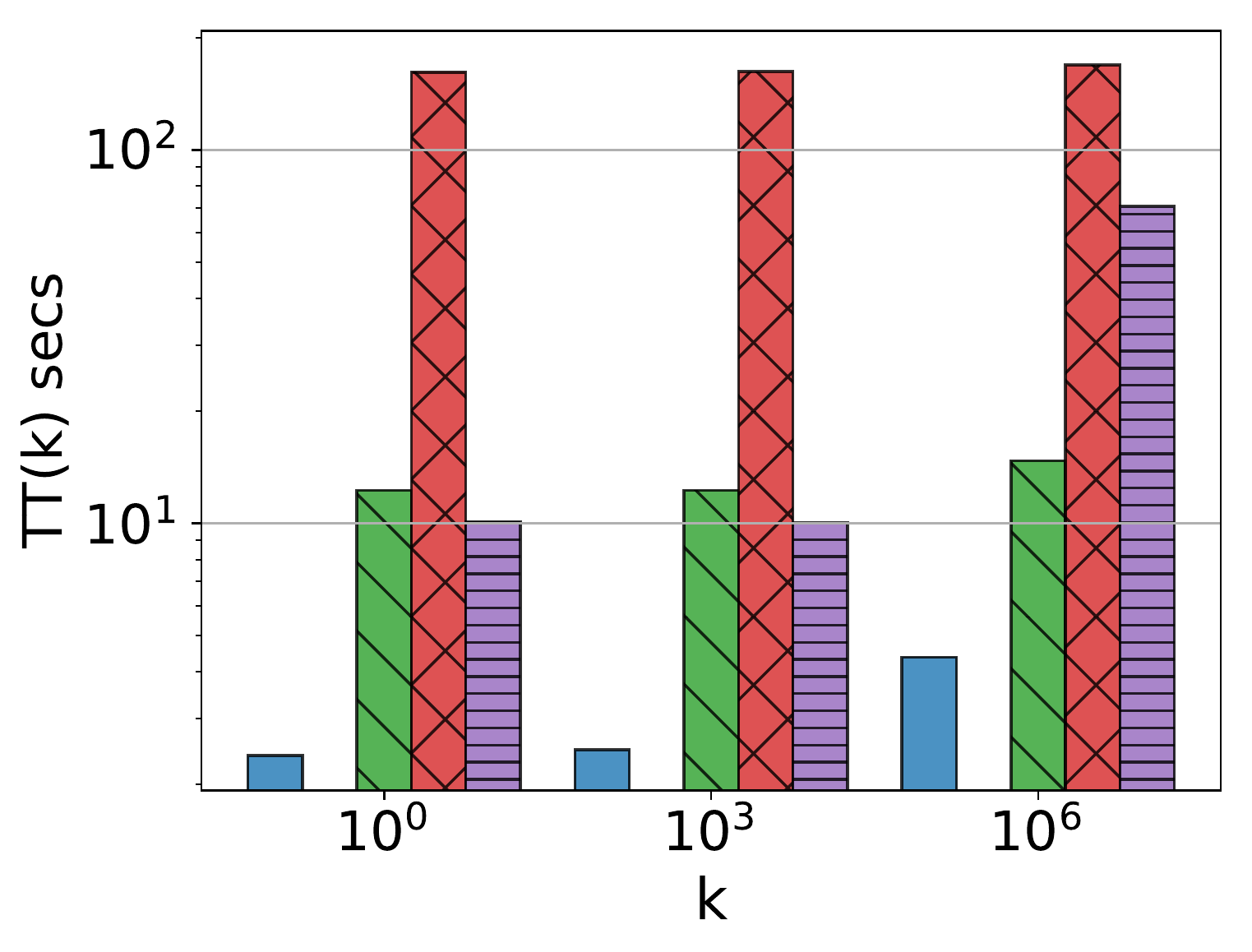}
        \caption{Query $Q_{R1}$, $\ell = 2$.}
		\label{exp:reddit_q1_l2}
    \end{subfigure}%
    \begin{subfigure}[t]{0.24\linewidth}
        \centering
        \includegraphics[height=3cm]{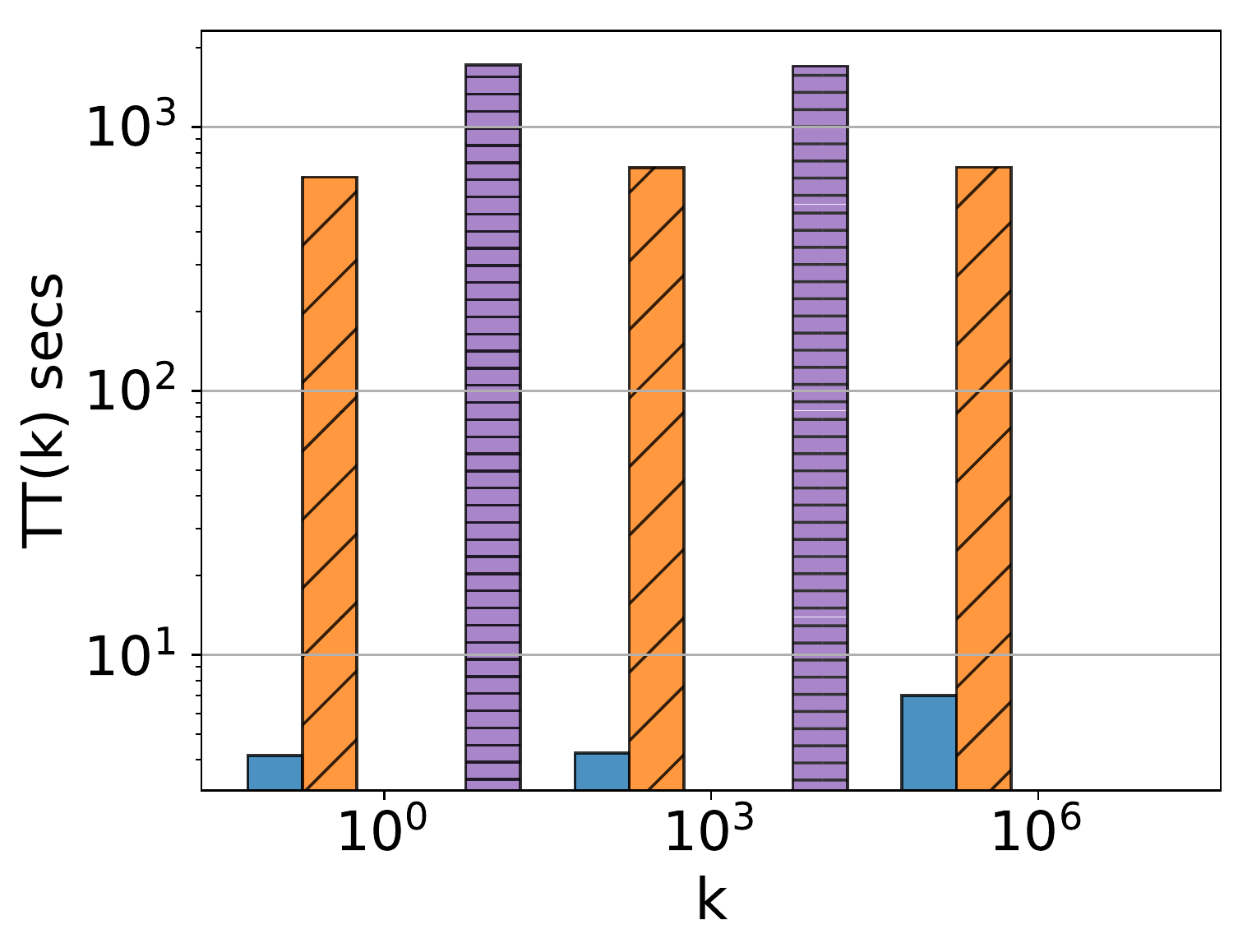}
        \caption{Query $Q_{R1}$, $\ell = 3$.}
		\label{exp:reddit_q1_l3}
    \end{subfigure}%
    \begin{subfigure}[t]{0.24\linewidth}
        \centering
        \includegraphics[height=3cm]{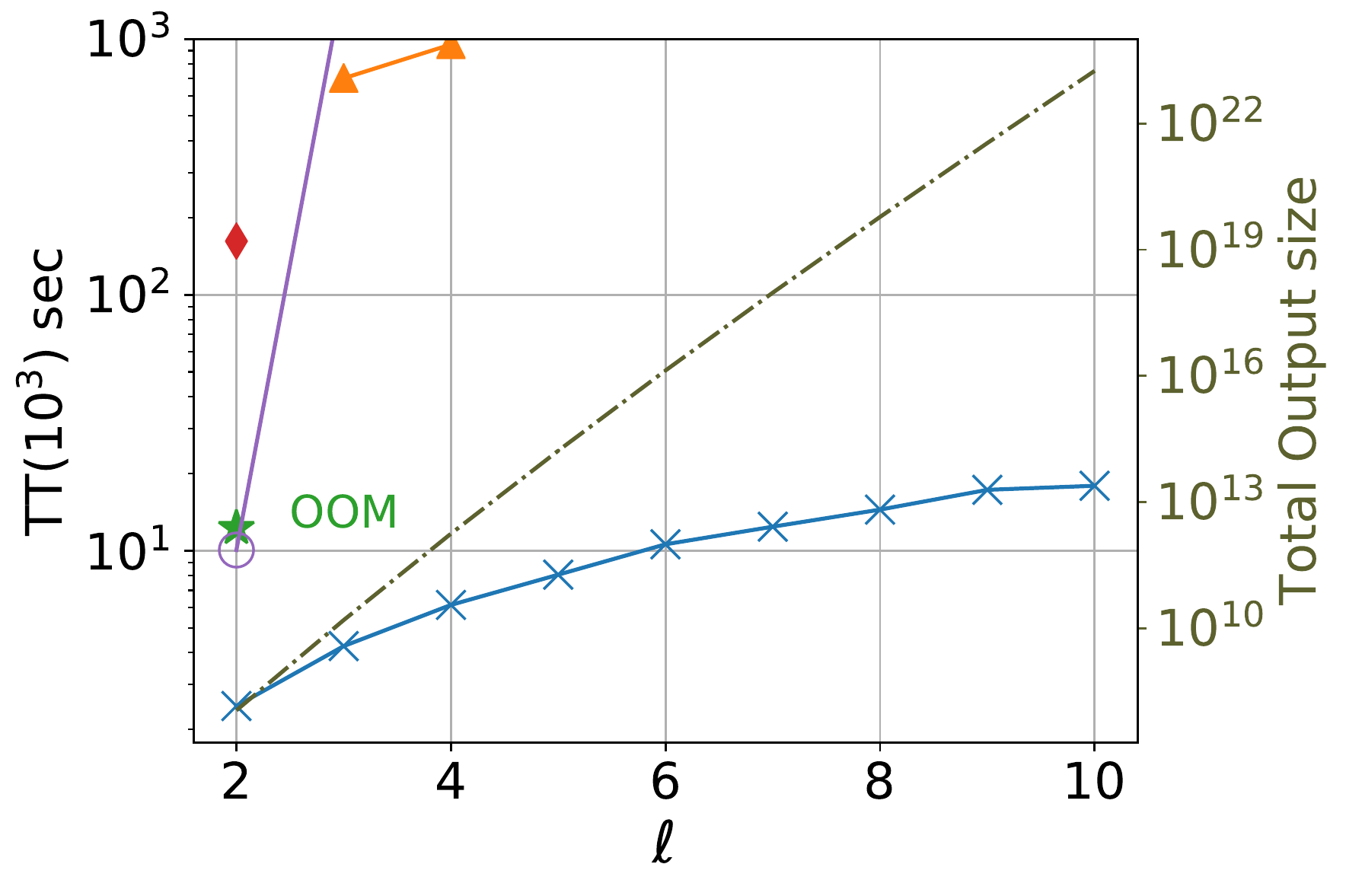}
        \caption{Query $Q_{R1}$, different lengths $\ell$.}
		\label{exp:reddit_q1}
    \end{subfigure}
    \hfill
    \begin{subfigure}[t]{0.24\linewidth}
        \centering
        \includegraphics[height=3cm]{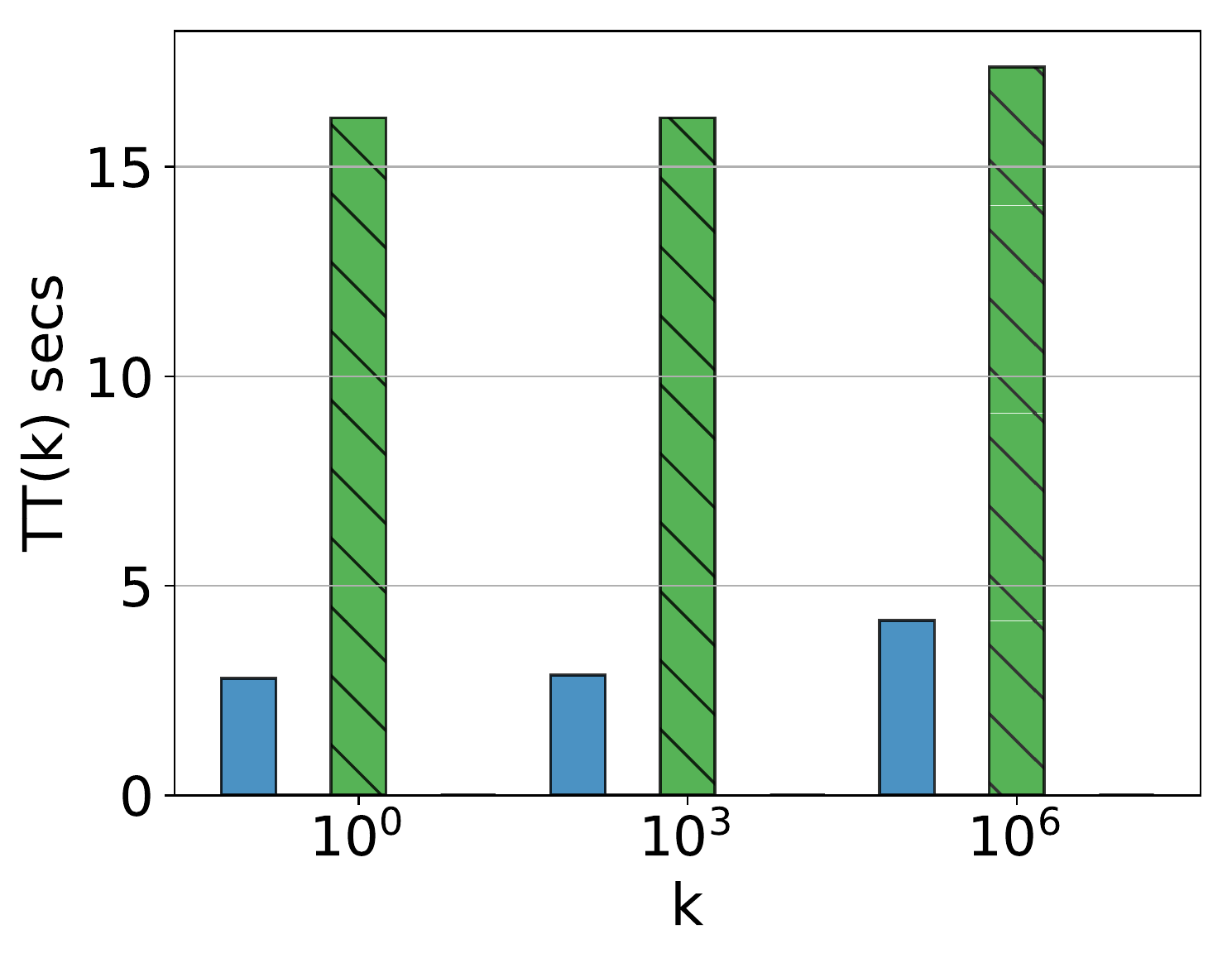}
        \caption{Query $Q_B$, fixed $\epsilon = 0.01$.}
		\label{exp:birds_e001}
    \end{subfigure}%
    \vspace{-1mm}

    \begin{subfigure}[t]{0.24\linewidth}
        \centering
        \includegraphics[height=3cm]{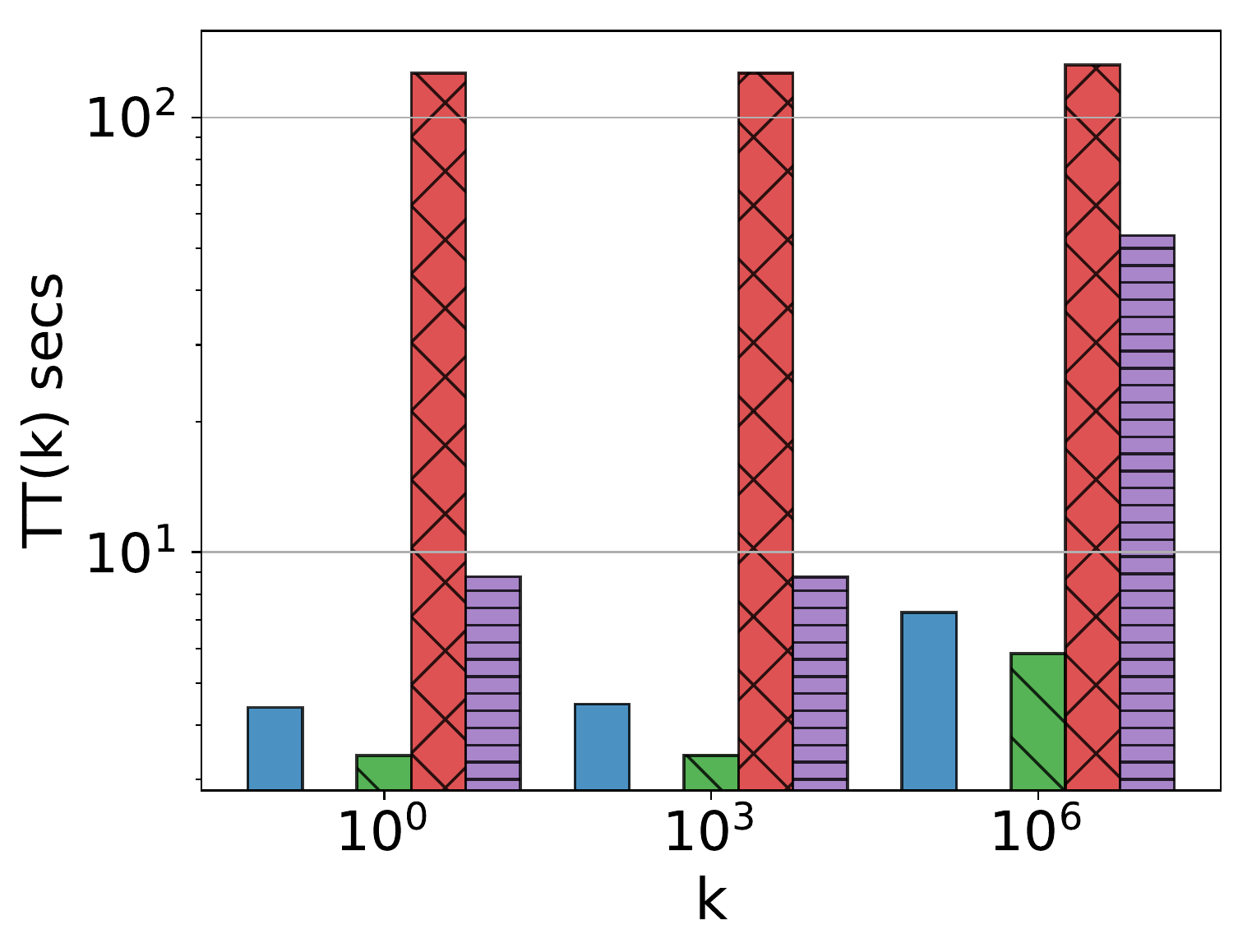}
        \caption{Query $Q_{R2}$, $\ell = 2$.}
		\label{exp:reddit_q2_l2}
    \end{subfigure}%
    \begin{subfigure}[t]{0.24\linewidth}
        \centering
        \includegraphics[height=3cm]{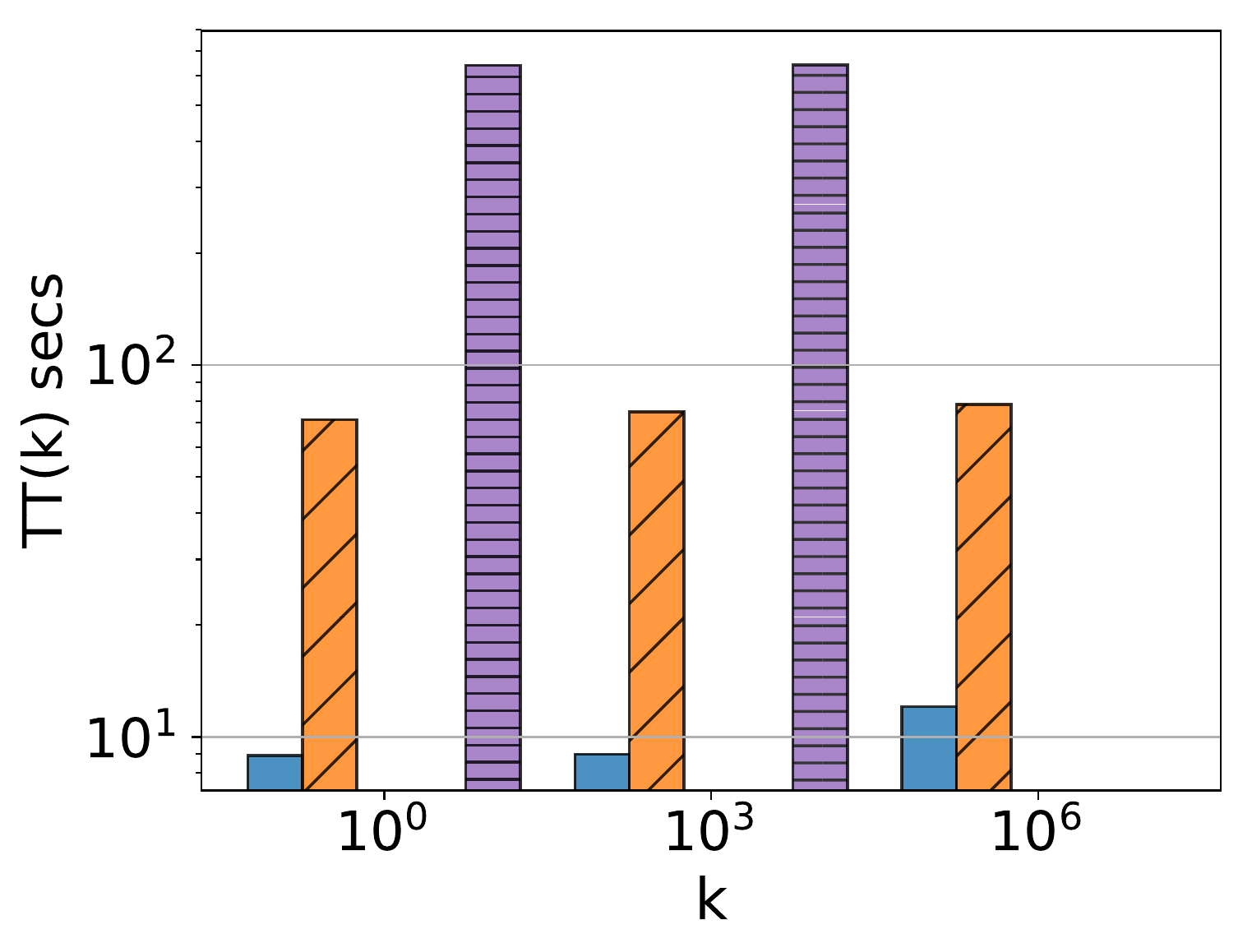}
        \caption{Query $Q_{R2}$, $\ell = 3$.}
		\label{exp:reddit_q2_l3}
    \end{subfigure}%
    \begin{subfigure}[t]{0.24\linewidth}
        \centering
        \includegraphics[height=3cm]{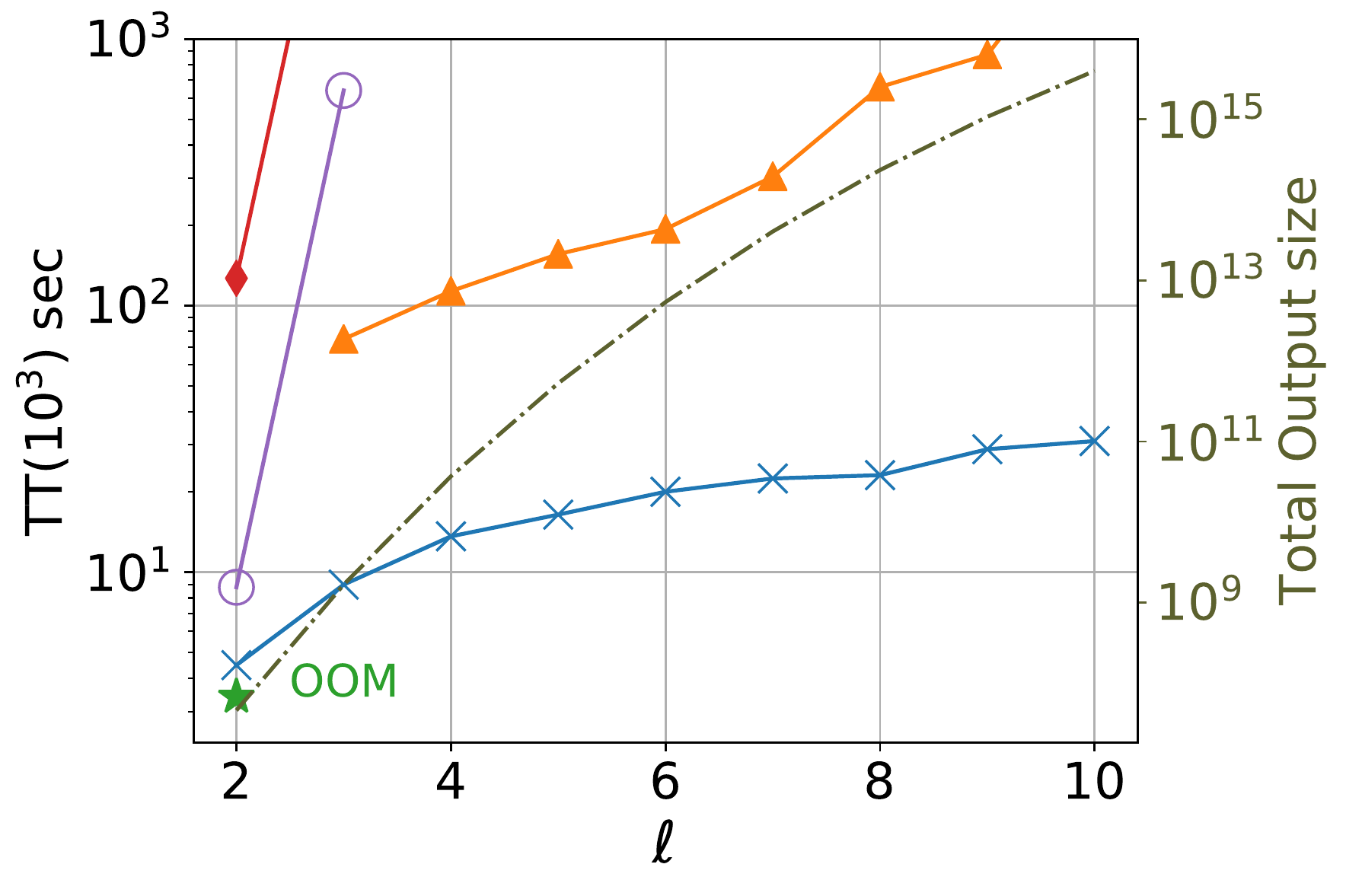}
        \caption{Query $Q_{R2}$, different lengths $\ell$.}
		\label{exp:reddit_q2}
    \end{subfigure}
    \hfill
    \begin{subfigure}[t]{0.24\linewidth}
        \centering
        \includegraphics[height=3cm]{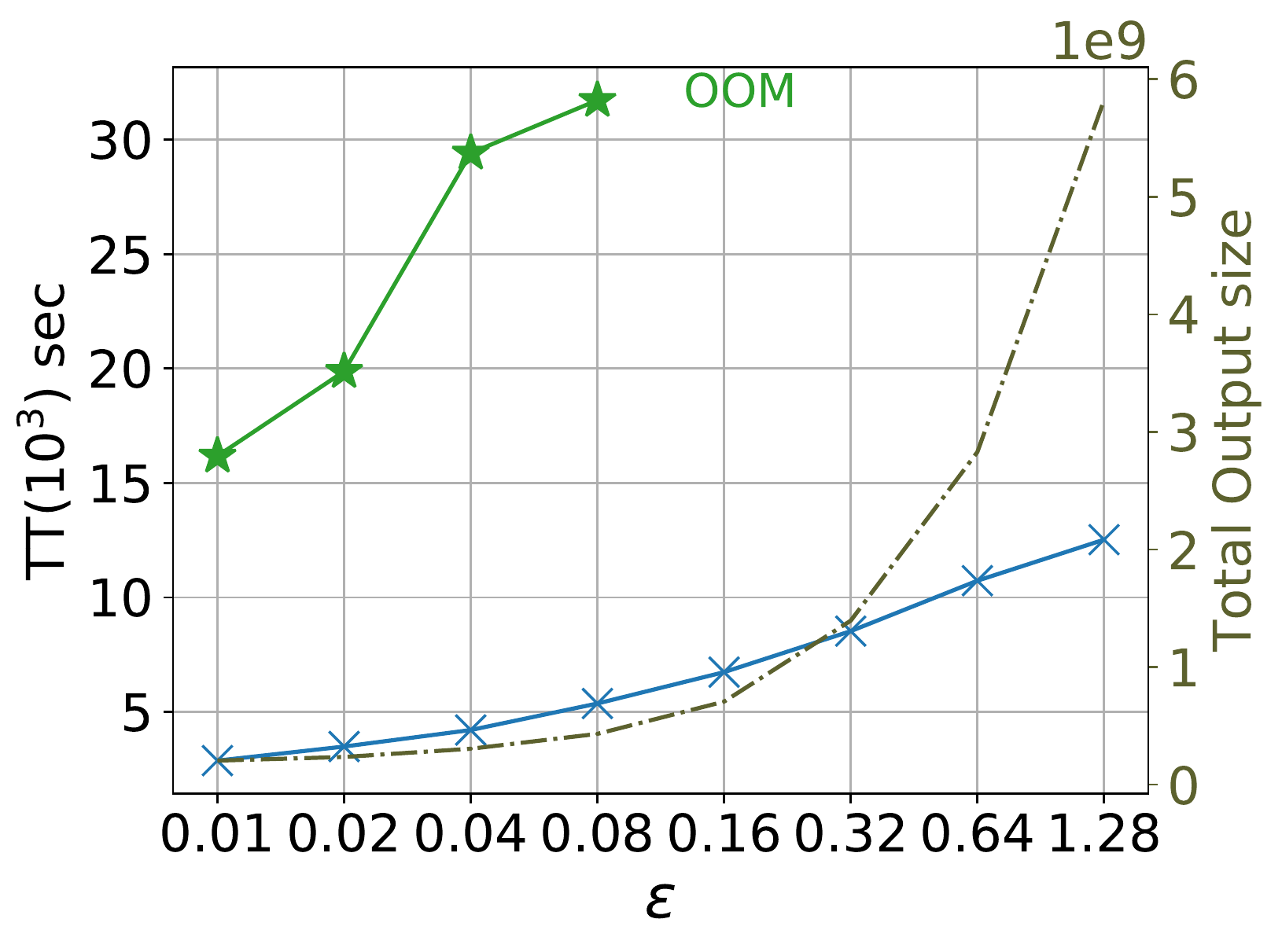}
        \caption{Query $Q_B$, different bands $\epsilon$.}
		\label{exp:birds_e}
    \end{subfigure}    
    \vspace{-1mm}
    
    \caption{\Cref{subsec:querylength}: a,b,c,e,f,g:  
    \Cref{subsec:band}: d, h:
    Temporal paths of different lengths on \Reddit (left), and spatial band-join on \Birds (right).
    Our method is robust to increasing query sizes and band-join ranges.}
    \label{exp:reddit}
\end{figure*}

\subsubsection{Effect of Query Length}\label{subsec:querylength}
Next, we test the effect of query length on \Reddit.
We plot $\TT(k)$ for three values ($k = 1, 10^3, 10^6$)
when the length is small ($\ell = 2, 3$) and one value ($k = 10^3$)
for longer queries.
Note that for $k=1$, the time of \OURS is essentially the time required for building our \TLFGs,
and doing a bottom-up Dynamic Programming pass \cite{tziavelis20vldb}.
\Cref{exp:reddit} depicts our results for queries $Q_{R1}, Q_{R2}$.
Increasing the value of $k$ 
does not have a serious impact for most of the approaches
except for \SYSX,
which for $k=10^6$ is not able to provide the same optimized execution.
For binary-join $Q_{R1}$, our \OURS is faster than the \BATCH lower bound (\Cref{exp:reddit_q1_l2}),
and its advantage increases for longer queries, 
since the output also grows (\Cref{exp:reddit_q1}).
\BATCH runs out of memory for $\ell=3$, 
\PSQL times out,
while \QUADEQUI and \SYSX are more than $100$ times slower (\Cref{exp:reddit_q1_l3}).
Query $Q_{R2}$ has an additional join predicate, 
hence its output size is smaller.
Thus, the \BATCH lower bound is slightly better than our approach for $\ell=2$ (\Cref{exp:reddit_q2_l2}),
but we expect it to be significantly slower if the cost
of computing and materializing the output was taken into account.
Either way, for $\ell \geq 3$ (\Cref{exp:reddit_q2}), our approach dominates even when compared against the lower bounds.
\PSQL again times out for $\ell=3$ (\Cref{exp:reddit_q2_l3}),
and the highly optimized \SYSX is outclassed by our approach.

\begin{figure*}[t]
    \centering
    \begin{subfigure}{\linewidth}
        \centering
        \includegraphics[width=0.55\linewidth]{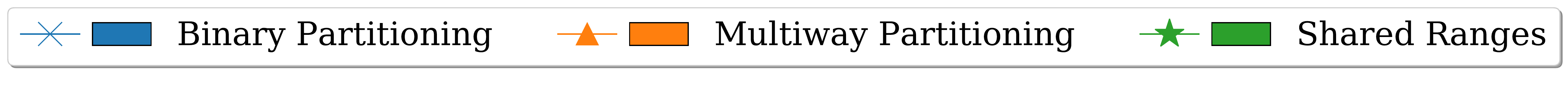}
    \end{subfigure}
    \vspace{-3mm}
    
    \begin{subfigure}[t]{0.24\linewidth}
        \centering
        \includegraphics[width=\linewidth]{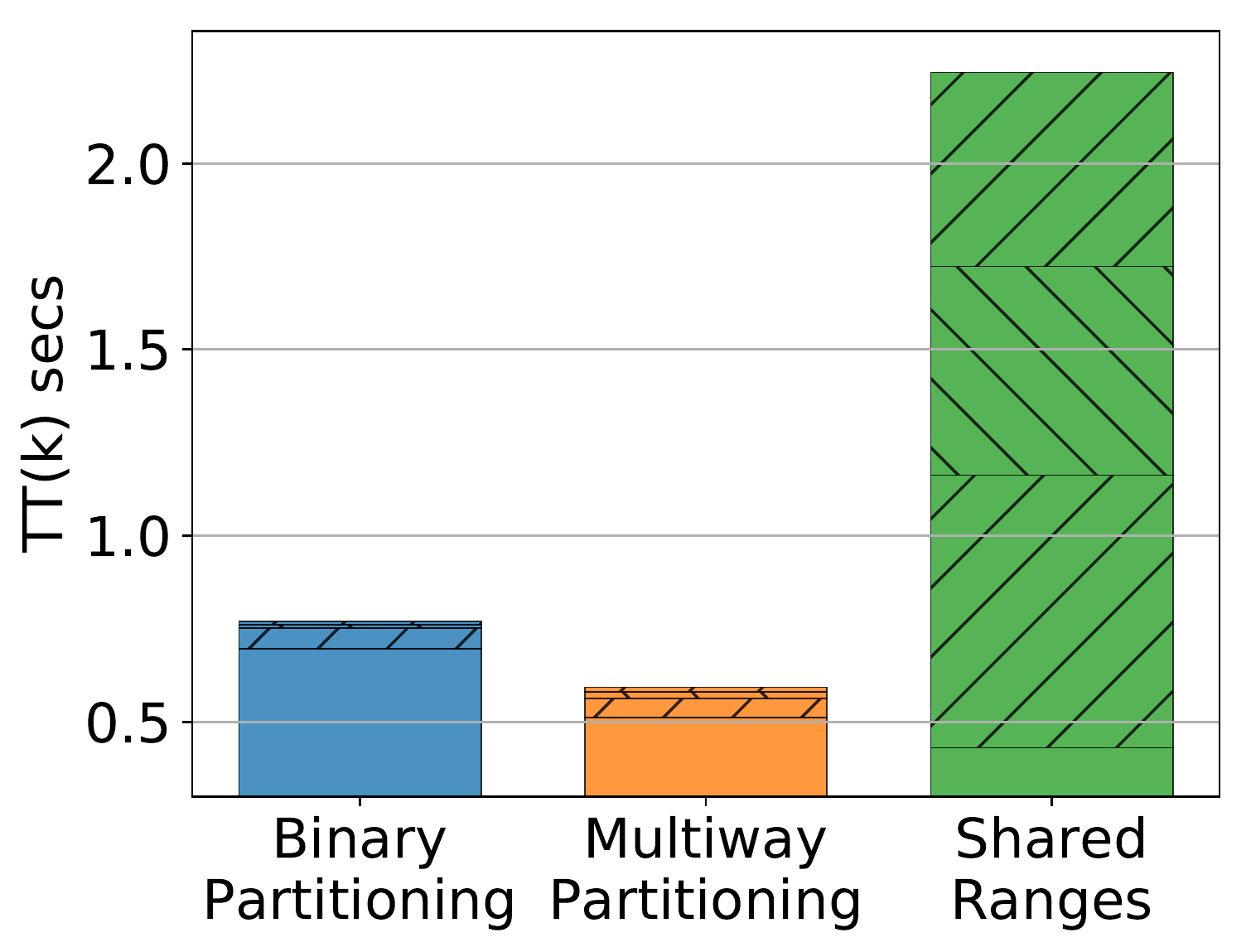}
        \caption{$\TT(k)$ for $n=2^{16}$. After preprocessing (no pattern), each bar represents $10^4$ results (alternating pattern).
		}
		\label{exp:fact_ttk}
    \end{subfigure}%
    \hfill
    \begin{subfigure}[t]{0.24\linewidth}
        \centering
        \includegraphics[width=\linewidth]{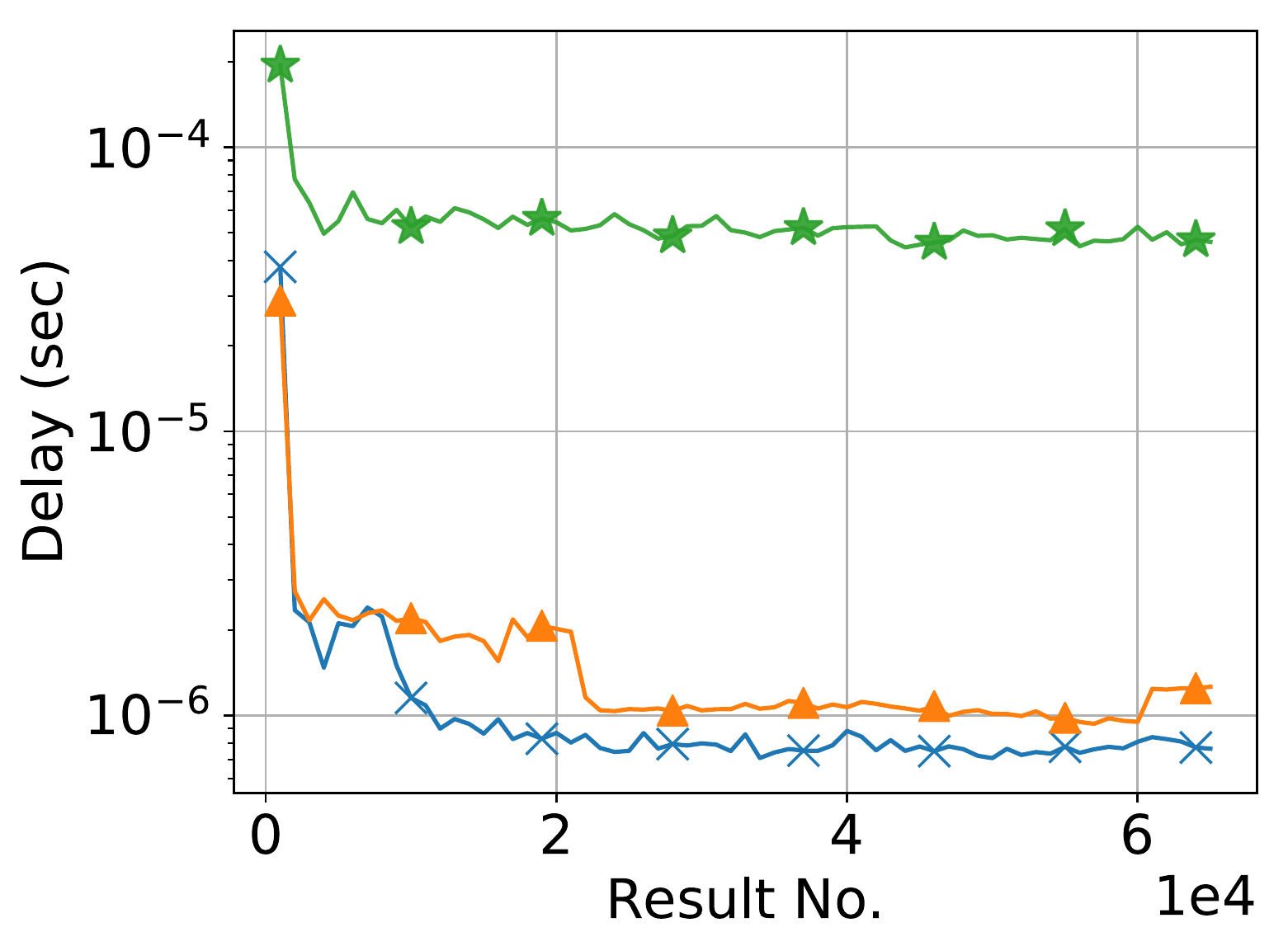}
        \caption{Delay between the first $n$ consecutive results for $n=2^{16}$. The delay is averaged in a window of size $10^3$.}
		\label{exp:fact_delay}
    \end{subfigure}%
    \hfill
        \begin{subfigure}[t]{0.24\linewidth}
        \centering
        \includegraphics[width=\linewidth]{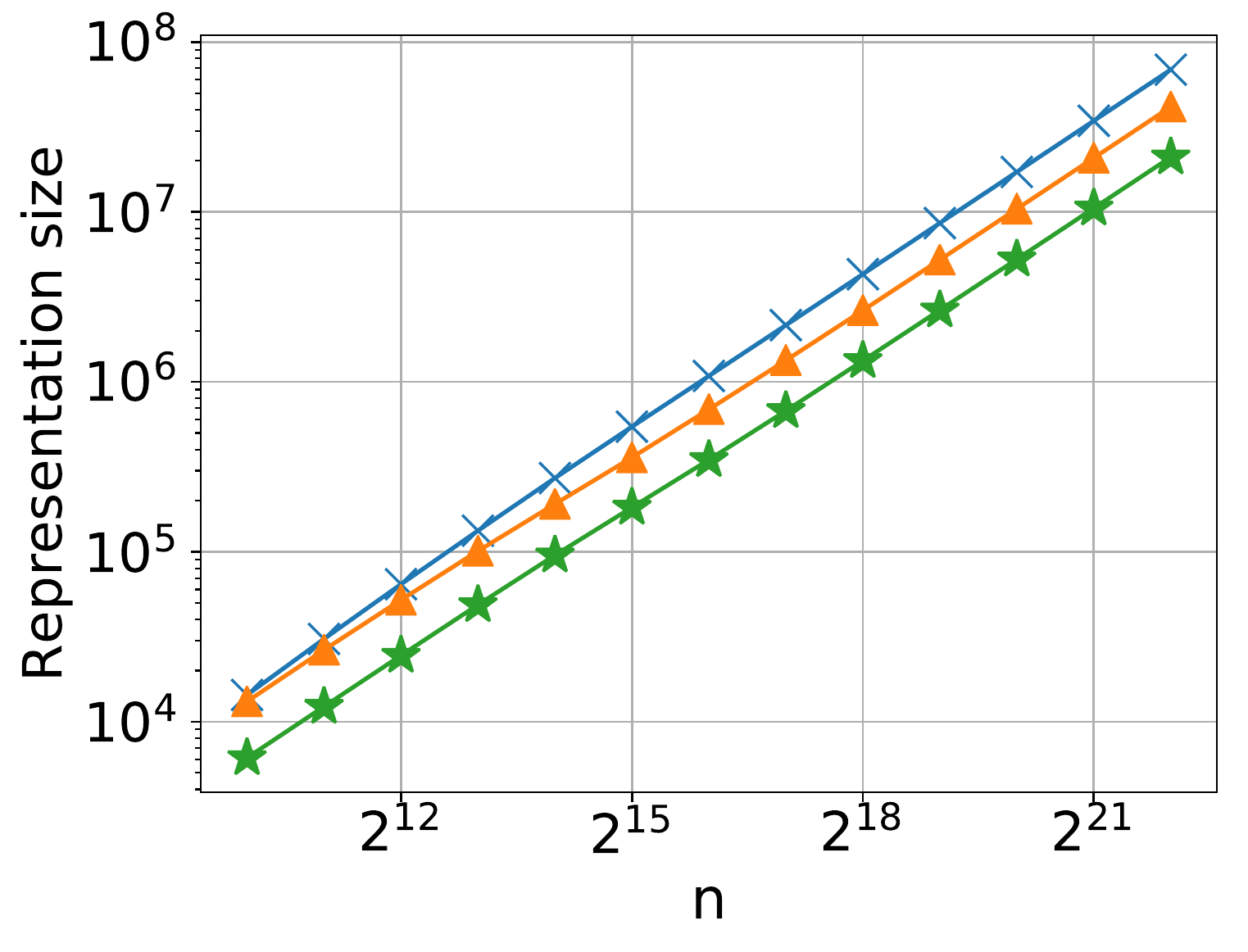}
        \caption{Size of the constructed \TLFG for increasing $n$. Measured as the total number of nodes and edges in the graph.}
		\label{exp:fact_mem}
    \end{subfigure}%
    \hfill
    \begin{subfigure}[t]{0.24\linewidth}
        \centering
        \includegraphics[width=\linewidth]{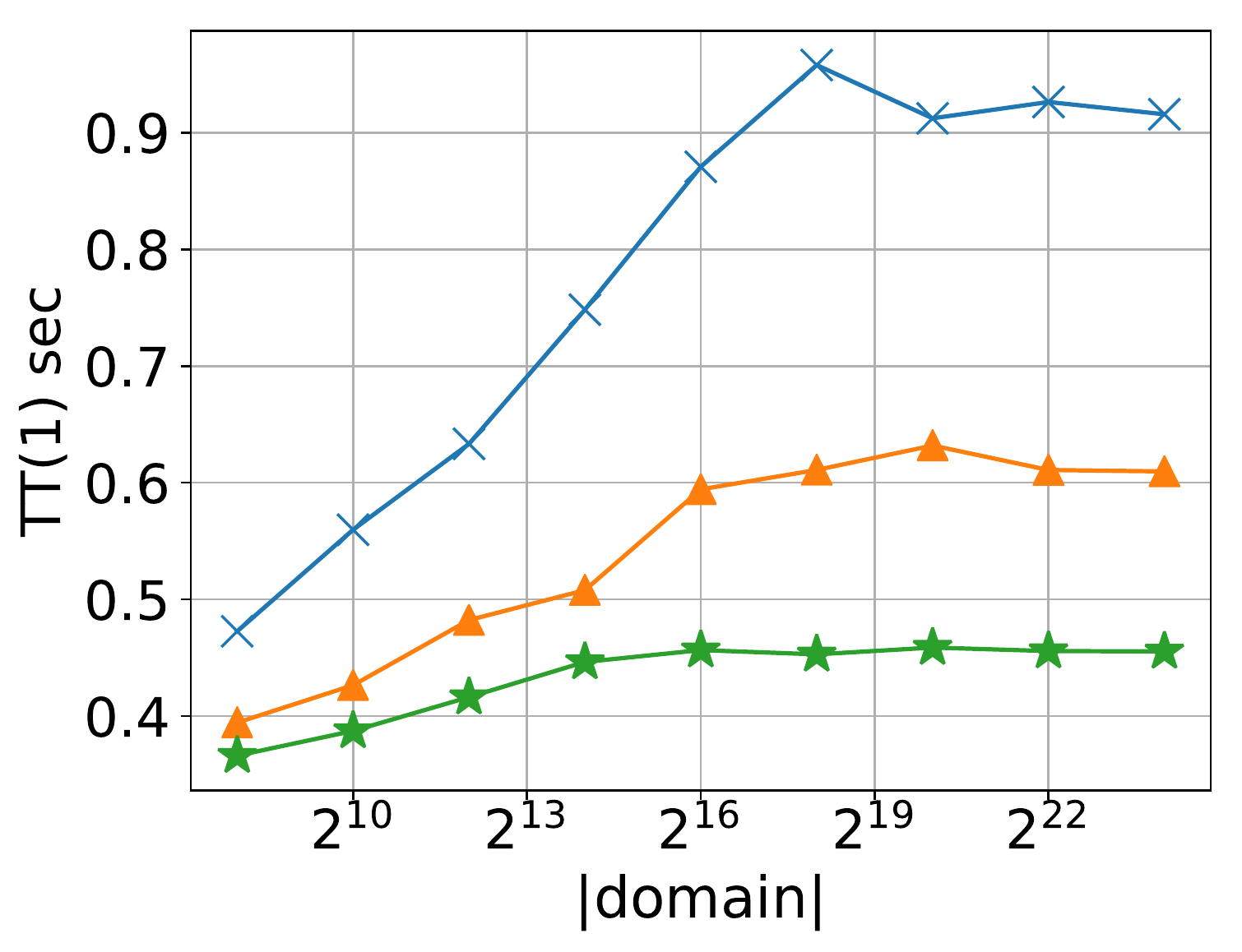}
        \caption{Time to the first result for $n=2^{16}$ and different domain sizes. The tuple values are sampled randomly.}
		\label{exp:fact_dom}
    \end{subfigure}%
    \vspace{-1mm}
    
    \caption{\Cref{sec:exp_fact}: Comparing different aspects of our factorization methods on query $Q_{S1}, \ell=2$.}
    \label{exp:fact}
\end{figure*}

\subsubsection{Effect of Band Parameter}\label{subsec:band}

We now test the band-join $Q_{B}$ on the \Birds dataset 
with various band widths $\epsilon$.
\Cref{exp:birds_e001} shows that \OURS is superior for
all tested $k$ values for $\epsilon=0.01$.
Increasing the band width yields more joining pairs and causes the size of the output to grow (\Cref{exp:birds_e}).
Hence, \BATCH consumes more memory and cannot handle $\epsilon \geq 0.16$.
On the other hand, the performance of \OURS is mildly affected by increasing $\epsilon$.
\PSQL and \SYSX were not able to terminate within the time limit even for the smallest $\epsilon$
because they use only one of the indexes 
(for \texttt{Longitude}),
searching over a huge number of possible results.

\subsection{Comparison of our Variants}
\label{sec:exp_fact}

We now compare our 3 factorization methods 
\circled{1a}, 
\circled{1b},
\circled{1c}.  

\subsubsection{Delay and $\TT(k)$}
\label{sec:exp_fact_delay}

Since only \BINPART is applicable to all types of join conditions considered,
we compare the different methods on $Q_{S1}$, which has only one inequality-type predicate.
\Cref{exp:fact_ttk} 
depicts $\TT(k)$ for $k = 1, 10^4, 2 \! \cdot \! 10^4, 3 \! \cdot \! 10^4$.
Even though \SHAREDRAN starts returning results faster because its \TLFG is constructed in a single pass (after sorting),
it suffers from a high enumeration delay (linear in data size),
and quickly deteriorates as $k$ increases.
The delay is also depicted in \Cref{exp:fact_delay},
where we observe that \BINPART returns results with lower delay than \MULTIPART
(recall that \MULTIPART has a depth of $3$ vs \BINPART's $2$).
These results are a consequence of the size-depth tradeoff of the \TLFGs (\cref{fig:sizeTradeoff}).
Note that the higher delay observed in the beginning is due to lazy initialization of data structures needed by the any-$k$ algorithm.

\subsubsection{Join Representation}
\label{sec:exp_fact_representation}

We show the sizes of the constructed representation in \Cref{exp:fact_mem}, using an implementation-agnostic measure.
As $n$ increases there is an asymptotic difference between the three methods 
($\O(n \log n)$ vs $\O(n \log\log n)$ vs $\O(n)$)
that manifests in our experiment.
To see how the presence of the same domain values could affect the construction of the \TLFG,
we also measure the time to the first result for different domain sizes (\Cref{exp:fact_dom}).
All three of our methods become faster when the domain is small and multiple occurrences of the same value are more likely.
This is expected since the intermediate nodes of our \TLFG essentially represent ranges in the domain and they are more compact for smaller domains.
Domain size does not significantly impact running time once it
exceeds sample size (around $n=2^{16}$) and 
the probability of sampling duplicate domain values
approaches zero. 

\section{Related Work}

\introparagraph{Enumeration for equi-joins}
\emph{Unranked} enumeration for equi-joins has been explored in various contexts
\cite{Berkholz20tutorial,Berkholz:2017:ACQ:3034786.3034789,DBLP:conf/pods/CarmeliK19,DBLP:journals/mst/CarmeliK20,durand20tutorial,DBLP:journals/sigmod/Segoufin15},
with a landmark result showing for self-join-free equi-joins that
linear preprocessing and constant delay are possible if and only if
the query is free-connex acyclic~\cite{bagan07constenum,brault13thesis}.
For the more demanding task of ranked enumeration, a logarithmic delay is
unavoidable \cite{deep21,bremner06xy}. %
Our recently proposed any-$k$ algorithms represent the state of the art for
ranked enumeration for equi-joins \cite{tziavelis20vldb}.
Other work in this space focuses on
practical implementations~\cite{ding21progressive} and
direct access \cite{carmeli20random,carmeli20direct} to output tuples.

\introparagraph{Non-Equality ($\neq$) and inequality ($<$) joins}
Techniques for batch-computation of the entire output for
joins with \emph{non-equality}
(also called \emph{inequality} \cite{koutris17noneq} or \emph{disequality} \cite{bagan07constenum})
predicates mainly rely on variations of color coding~\cite{Alon:1995:COL:210332.210337,koutris17noneq,papadimitriou99complexity}.
The same core idea is leveraged by the unranked enumeration algorithm of
Bagan et al.\ \cite{bagan07constenum}.
Queries with negation can be answered by rewriting them with
\emph{not-all-equal-predicates} \cite{khamis19negation}, 
a generalization of non-equality.

Khayatt et al.~\cite{khayyat17ineq} provide optimized and distributed
\emph{batch} algorithms for up to two inequalities per join.
Aggregate computation \cite{AboKhamis:2019:FAQ:3294052.3319694} and
Unranked enumeration under updates \cite{idris20dynamic_theta} 
have been studied for inequality predicates by
using appropriate index structures.

We are the first to consider \emph{ranked} enumeration for non-equality and
inequality predicates, including DNF conditions containing both types, and to prove
strong worst-case guarantees for a large class of these queries.

\introparagraph{Orthogonal range search}
Our binary partitioning method shares a similar intuition with index structures that have been devised for 
orthogonal range search \cite{chazelle88range,Agarwal17}.
For unranked enumeration, it has been shown \cite{willard96range,willard02range,agarwal21dynamic} how, for two relations, a range tree \cite{de97geometry} can be used to identify pairs of matching tuple sets.
This gives an alternative method to construct our depth-2 \TLFGs because a pair of matching tuple sets can be connected via one intermediate node.
Our approach supports ranking and it is simpler since it does not require building a range tree.
Our \TLFG abstraction is also more general: our other representations (such as multiway partitioning) do not have any obvious representation as range trees.

\introparagraph{Factorized databases}
Factorized representations of query results
\cite{olteanu16record,bakibayev13fordering}
have been proposed for \emph{equi-joins} in the context of
enumeration \cite{olteanu12ftrees,olteanu15dtrees}, 
aggregate computation \cite{bakibayev13fordering},
provenance management \cite{olteanu11provenance,olteanu12ftrees,DBLP:journals/corr/abs-2105-14307}
and machine learning \cite{olteanu16ml,schleich16ml,khamis18ml,kumar15ml,DBLP:conf/sigmod/PLG20}.
Our novel \TLFG approach to factorization complements this line of research and
extends the fundamental idea of factorization to ranked enumeration for
theta-joins.
For probabilistic databases,
factorization of non-equalities \cite{olteanu08obdd} and 
inequalities \cite{olteanu09confidence} is possible with OBDDs.
Although these are for a different purpose,
we note that the latter exploits the transitivity of inequality,
as our \SHAREDRAN (\Cref{fig:Inequality_sharing})
and other approaches for aggregates do \cite{cluet95agg}.

\introparagraph{Top-$k$ queries}
Top-$k$ queries \cite{rahul19topk} are a special case of ranked enumeration 
where the value of $k$ is given in advance and its knowledge can be exploited.
Fagin et al. \cite{fagin03} present the Threshold Algorithm, which is
instance-optimal under a ``middleware'' cost model for a restricted class
of 1-to-1 joins. Follow-up work generalizes the idea to more general joins
\cite{ilyas04,mamoulis07lara,finger09frpa,wu10topk}, including
theta-joins \cite{natsev01}.
Since all these approaches focus on the middleware cost model, they
do not provide non-trivial worst-case guarantees when the join cost
is taken into account \cite{tziavelis20tutorial}.
Ilyas et al.~\cite{ilyas08survey} survey some of these approaches,
along with some related ones such as building 
top-$k$ indexes~\cite{chang00topk,tsaparas03topk} 
or views~\cite{hristidis01topk,das06topk}.

\introparagraph{Optimal batch algorithms for joins}
Acyclic equi-joins are evaluated optimally in $\O(n + \out)$
by the Yannakakis algorithm~\cite{DBLP:conf/vldb/Yannakakis81},
where $\out$ is the output size.
This bound is unattainable for cyclic queries~\cite{ngo2018worst},
thus worst-case optimal join algorithms
\cite{navarro19wco,
ngo2018worst,
Ngo:2014:SSB:2590989.2590991,
veldhuizen14leapfrog}
settle for the AGM bound~\cite{AGM}, i.e., the worst-case output size.
(Hyper)tree decomposition methods \cite{GottlobGLS:2016,khamis17panda,Marx:2013:THP:2555516.2535926}
can improve over these guarantees, 
while a geometric perspective has led to even stronger notions of optimality
\cite{ngo14mine,Khamis:2016:JVG:3014437.2967101}.
Ngo \cite{ngo18open} recounts the development of these ideas.
That line of work focuses on batch-computation, i.e., on
\emph{producing all the query results}, or on Boolean queries,
while we explore ranked enumeration.

\section{Conclusions and Future Work}
\label{sec:conclusion}

Theta- and inequality-joins of multiple relations are
generally considered ``hard'' and even
state-of-the-art commercial \DBMSs struggle with their efficient
computation. We developed the first ranked-enumeration techniques that achieve non-trivial worst-case guarantees
for a large class of these joins:
For small $k$, returning the $k$ top-ranked join answers for full acyclic queries takes only
slightly-more-than-linear time and space ($\O(n \polylog n)$)
for any DNF of inequality predicates. For general theta-joins,
time and space complexity are quadratic in input size.
These are strong worst-case guarantees,
close to the lower time bound of $\O(n)$
and much lower than the $\O(n^\ell)$ size of intermediate or
final results traditional join algorithms may have to deal with.
Our results apply to many cyclic joins (modulo higher pre-processing cost
depending on query width) and all acyclic joins, even those with
selections and many types of projections. In the future, we will study
parallel computation and more general cyclic joins and projections.
\begin{acks}
This work was supported in part by 
the National Institutes of Health (NIH) under award number R01 NS091421 and by
the National Science Foundation (NSF) under award numbers CAREER IIS-1762268
and IIS-1956096.
\end{acks}

\bibliographystyle{ACM-Reference-Format}
\balance
\bibliography{bibliography-anyk}

\clearpage
\appendix
\section{Nomenclature}

\begin{table}[h]
\centering
\small
\begin{tabularx}{\linewidth}{@{\hspace{0pt}} >{$}l<{$}  @{\hspace{2mm}}  X @{}}
\hline
\textrm{Symbol}	& Definition 	\\
\hline
Q               & Join query \\
R, S, T         & Relations \\
A, B, C         & Attributes \\
\vec{X, Y, Z}   & Lists of attributes \\
r, s, t         & Tuples \\
\theta          & Join Predicate \\
S \bowtie_\theta T & Join between $S, T$ on predicate $\theta$ \\
n               & Total number of tuples \\
\delta          & Number of distinct values \\
\ell            & Number of relations \\
q               & Number of predicates in the query \\
G(V, E)         & Graph with nodes $V$ and edges $E$ \\
v_s, v_t        & Nodes corresponding to tuples $s \in S, t \in T$ \\
\mathcal{S}     & Size of TLFG \\
\depth               & Depth of TLFG \\
u               & Duplication factor of TLFG \\
p               & Number of conjuncts or disjuncts \\
\rho               & Number of partitions in equality/inequality factorization \\
M_i             & Partition in inequality factorization \\
m               & Number of groups in band factorization \\
H_i             & Group in band factorization \\
\TT(k)          & Time-to-$k^{\text{th}}$ result \\
\MEM(k)         & Memory until the $k^{\text{th}}$ result \\
\mathcal{T}     & Time for constructing a TLFG \\
\Prep(n)        & Time for preprocessing \\
h               & Height of tree \\
f, g            & (Computable) functions \\
\hline
\end{tabularx}
\end{table}

\section{Delay vs $\TT(k)$ as Complexity Measure}

In this section, we discuss the relationship between delay and $\TT(k)$ as complexity measures for enumeration.
For unranked enumeration, 
our goal is to achieve 
$\TT(k) = \O(\Prep(n) + k)$ with the lowest possible preprocessing time $\Prep(n)$.
The majority of papers on enumeration
\cite{bagan07constenum,DBLP:journals/sigmod/Segoufin15,DBLP:conf/pods/CarmeliK19,idris20dynamic_theta}
have traditionally focused instead on \emph{constant delay} after $\Prep(n)$ preprocessing. 
This is desirable because it implies the same guarantee 
$\TT(k) = \O(\Prep(n)) + k \cdot \O(1) = \O(\Prep(n) + k)$.
However, setting constant delay as the goal
can lead to misjudgments about practical performance, as we illustrate next:

\begin{figure*}[t]
\centering
\includegraphics[width=0.7\linewidth]{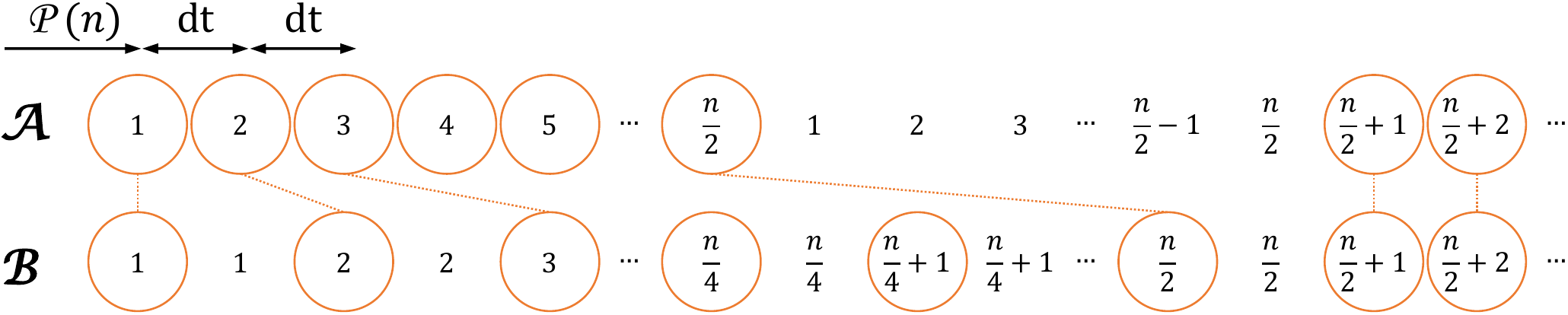}
\caption{Two enumeration algorithms with $\TT(k)=\O(\Prep(n)+k)$.}
\label{fig:TTKvsDelay}
\end{figure*}

\begin{example}
Consider an enumeration problem where the output consists of 
the integers $1, 2, \ldots, n$,
but algorithms produce duplicates that have to be filtered out
on-the-fly. Assume that two algorithms $\mathcal{A}$ and $\mathcal{B}$
spend preprocessing $\Prep(n)$, then generate a sequence of results
with constant delay. For $\mathcal{A}$, let this sequence be
$1, 2, \ldots, n/2, 1, 2, \ldots, n/2, n/2 + 1, \ldots$
and for $\mathcal{B}$ it is
$1, 1, 2, 2, \ldots, n/2, n/2, n/2 + 1, \ldots$ (see \cref{fig:TTKvsDelay}).
Even though both achieve $\TT(k) = \O(\Prep(n) + k)$, due to 
duplicate filtering the worst-case delay of $\mathcal{A}$ is $\O(n)$
(between $n/2$ and $n/2 + 1$),  while $\mathcal{B}$ has $\O(1)$ delay.
However, $B$ is clearly \emph{slower} than $A$ by a factor of
$2$ for all $k \in [n/2]$.
Since $\mathcal{A}$ outputs all these values earlier than $\mathcal{B}$,
we could make $\mathcal{A}$ \emph{simulate the delay} of $\mathcal{B}$
for $k \in [n/2]$ by storing the computed values on even iterations
and returning them later.
\end{example}

As the example illustrates, for a preprocessing cost of $\O(\Prep(n))$,
the ultimate goal is to guarantee $\TT(k) = \O(\Prep(n) + k)$.
Constant-delay enumeration is a sufficient condition for achieving this goal,
but not necessary.
Similarly, for ranked enumeration, we aim for $\TT(k) = \O(\Prep(n) + k \log k)$.

\section{Multiway Partitioning}

We provide more details on the multiway partitioning method
discussed in \cref{sec:fact_methods}.
Recall that it constitutes an improvement over the binary partitioning method of \cref{sec:inequality}
for the case of a single inequality predicate.
More specifically, it creates a \TLFG of size $\O(n \log\log n)$ instead of $\O(n \log n)$,
while only increasing the depth to $3$ from $2$ (see \cref{fig:sizeTradeoff}).

The main idea is to create more data partitions per recursive step. 
In particular, we pick $\partitions - 1$ pivots that create $\partitions$ partitions of nodes with a roughly equal number of distinct values.
\cref{fig:Inequality_many} depicts how the partitions are connected for a less-than ($<$) predicate. 
Each source partition $S_i, i \in [1, \rho - 1]$ is connected to all target partitions $T_j, j \in [i + 1, \rho]$,
since all values in $S_i$ are guaranteed to be smaller than all values in $T_j$.
The ideal number of partitions is $\Theta(\sqrt{d})$, 
so that the connections between them can be built in $\O(\sqrt{d}^2) = \O(n)$,
i.e., the same that binary partitioning needs per recursive step.
The advantage of the multiple partitions is that we can reach the base case $d = 1$ faster since each partition is smaller.
\Cref{alg:ineq} shows the pseudocode of this approach.

\begin{figure}[t]
\centering
\begin{subfigure}{.47\linewidth}
    \centering
    \includegraphics[width=\linewidth]{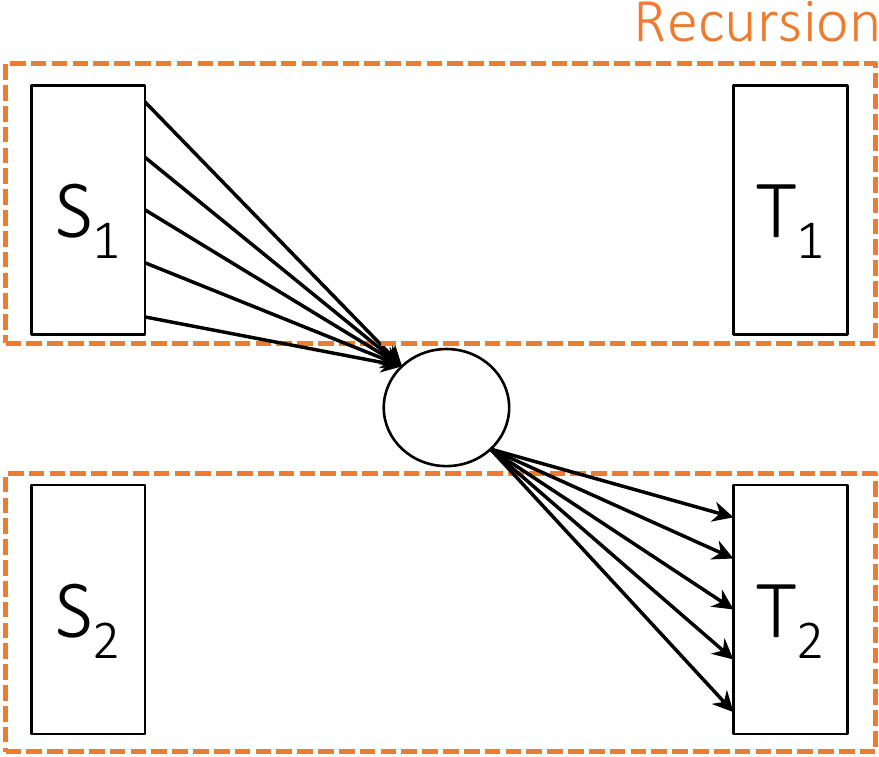}
    \caption{$2$ partitions: depth $2$ with one intermediate node.}
    \label{fig:Inequality_binary}
\end{subfigure}%
\hfill
\begin{subfigure}{.47\linewidth}
    \centering
    \includegraphics[width=\linewidth]{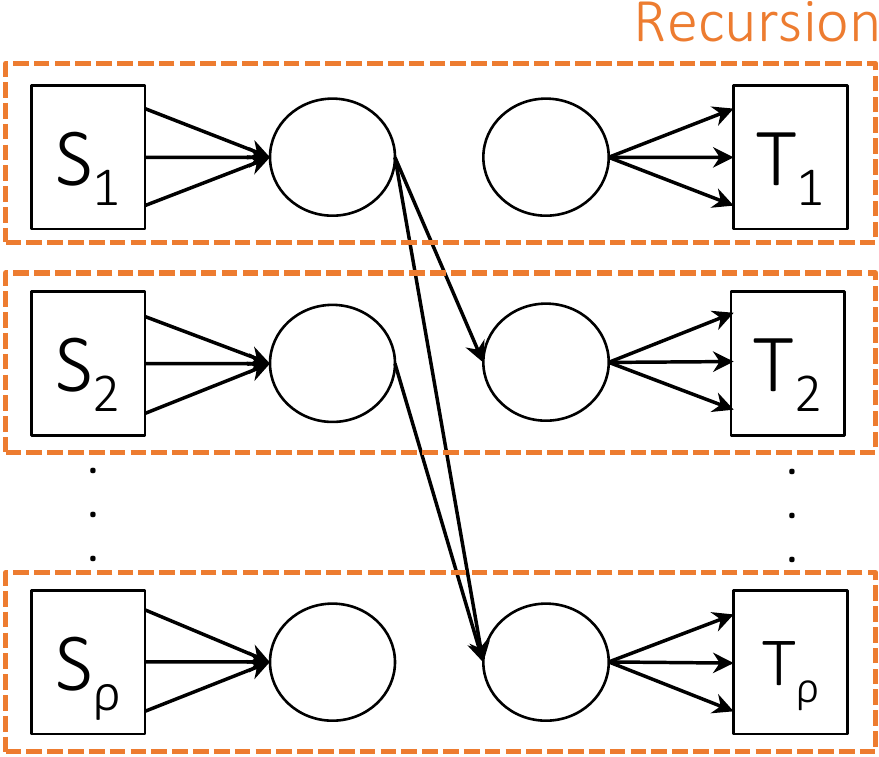}
    \caption{$\partitions$ partitions: depth $3$ with $2 \partitions$ intermediate nodes.}
    \label{fig:Inequality_many}
\end{subfigure}
\caption{Binary vs Multi-way partitioning for inequalities.}
\label{fig:splitting}
\end{figure}

\begin{algorithm}[tb]
\setstretch{0.85}   %
\small
\SetAlgoLined
\LinesNumbered
\SetKwFunction{ineq}{\IneqMultiFun}
\textbf{Input}: Relations $S, T$, nodes $v_s, v_t$ for $s \in S, t \in T$,\\\phantom{Inputt: }predicate $\theta \equiv S.A < T.B$\\
\textbf{Output}: A \TLFG of the join $S \bowtie_\theta T$\\
Sort $S, T$ according to attributes $A, B$\;
\ineq{$S, T, \theta$}\;

\SetKwProg{myproc}{Procedure}{}{}
\SetKwFunction{distinct}{vals}
\myproc{\ineq{$S, T, \theta$}}{
$d =$ \distinct{S $\cup$ T}\algocomment{Number of distinct A, B values}\;
\lIf{d == 1}{
\KwRet \algocomment{Base case} \label{alg_line:base_case_ineq}
}
$\partitions = \lceil \sqrt(d) \rceil$ \algocomment{Number of partitions}\;
Partition $(S \cup T)$ into $(S_1 \cup T_1), \ldots, (S_\partitions \cup T_\partitions)$ with $\rho$-quantiles of distinct values as pivots\;
\For{$i\gets1$ \KwTo $\partitions$}{
            Materialize intermediate nodes $x_i, y_i$\;
            \lForEach{$s$ in $S_i$}{
                    Create edge $v_s \longrightarrow x_i$ \label{alg_line:s_to_x}
            }    
            \lForEach{$t$ in $T_i$}{
                    Create edge $y_i \longrightarrow v_t$
                    \label{alg_line:y_to_t}
            }  
            \lFor{$j\gets1$ \KwTo $i-1$}{
                    Create edge $x_j \longrightarrow y_i$
                    \label{alg_line:x_to_y}
            }
            \ineq{$S_i, T_i, \theta$} \algocomment{Recursive call}\;
        }
}

\caption{Multiway partitioning}
\label{alg:ineq}
\end{algorithm}

\begin{lemma}
\label{lem:inequality}
Let $\theta$ be an inequality predicate
between relations $S, T$ of total size $n$.
A duplicate-free \TLFG of the join $S \bowtie_\theta T$
of size $\O(n \log\log n)$ and depth $3$ can be constructed in $\O(n \log n)$ time.
\end{lemma}
\begin{proof}
The arguments for correctness and the duplicate-free property 
are similar to the case of binary partitioning (\cref{lem:inequality_binary}).
For the depth, notice that all the edges we create are 
either from the source nodes to a layer of $x$ nodes (\cref{alg_line:s_to_x}) 
or from $x$ nodes to a layer of $y$ nodes (\cref{alg_line:x_to_y}) 
or from $y$ nodes to target nodes(\cref{alg_line:y_to_t}).
Thus, all paths from source to target nodes have a length of $3$.
The running time is dominated by the $\O(n \log n)$ initial sorting of the relations,
but the recursion (which bounds the space consumption) 
is now more efficient than the binary partitioning case. 
Each recursive step with size $|S| + |T| = n$ requires
$\O(n)$ to partition the sorted relations.
Then, we materialize $\O(n)$ edges for source and target nodes,
$\O(\sqrt{\delta})$ intermediate nodes and $\O(\sqrt{\delta}^2)$ edges between them.
This adds up to $\O(n)$ because $\delta \leq n$.
We then invoke $\partitions = \lceil \sqrt{\delta} \rceil = \O(\sqrt{n})$ recursive calls with sizes
$n_1 + n_2 + \ldots + n_\partitions = n$.
Therefore, in every level of the recursion tree, 
the sizes of all the subproblems add up to $n$.
Since we spend linear time per problem, the total work per level of the
recursion tree is $\O(n)$.
The height $h$ of the tree is the number of times we have to take the square root of $\delta$ (and then the ceil function) in order to reach $d=1$, 
which is $\O(\log \log \delta) = \O(\log \log n)$.
To see this, observe that
$d^{(\frac{1}{2})^h} = 2 \Rightarrow 
(\frac{1}{2})^h \log \delta = 1 \Rightarrow 
h = \log\log \delta$.
Overall, the time spent on the recursion and thus, the size of the \TLFG is bounded by $\O(n \log\log n)$.
\end{proof}

\section{Non-Equality Predicates}
\label{sec:non-equality_details}

A non-equality condition $S.A \neq T.B$ is satisfied if either
$S.A < T.B$ or $S.A > T.B$.
Even though it can be modeled as a disjunction of two inequalities,
we now establish that (in contrast to arbitrary disjunctions),
they do not increase the \TLFG duplication factor.
The main observation is that the pairs which satisfy one of the inequalities cannot satisfy the other one.
Therefore, if we union the two inequality \TLFGs no path will be duplicated.
The guarantees we obtain are the same as the inequality case by using multiway partitioning (once for each inequality).

\begin{lemma}
\label{lem:nonequality}
Let $\theta$ be an non-equality predicate
between relations $S, T$ of total size $n$.
A duplicate-free \TLFG of the join $S \bowtie_\theta T$
of size $\O(n \log\log n)$ and depth $3$ can be constructed in $\O(n \log n)$ time.
\end{lemma}
\begin{proof}
We sort once in $\O(n \log n)$ and then call the inequality multiway partitioning algorithm twice. 
Thus, we have to spend two times $\O(n \log \log n)$ time and space.
The depth of the final \TLFG is still $3$ since the two \TLFGs are constructed independently.
It also remains duplicate-free since the two inequality conditions cannot hold simultaneously.
Suppose that the calls to $\mathrm{\IneqMultiFun(S, T, S.A < T.B)}$ and 
$\mathrm{\IneqMultiFun(S, T, S.A > T.B)}$ both create a path
between $v_s$ and $v_t$ for two tuples $s \in S, t \in T$.
Then, the two tuples would have to satisfy $s.A < t.B$ and $s.A > t.B$, which is impossible.
\end{proof}

\section{Band Predicates}
\label{sec:band_details}

In this section, we target band predicates of the type $|S.A - T.B| < \epsilon$.
We provide an algorithm that leverages the structure of the band
to achieve asymptotically the same guarantees as the inequality case.
If a band condition is handled as a generic conjunction of inequalities,
then the time spent, as well as the \TLFG size are higher than our specialized construction.

Our algorithm translates the band problem 
into a set of inequality problems for smaller groups of tuples, 
which can then be solved independently.
First, we describe the intuition.
The band predicate consists of two inequalities
$(S.A < T.B + \epsilon)$ and $(S.A > T.B - \epsilon)$ that need to hold simultaneously.
If for some source-target tuples we can guarantee that one of the two inequalities is always satisfied, 
then it suffices to use the inequality algorithm we developed in \cref{sec:inequality} for the other one.
Therefore, the idea is to create groups of tuples with that property
and cover all the possible joining pairs with these groups.

\begin{figure}[t]
\centering
\begin{subfigure}{.42\linewidth}
    \centering
    \includegraphics[width=\linewidth]{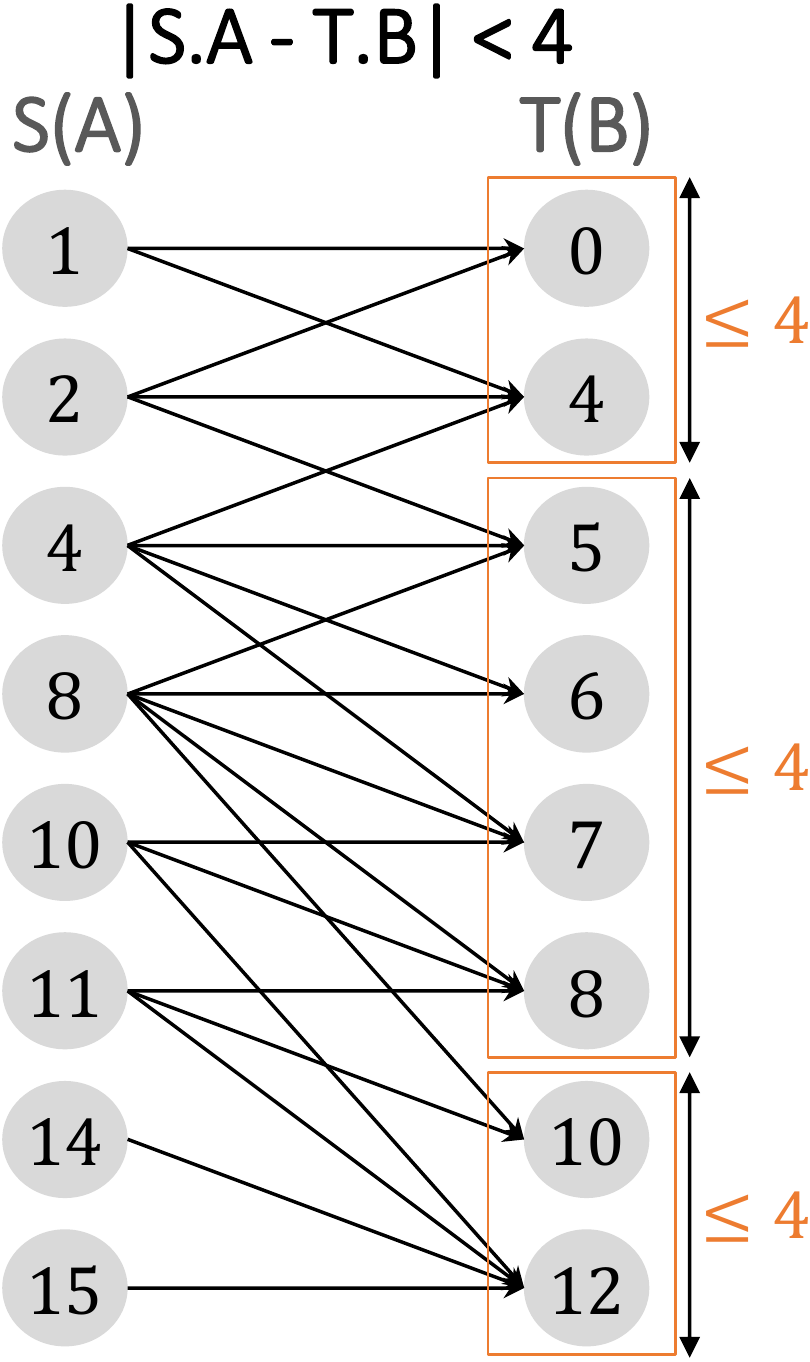}
    \caption{Edges between all $\O(n^2)$ joining pairs and grouping with $\epsilon$-intervals.}
    \label{fig:Band_all_groups}
\end{subfigure}%
\hfill
\begin{subfigure}{.47\linewidth}
    \centering
    \includegraphics[width=\linewidth]{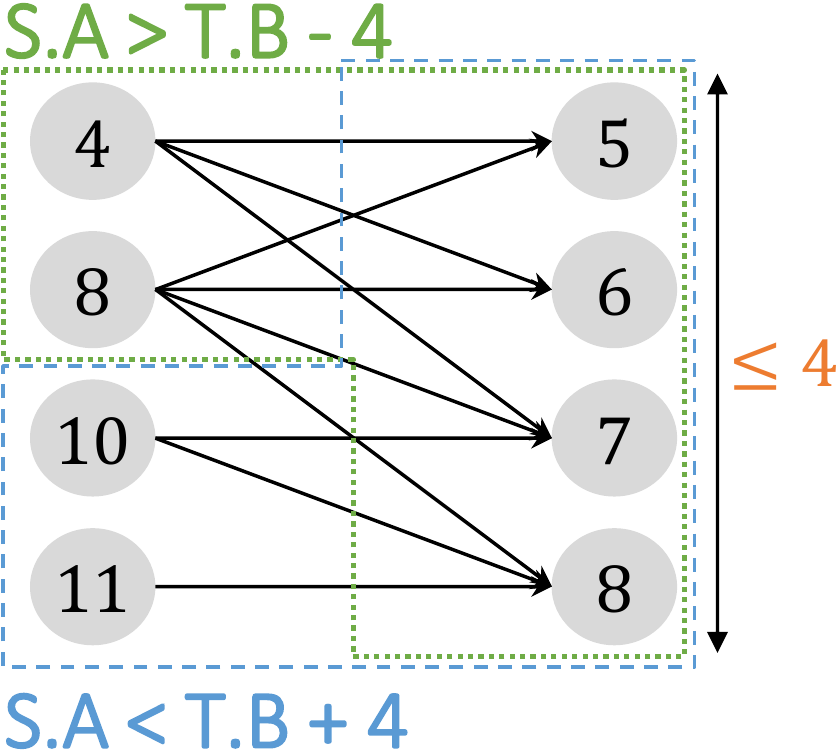}
    \caption{Edges within a group can be modeled as two inequalities.}
    \label{fig:Band_inequalities}
\end{subfigure}
\caption{\Cref{ex:band}: \TLFG construction for band conditions.}
\label{fig:Band}
\end{figure}

The first step is to sort the input relations and group the tuples of the target relation 
into maximal $\epsilon$-intervals.
More specifically, we start from the first $T$ tuple and
group together all those whose $B$ values are at most $\epsilon$ apart from it.
We then repeat the same process starting from the $T$ tuple that is immediately after the group, 
creating $m \leq n$ groups, whose range of $B$ values is at most $\epsilon$.
A source tuple is assigned to a group if it joins with at least one target tuple in the group.
Since the groups represent $\epsilon$-intervals of target tuples, 
\emph{each source tuple can be assigned to at most three groups}.

\begin{example}
\label{ex:band}
\Cref{fig:Band} depicts an example with $\epsilon = 4$.
Notice that as the number of tuples grows, the output is $\O(n^2)$,
e.g., if the domain is fixed or if $\epsilon$ grows together with the domain size.
Initially, we group the target tuples by $\epsilon$ intervals (\cref{fig:Band_all_groups}).
Thus, the first group starts with the first $T$ tuple $0$,
and ends before $5$ since $5 - 0 > \epsilon = 4$.
This process creates three groups of target tuples, each one having a range of $B$ values bounded by $4$.
Then, a source tuple is assigned to a group by comparing its $A$ value with the limits of the group.
For instance, tuple $11$ is assigned to the middle group because $5 - 4 < 11 < 8 + 4$, 
hence it joins with at least one target tuple in that group.

After the assignment of tuples to groups, we work on each group separately.
For example, consider the middle group depicted in \cref{fig:Band_inequalities}.
Source tuple $4$ joins with the top $T$ tuple $5$, which means that the pair $(4, 5)$ satisfies both inequalities.
From that we can infer that $4$ \emph{satisfies the less-than inequality with 
all the target tuples in the group}, 
since their $B$ values are at least $5$.
Thus, we can handle it by using our inequality algorithm for the greater-than condition $(S.A > T.B - \epsilon)$.
Conversely, tuple $10$ joins with the bottom $T$ tuple $8$, 
thus satisfies the greater-than inequality with all the target tuples in the group.
For that tuple, we only have to handle the less-than inequality $(S.A < T.B + \epsilon)$.
Notice that all the source tuples in the group are covered by at least one of the above scenarios.
\end{example}

For each group of source-target tuples we created, there are three cases for
the $S$ tuples:
($1$) those who join with the top target tuple but not the bottom,
($2$) those who join with the bottom target tuple but not the top,
($3$) those who join with all the target tuples.
These are the only three cases since by construction of the group,
the distance between the target tuples is at most $\epsilon$.
Case ($1$) can be handled as a greater-than \TLFG, 
case ($2$) as a less-than,
and case ($3$) as either one of them.
As \cref{alg:band} shows,
$\texttt{\IneqMultiFun}()$ 
is called twice for each group.

\begin{lemma}
\label{lem:band}
Let $\theta$ be a band predicate
between relations $S, T$ of total size $n$.
A duplicate-free \TLFG of the join $S \bowtie_\theta T$
of size $\O(n \log\log n)$ and depth $3$ can be constructed in $\O(n \log n)$ time.
\end{lemma}
\begin{proof}
First, we create disjoint $T$ groups based on $\epsilon$-intervals 
and assign each $S$ tuple to all groups where it has joining partners
(\crefrange{alg_line:band_groups_start}{alg_line:band_groups_end}).
This can be done with binary search in $\O(n \log n)$.
Each $T$ tuple is assigned to a single group.
An $S$ tuple cannot be assigned to more than three consecutive groups 
since their values span a range of at least $2\epsilon$.
Within each group $H_j=(S_j \cup T_j)$, 
the correctness of our algorithm follows from the fact that 
the $T_j$ tuples are at most $\epsilon$ apart on the $B$ attribute.
Since all the assigned $S_j$ tuples have at least one joining partner in $T_j$, 
they have to join either with the first $T_j$ tuple (in sorted $B$ order) or with the last one.
Recall that the band condition can be rewritten as 
$(S.A < T.B + \epsilon) \wedge (S.A > T.B - \epsilon)$, 
i.e., two inequality conditions that both have to be satisfied.
In case some $s \in S_j$ joins with the first $T_j$ tuple, 
then we know that the less-than condition is always satisfied for $s$ within the group $H_j$.
Thus, we just need to connect $v_s$ with all $v_t$ for $t \in T_j$ that satisfy the greater-than condition.
We argue similarly for the case when $s$ joins with the last tuple of $T_j$, 
where we have to take care only of the less-than condition.
Finally, there is also the possibility that $s$ joins with all $T_j$ tuples.
In that case, both inequality conditions are satisfied -- 
we assign those tuples to only one of the inequalities 
which ensures the duplicate-free property. 
For the running time, the total size of the groups we create is
$n_1 + n_2 + \ldots + n_m \leq 3n$.
If for a problem of size $|S|+|T| = n$ where the relations have been sorted, $\mathcal{T}_B(n)$ is the time for factorizing a band condition 
and $\mathcal{T}_I(n)$ for an inequality,
we have $\mathcal{T}_B(n) = \O(n) + 2\mathcal{T}_I(n_1) + 2\mathcal{T}_I(n_2) + \ldots + 2\mathcal{T}_I(n_m)$,
since we call the inequality algorithm twice within each group.
For $\mathcal{T}_I(n) = \O(n \log\log n)$, we get $\mathcal{T}_B(n) = \O(n \log\log n)$, which also bounds the size of the \TLFG.
Each call to the inequality algorithm involves different $S, T$ pairs, 
giving us the duplicate-free property and the same depth as the inequality \TLFG.
\end{proof}

\begin{algorithm}[tb]
\setstretch{0.85}   %
\small
\SetAlgoLined
\LinesNumbered
\SetKwFunction{ineq}{\IneqMultiFun}
\SetKwFunction{band}{\BandFun}
\textbf{Input}: Relations $S, T$, nodes $v_s, v_t$ for $s \in S, t \in T$,\\\phantom{Inputt: }predicate $\theta \equiv |S.A - T.B| < \epsilon$\\
\textbf{Output}: A \TLFG of the join $S \bowtie_\theta T$\\
Sort $S, T$ according to attributes $A, B$\;
\ForEach{($S_b$, $T_b$, $\theta_b$) in \band($S$, $T$, $\theta$)}
{
    \ineq{$S_b$, $T_b$, $\theta_b$}\;
}

\SetKwProg{fn}{Function}{}{}
\fn{\band{$S, T, \theta$}}{

ineqs = $[]$\;
\algocomment{Find the limits of the groups on the right}
\label{alg_line:band_groups_start}\;
$H_1.\mathrm{start} = t_1.B, \; m = 1$\;
\For{$i\gets1$ \KwTo $|T|$}{
        \If{$t[i].B > H_m.\mathrm{start} + \epsilon$}{
            $H_m.\mathrm{end} = T[i-1].B$\;
            $m++$\;
            $H_m.\mathrm{start} = T[i].B$
        }
}
$H_m.\mathrm{end} = T[i].B$
\label{alg_line:band_groups_end}\;
\ForEach{$H_j$ in $[H_1, \ldots, H_m]$}{
    \algocomment{Assign tuples to the group}\;
    $S_j = [s \in S \;|\; H_j.\mathrm{start} - \epsilon \leq s.A \leq \mathrm{H_j.end} + \epsilon$]\;  
    $T_j = [t \in T \;|\; H_j.\mathrm{start} \leq t.B \leq \mathrm{H_j.end}]$\;
    \algocomment{Greater-than inequality}\;
    $S_> = [s \in S_j \;|\; s.A < H_j.\mathrm{start} + \epsilon ]$\;
    ineqs.add(($S_<, T_j, S.A > T.B - \epsilon$))\;
    \algocomment{Less-than inequality}\;
    ineqs.add(($S_j - S_>, T_j, S.A < T.B + \epsilon$))\;
    \KwRet ineqs\;
}
}

\caption{Handling a band predicate}
\label{alg:band}
\end{algorithm}

\section{Additional Proofs}

\subsection{Proof of \cref{thm:genericComplexity}}

Since each \TLFG that is in-between two relation layers
has $\O(|E|)$ edges and $\O(\lambda)$ layers, 
the \dpgraph has $\O(|E|)$ edges and $\O(\lambda)$ layers as well.
That is because the number of relation layers is $\ell$,
which is considered to be constant.
The theorem follows by applying \Cref{lem:acyclicEquiJoinComplexity} on the resulting \dpgraph.

\subsection{Proof of \cref{lem:inequality_binary}}

Correctness is easy to establish by induction: 
each recursive step
connects precisely the joining pairs between the two partitions
and the graph within each partition is correct inductively.
For the running time, we begin by sorting the relations in $\O(n \log n)$.
We analyze the recursion in terms of its recursion tree.
Each recursive step with size $|S| + |T| = n$ requires
$\O(n)$ to partition the sorted relations.
Then, we materialize one intermediate node and for each source and target node at most one edge.
We then invoke $2$ recursive calls with sizes
$n_1 + n_2 = n$.
Therefore, in every level of the recursion tree, 
the sizes of all the subproblems add up to $n$.
Since we spend linear time per recursive step, 
the total work per level of the
recursion tree is $\O(n)$.
We always cut the distinct values (roughly) in half,
thus the height $h$ of the tree is 
$\O(\log d)= \O(\log n)$.
Overall, the time spent on the recursion is $\O(n h)=\O(n \log n)$, 
which also bounds the size of the \TLFG.
Across all recursive steps, edges are created 
either from source nodes to intermediate nodes
or from intermediate nodes to target nodes.
Thus, the length of all paths from source to target nodes is $2$.
The invariant property which ensures that the \TLFG is duplicate-free is that
whenever a recursive step is called on a set of $S', T'$ tuples,
no path exists between $v_{s'}$ and $v_{t'}$ for $s' \in S', t' \in T'$.

\subsection{Handling equality predicates in a conjunction}

\begin{lemma}
\label{lem:equalities}
Let $\theta$ be a conjunction of predicates between relations $S, T$ of total size $n$,
and $\theta'$ be that conjunction with all the equality predicates removed.
If for $S', T'$ with $|S'|+|T'| = n'$ we can construct a \TLFG 
of the join $S' \bowtie_{\theta'} T'$ 
of size $\O(f(n'))$,
depth $\depth$, 
and duplication factor $u$
in time $\O(g(n'))$, 
and $f, g$ are superadditive functions,
then we can construct a \TLFG of the join $S \bowtie_{\theta} T$
of size $\O(f(n))$,
depth $\depth$, 
and duplication factor $u$
in time $\O(g(n) + n)$.
\end{lemma}
\begin{proof}
To construct the \TLFG for $S \bowtie_{\theta} T$, 
we gather all the equality predicates and use hashing to create partitions 
of tuples that correspond to equal joining values for the equality predicates.
This takes $\O(n)$.
We then construct the \TLFG for each partition independently 
with the conditions $\theta'$ through some algorithm $\mathcal{A}$.
If $\mathcal{A}$ elects to connect two nodes, 
then they satisfy both $\theta'$, 
and also the equalities since they belong to the same partition. 
Conversely, two nodes that remain disconnected at the end of the process 
either do not belong to the same equality partition 
or were not connected by $\mathcal{A}$, thus do not satisfy $\theta'$.

Assume that the number of tuples in each partition 
is $n_i, i \in [\partitions]$ with $n_1 + \ldots + n_\partitions = n$.
The total time spent on each partition is $\O(g(n_1) + \ldots + g(n_\partitions))$
which by the superadditivity property of $g$ is $\O(g(n_1 + \ldots + n_\partitions)) = \O(g(n))$.
The same argument applies to the size, giving us $\O(f(n))$.
Since the partitions are disjoint, we cannot create additional duplicate paths 
apart from the ones created by $\mathcal{A}$, 
or increase the depth of each \TLFG.
\end{proof}

\subsection{Proof of \cref{lem:conjunction}}

As a first step, all the equality predicates are handled by \cref{lem:equalities}.
Since the time and size guarantees we show are $\O(n \log^p n)$ 
and $n \log^p n$ is a superadditive function,  
they are unaffected by this step.
The remaining inequality predicates are handled by \cref{alg:conj}.
We denote by $\mathcal{T}_I(n, p)$ the running time for  
$n$ tuples and $p$ inequality predicates.
We proceed by induction on the number of predicates $p$ to show that 
$\mathcal{T}_I(n, p) \leq f(p) n \log^p n$ for some function $f$ and sufficiently large $n$.
First, assume that all the predicates are inequalities.
For the base case $p = 1$, the analysis is the same as in the proof of \cref{lem:inequality_binary}:
Tthe height of the recursion tree is $\O(\log n)$ and the total time is $\O(n \log n)$ together with sorting once.
In other words, we have $\mathcal{T}_I(n, 1) \leq c n \log n$ for sufficiently large $n$.
For the inductive step, we assume that $\mathcal{T}_I(n, p-1) \leq f(p-1) n \log^{p-1} n$.
The inequality at the head of the list creates a recursion tree where every node has a subset of the tuples $n'$ and calls the next inequality, 
thus is computed in $\mathcal{T}_I(n', p-1)$.
The problem sizes in some level of the tree add up to $n_1 + \ldots + n_\partitions = n$.
Thus, the work per level is bounded by 
$\mathcal{T}_I(n_1, p-1) + \ldots + \mathcal{T}_I(n_t, p-1)
\leq f(p-1) n_1 \log^{p-1} n_1 + \ldots + f(p-1) n_t \log^{p-1} n_t
\leq f(p-1) n \log^{p-1} n$.
The height of the tree is $\O(\log n)$, 
thus the total work in the tree is bounded by
$c' \log n f(p-1) n \log^{p-1} n = c' f(p-1) n \log^p n$.
We also take into account the time for sorting according to the attributes of the current inequality, 
which is bounded by $c'' n \log n$.
Thus, we get that $\mathcal{T}_I(n, p) \leq c' f(p-1) n \log^p n + c'' n \log n$.
If we pick a function $f$ such that $f(1) \geq c$
and $f(p) \geq c' f(p-1) + \frac{c''}{log^{p-1}n}$,
then $\mathcal{T}_I(n, p) \leq f(p) n \log^{p} n$.
This completes the induction, establishing that 
$\mathcal{T}_I(n, p) = \O(n \log^p n)$ in data complexity.

The size of the TLFG cannot exceed the running time, thus it is also $\O(n \log^p n)$.
The depth is $2$ because in all cases we use the binary partitioning method
and the duplication factor is $1$ because 
we only connect tuples in the base case of one predicate $p=1$,
which we already proved does not create duplicates (\cref{lem:inequality_binary}).

\subsection{Proof of \cref{lem:disjunctions}}

Correctness follows from the fact that the paths in the constructed 
\TLFG is the union of the paths in the \TLFGs for $S \bowtie_{\theta_i} T$.
For the depth, note that each $\theta_i$ is processed independently, 
thus the component \TLFGs do not share any nodes or edges other than the endpoints.
A path from $v_s$ to $v_t$ for $s \in S, t \in T$ may only be duplicated by
different \TLFG constructions since each one is duplicate-free.
Thus, the duplication factor cannot exceed the number of predicates $p$.

\subsection{Proof of \cref{lem:tlfg_improved}}

\Cref{lem:inequality,lem:band,lem:nonequality} together prove
\cref{lem:tlfg_improved}.

\subsection{Proof of \cref{th:enumeration_improved}}

For each edge of the theta-join tree, we construct a \TLFG by processing the join condition as a DNF formula.
The guarantees of the theorem follow from \Cref{thm:genericComplexity} by applying the properties of the \TLFGs we construct,
along with a duplicate elimination filter.

To construct each \TLFG, disjunctions are handled according to \cref{lem:disjunctions}
and for conjunctions, the proof is the same as that of \cref{lem:conjunction} with some changes:
we use ($1$) multiway partitioning for the base case of $p=1$ in the conjunction algorithm
and ($2$) specialized constructions for non-equalities and bands (see \Cref{lem:tlfg_improved}).
Equalities are removed from the conjunction because of
\cref{lem:equalities} and the fact that
$n \log^{p} n$ and
$n \log^{p-1} n \cdot \log\log n$ 
are superadditive functions.
In the conjunction algorithm, we use multiway partitioning for $p=1$ and binary partitioning for $p > 1$.
Therefore $T_I(n, 1) \leq c n \log\log n$,
resulting in $T_I(n, p) = \O(n \log^{p-1} n \cdot \log\log n)$ overall.
Non-equalities and bands are translated into inequalities by using the techniques we developed in \cref{sec:non-equality_details,sec:band_details}:
a non-equality results into two inequalities on the same sets of nodes,
while a band creates multiple inequality subproblems.
We use the same arguments as in the proofs of \cref{lem:nonequality,lem:band}.
We denote by $\mathcal{T}_I(n, p), \mathcal{T}_N(n, p), \mathcal{T}_B(n, p)$ the running time for  
$n$ tuples and $p$ predicates
when the head of the list of predicates is an inequality, non-equality or band respectively. 
$\mathcal{T}_N(n, p) = \O(\mathcal{T}_I(n, p) + \mathcal{T}_I(n, p))$ and 
$\mathcal{T}_B(n, p) = \O(n) + 2\mathcal{T}_I(n_1, p) + 2\mathcal{T}_I(n_2, p) + \ldots + \mathcal{T}_I(n_m, p)$ for
$n_1 + n_2 + \ldots + n_m \leq 3n$.
By these formulas, and since 
$\mathcal{T}_I(n, p) = \O(n \log^{p-1} n \cdot \log\log n)$,
it is easy to show the same bound for the other two.
This proves the space consumption of the \TLFGs,
thus the space bound of the theorem.

As we are enumerating subtrees of the \dpgraph in order,
we detect those that correspond to duplicate query results
and filter them out using a lookup table.
The duplication factor of our \TLFGs is $1$, except if we have disjunctions (\cref{lem:disjunctions}).
Let $u_{max}$ be the maximum duplication factor among the constructed \TLFGs.
The number of ``duplicate'' query answers (that correspond to the same answer $q$ of $Q$) are bounded by 
$u_{max}^\ell$, where $\ell$ is the number of $Q$ atoms.
That depends only on the query size which we consider as constant, thus it is $\O(1)$.
If the time for each answer without the filtering is $\TT'(k)$,
then we have that $\TT(k) = \O(\TT'(k \cdot u_{max}^\ell)) 
= \O(\TT'(k))$,
since $u_{max}$ and $\ell$ are $\O(1)$.

\section{SQL Code for Queries used in Experiments}

Query $Q_{S1}$:

\begin{lstlisting}
SELECT   *, S1.W + S2.W as Weight
FROM     S1, S2
WHERE    S1.A2 < S2.A3
ORDER BY Weight ASC
\end{lstlisting}

Query $Q_{S2}$:

\begin{lstlisting}
SELECT   *, S1.W + S2.W as Weight
FROM     S1, S2
WHERE    ABS(S1.A2 - S2.A3) < 50 AND S1.A1 <> S2.A4
ORDER BY Weight ASC
\end{lstlisting}

Query $Q_{R1}$:

\begin{lstlisting}
SELECT   *, R1.Sentiment + R2.Sentiment as Weight
FROM     Reddit R1, Reddit R2
WHERE    R1.To = R2.From AND 
         R2.Timestamp > R1.Timestamp        
ORDER BY Weight ASC
\end{lstlisting}

Query $Q_{R2}$:

\begin{lstlisting}
SELECT   *, R1.Readability + R2.Readability as Weight
FROM     Reddit R1, Reddit R2
WHERE    R1.To = R2.From AND 
         R2.Timestamp > R1.Timestamp AND 
         R2.Sentiment < R1.Sentiment
ORDER BY Weight DESC
\end{lstlisting}

Query $Q_{B}$:

\begin{lstlisting}
SELECT   *, B1.IndivCount + R2.IndivCount as Weight
FROM     Birds B1, Birds B2
WHERE    ABS(B2.Latitude - B1.Latitude) < $\epsilon$ AND 
         ABS(B2.Longitude - B1.Longitude) < $\epsilon$
ORDER BY Weight DESC
\end{lstlisting}

\section{\TLFG Factorization Formulas}

Typically, factorization refers to the process of compacting 
an algebraic formula by factoring out common sub-expressions using the distributivity property \cite{crama11boolean}. 
Under that perspective, factorized databases \cite{olteanu16record}
represent the results of an equi-join efficiently, 
treating them as
a formula built with product and union.
Besides distributivity, d-representations \cite{olteanu15dtrees} 
replace shared sub-expressions with variables,
further improving succinctness through memoization \cite{dpv08book}.
Our \TLFGs directly give a representation of that nature, 
complementing known results on join factorization.
(Note that in addition to supporting joins with non-equality conditions,
in \TLFG the atomic unit of the formulas is a database tuple (hence Tuple-Level),
while in previous work on factorized databases it is an attribute value.)
We illustrate this with \Cref{ex:fact} below.

\begin{example}
\label{ex:fact}
Consider the inequality join $S \bowtie_{A<B} T$.
A naive \TLFG for some example relations $S, T$ is shown in \cref{fig:Inequality_all}.
The join results
can be expressed with the ``flat'' representation:
\begin{equation*}
\Phi = (1 \times 2) \cup (1 \times 3) 
    \cup \ldots \cup (1 \times 6)
    \cup (2 \times 3) \cup \ldots \cup (3 \times 4)
    \cup \ldots
\end{equation*}
where for convenience
we refer to tuples by their $A$ or $B$ value,
and $\times$ and $\cup$ denote Cartesian product
and union respectively.
The flat representation has one term for each query result, 
separated by the union operator.
In terms of the \TLFG, $\times$ corresponds to path concatenation,
and $\cup$ to branching.
To make the formula more compact,
we can factor out tuples that appear multiple times
and reuse common subexpressions by giving them a variable name.
Equivalently, the size of the \TLFG can be reduced if we introduce intermediate nodes,
making the different paths share the same edges.
Such a factorized representation is shown in 
\cref{fig:Inequality_partitioning}.
We can write the corresponding algebraic formula
by defining new variables $v_i, i \in [5]$ for the intermediate nodes:
\begin{align*}
&\Phi_3 = (1 \times v_1) \cup (2 \times v_2) \cup (2 \times v_3) \cup (3 \times v_3), 
\ldots,
(5 \times v_5)\\
& v_1 = (2 \cup 3), v_2 = (3), v_3 = (4 \cup 5 \cup 6),
\ldots,
v_5 = (6).
\end{align*}
Notice that the total size of these formulas is asymptotically 
the same as the \TLFG size.
\end{example}

\section{Application of the Technique to Unranked Enumeration}

As a side benefit, our techniques are also applicable to \emph{unranked}
enumeration (where answers can be returned in \emph{any} order) for joins with
inequalities, returning $k$ answers in $\O(n \polylog n + k)$.

\label{th:enumeration_improved}
Let $Q$ be a full acyclic theta-join query over a database $D$ of size $n$
where all the join conditions are DNF formulas of
equality, inequality, non-equality, and band predicates.
Let $p$ be the maximum number of predicates, excluding equalities,
in a conjunction of a DNF on any edge of the theta-join tree.
Ranked enumeration of the answers to $Q$ over $D$ can
be performed with $\TT(k) = \O(n \log^p n + k \log k)$.
The space requirement
is $\MEM(k) = \O(n \log^{p-1} n \cdot \log\log n+ k)$.

\begin{theorem}
\label{th:unranked_enumeration}
Let $Q$ be a full acyclic theta-join query over a database $D$ of size $n$
where all the join conditions are DNF formulas of
equality, inequality, non-equality, and band predicates.
Let $p$ be the maximum number of predicates, 
excluding equalities, in a conjunction 
of a DNF on any edge of the theta-join tree.
Enumeration of the answers to $Q$ over $D$
in an arbitrary order can
be performed with 
$\TT(k) = \O(n \log^p n + k)$
and
$\MEM(k) = \O(n \log^{p-1} n \cdot \log\log n+ k)$.
\end{theorem}

\subsection{Experimental Comparison}
\label{sec:exp_duplicates}

We use $Q_T^U$ for the query that is the same as $Q_T$, but without the ranking.
To illustrate how the duplicates from disjunctions or the presence of ranking
change the delay of the enumeration,
we plot $\TT(k)$ for query $Q_{T}$,
together with its disjunction $Q_{TD}$
and unranked $Q_{T}^U$ variants (\Cref{exp:syn_qt_msf32}).
For $Q_{TD}$ the constructed \TLFG is $~3$ times larger (because of the three date inequalities),
which is reflected in the time it starts to return results.
The delay is higher by a similar factor,
since the three predicates in the disjunction have a very high overlap.
In fact, that is the worst case for our technique because of the high number of duplicates that have to be filtered.
As illustrated in \Cref{exp:syn_qt_msf64}, this number is not affected by the size of the database and only depends on the query. 
Without the ranking, the enumeration for $Q_{T}^U$
starts slightly faster than $Q_T$
and has significantly lower delay between results. 

\begin{figure}[t]
    \centering   
    \begin{subfigure}{0.5\linewidth}
        \centering
        \includegraphics[height=2.8cm]{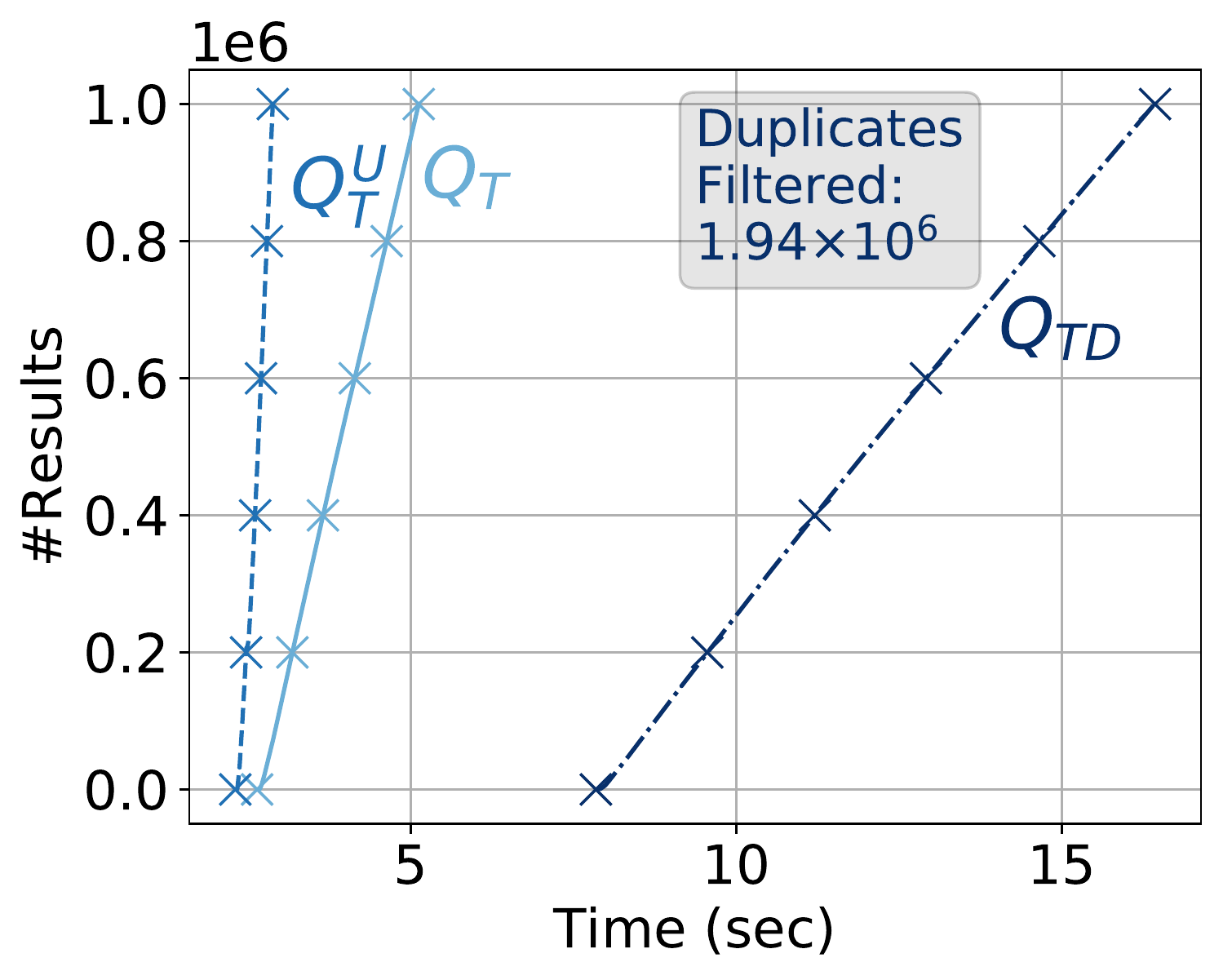}
        \caption{Scale factor $= 2^5 \times 10^{-3}, \ell = 2$.}
		\label{exp:syn_qt_msf32}
    \end{subfigure}%
    \hfill
    \begin{subfigure}{0.5\linewidth}
        \centering
        \includegraphics[height=2.8cm]{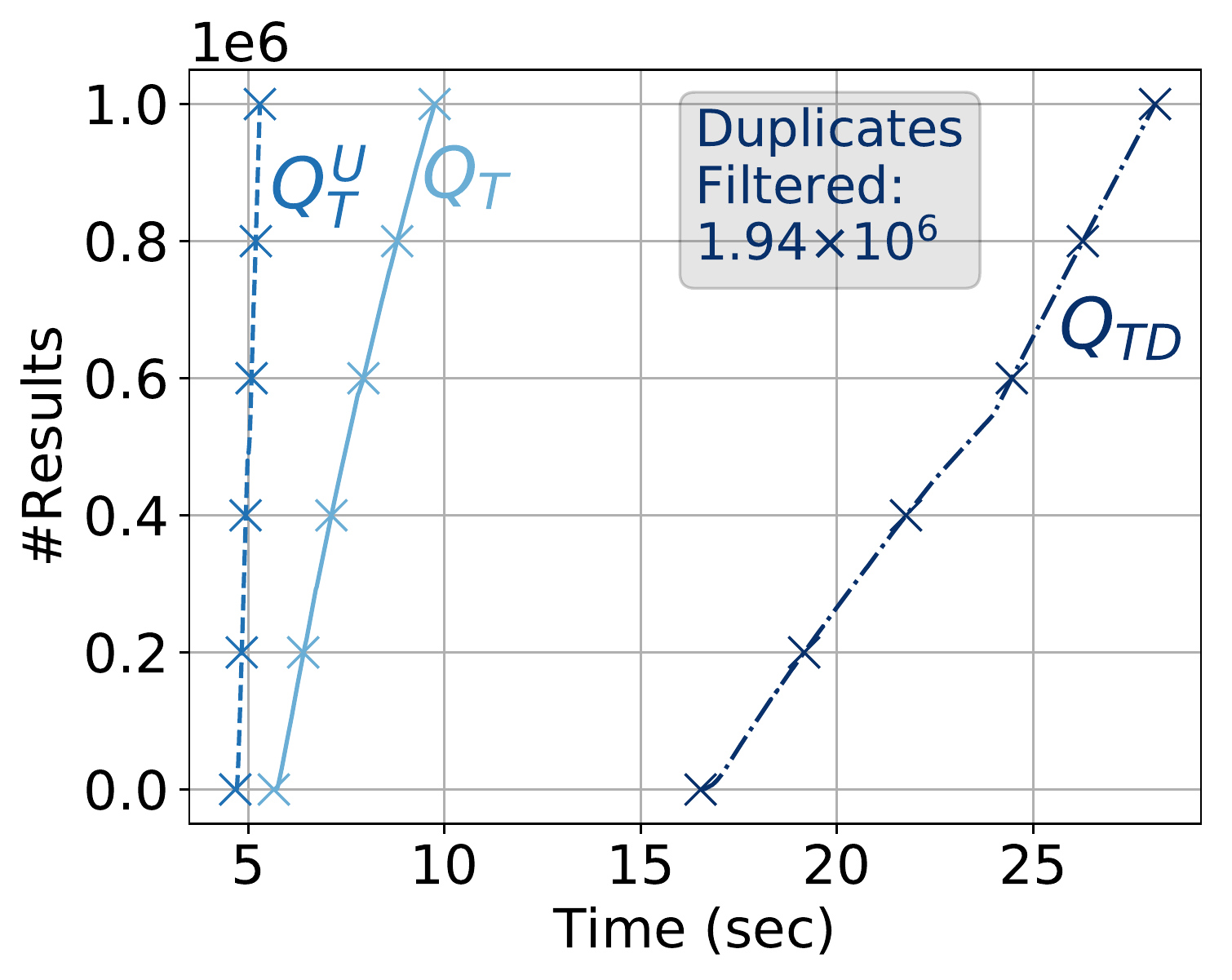}
        \caption{Scale factor $= 2^6 \times 10^{-3}, \ell = 2$.}
		\label{exp:syn_qt_msf64}
    \end{subfigure}
    \vspace{-1mm}
    \caption{\Cref{sec:exp_duplicates}: Variants of query $Q_T$ with disjunctions ($Q_{TD}$) and unranked enumeration ($Q_T^U$) on TPC-H data.}
    \label{exp:tpch_ttk}
\end{figure}

\section{Why the DBMS Top-k Plan Must Produce the Entire Output}

\definecolor{wtcl}{HTML}{5A9AD4}
\newcommand{\wtcl}[1]{{{\color{wtcl}{#1}}}}

In this section, we discuss why any approach that first applies the join
and then the ranking (e.g. with a heap over the join results)
will unavoidably spend $\O(n^2)$ even for a simple binary join with one inequality predicate.

First, we would like to emphasize that \textbf{we do not compare against a naive $\O(n^2)$ join algorithm}. The quadratic worst-case complexity is not caused by an inferior join algorithm but by the output size itself.
In short, even if we want to retrieve only $k$ join output tuples, the algorithm
has to insert $\O(n^2)$ output tuples into the heap: At any moment in time
(until the full output is known) the algorithm does not know if all of the
top-$k$ answers are already in the heap or if some of them will be emitted
by the join later.

We illustrate this with an example. 
Consider the inequality join in \Cref{fig:Inequality_Ex_Quadratic} with join condition $S.A < T.B$.
To efficiently find joining pairs, we can sort input relation $S$ on $A$ and $T$ on $B$. (Alternatively, one could use clustered B-tree indexes---one on $A$ for $S$ and the other on $B$ for $T$---to the same effect.)
This step indeed takes $\O(n \log n)$ time and it allows us to retrieve the joining pairs with a sort-merge type algorithm.
Using the sorted inputs, this algorithm can produce $k$ output tuples in time $\O(k)$. With $k$ upper bounded by some constant, say $k=3$, $k$ join answers can then indeed be retrieved in total time $\O(n \log n)$.

\begin{figure*}[h]
\begin{subfigure}[t]{.2\linewidth}
    \centering
    \includegraphics[height=5.2cm]{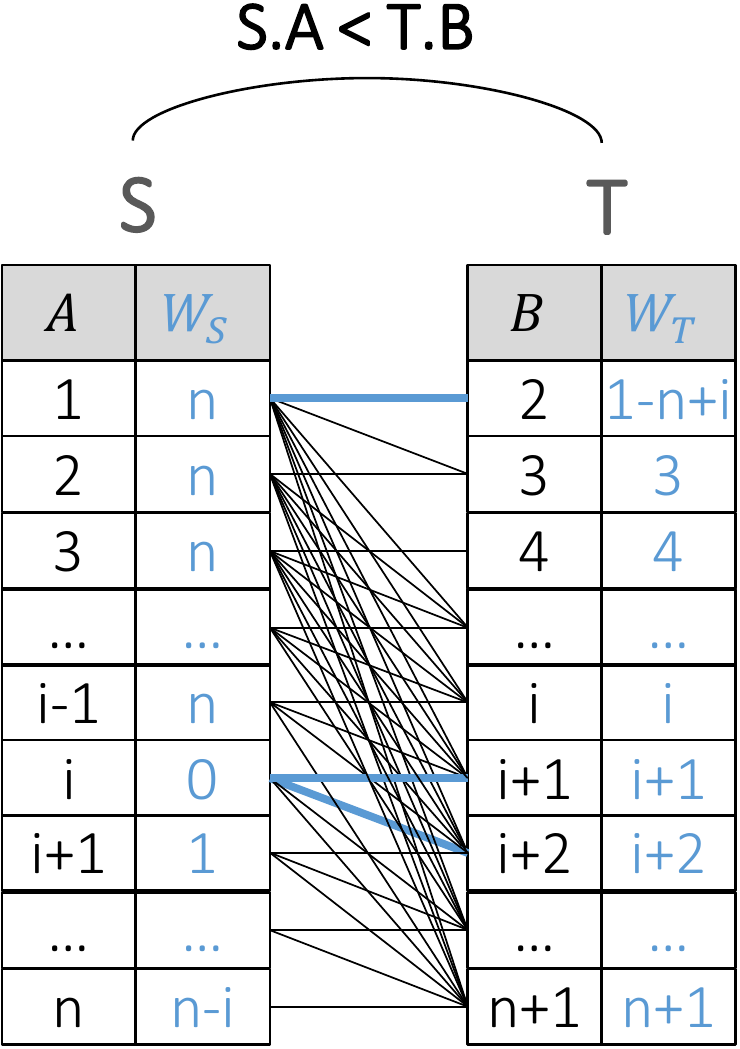}
    \caption{An example inequality-join.}
    \label{fig:Inequality_Ex_Quadratic}
\end{subfigure}%
\hfill
\begin{subfigure}[t]{.7\linewidth}
    \centering
    \includegraphics[height=5.2cm]{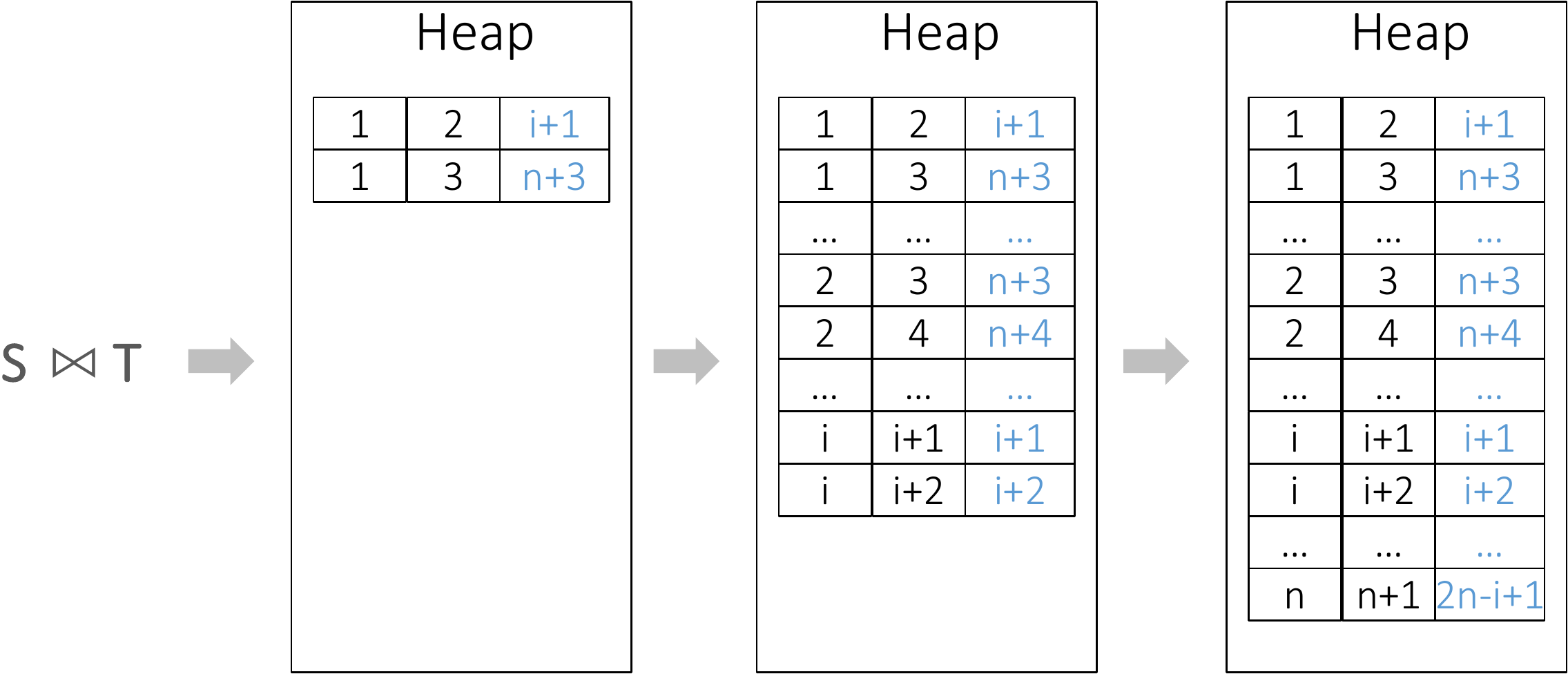}
    \caption{Steps of the computation. For readability, only heap content is shown,
    but the actual heap structure is omitted. Clearly, the heap data structure would
    organize the tuples in a better way.}
    \label{fig:Inequality_Ex_Heap}
\end{subfigure}%

\end{figure*}

While this works well if we want to get an \emph{arbitrary} set of $k$ result tuples,
\textbf{ranking makes the situation more challenging}.
To illustrate this, suppose in the example we want to find the top-$3$ join results
\textbf{according to the minimum sum of weights $W_S + W_T$}.
Notice that in general, tuple weights may or may not be correlated with join-attribute values.
In our example, we highlight the top-$3$ joining pairs $((1, \wtcl{n}), (2, \wtcl{1-n+i}))$,
$((i, \wtcl{0}), (i+1, \wtcl{i+1}))$, and $((i, \wtcl{0}), (i+2, \wtcl{i+2}))$ with blue edges,
where $i$ is some value $1 < i < n-1$.
Notice that even after sorting each relation by the join attributes, the algorithm
still does not know in which positions in each sorted relation
the winning $W_S + W_T$ combinations occur.
This means that as the join algorithm returns output tuples, the weight sum $W_S + W_T$
may go up or down between consecutive output tuples as
illustrated in \Cref{fig:Inequality_Ex_Heap}, where we show how the heap gradually fills up
with output tuples from the join.
We cannot determine the winners until all the $\O(n^2)$ join results have been inserted into the heap.
Even in the middle step where the top-$3$ results happen to be in the heap already,
\textbf{we cannot stop the join computation early because the algorithm does not
know if a not-yet-returned join output tuple could have a lower sum of weights}.
Only after all the join result tuples have been inserted into the heap
can the algorithm know for sure what the top-$k$ results based on weight $W_S + W_T$ are.
This implies that in order to find the top-$k$ results, even for a small value of $k$,
the algorithm must run the join until the end, i.e., consider all matching combinations
produced by the join.
No matter how efficient the join implementation or the heap data structure,
just looking once at each of the $\O(n^2)$ join output tuples already takes time
$\O(n^2)$---and this is the quadratic complexity we refer to.

One may look at the example and think ``couldn't we avoid having to look at the entire
join output by making join processing more aware of the weight attributes?''
And that is exactly what our algorithm does. 
The challenge is that when sorting the input
by $W_S$ and $W_T$, respectively, \textbf{the first
pairs of $S$ and $T$ tuples considered based on weight may not join at all.}
In our example, the lightest $S$-tuples are $(i, \wtcl{0}), (i+1, \wtcl{1})$,\ldots,
but unfortunately for larger values of $i$ they do not join with the 
lightest $T$-tuples $(2, \wtcl{1-n+1}), (3, \wtcl{3})$ etc.
Therefore, there is no guarantee that the winning pairs will be found in less than $\O(n^2)$ time when following the weight order on the input.
(This may seem ``not too bad'' for the specific example, but is a major concern
for more complex DNFs of inequalities and for joins of more than 2 relations.)

To summarize, there are $2$ non-trivial aspects of the problem:
($1$) determine which pairs of input tuples join with each other, and
($2$) rank the joining pairs by sum of weights or another given ranking function.
\textbf{No approach that we know of, including the sort-join-and-heap algorithm
can do both ($1$) and ($2$)---even for a join of only 2 relations---while guaranteeing
worst-case time complexity better than $\O(n^2)$}. This holds
even if one asks only for the $k$ top-ranked (by weight) results for some constant $k$.

The techniques proposed in our paper avoid that cost by \emph{joining and ranking simultaneously},
achieving end-to-end complexity of $\O(n \log n)$ for a 2-relation join with one inequality
or one band-join condition (and $\O(n \polylog n)$ for a general DNF of inequality conditions)
to retrieve the top-$k$ results.
Stated differently, it takes a non-trivial combination of both sorting by join attributes and
sorting by ranking function---and that is the core of our factorization approach.

\section{More Motivating Examples}

\begin{example}
Consider an ornithologist studying interactions between bird species
using a bird observation dataset
\texttt{B(Species, Family, ObsCount, Latitude, Longitude)}.
For her analysis, she decides to extract pairs of observations for birds of
different species from the same larger family that have been spotted
in the same region. Pairs with higher \texttt{ObsCount} should also appear
first:
\begin{lstlisting}
SELECT   *, B1.ObsCount + B2.ObsCount as Weight
FROM     B B1, B B2
WHERE    B1.Family = B2.Family  
  AND    ABS(B1.Latitude - B2.Latitude) < 1 
  AND    ABS(B1.Longitude - B2.Longitude) < 1 
  AND    B1.Species <> B2.Species    
ORDER BY Weight DESC LIMIT 1000
\end{lstlisting}
With $n$ denoting the number of tuples in $B$, no existing approach can guarantee
to return the top-$1000$ results in sub-quadratic time complexity $o(n^2)$.
In this paper, we show how to achieve $\O(n \log^3 n)$
even if the size of the output is $\O(n^2)$.
After returning the top-$1000$ answers, our approach is also capable of returning
more answers in order without having to restart the query.
The exponent of the logarithm is determined by the number of join predicates
that are not equalities (3 here).
Interestingly, this guarantee is not affected by the number of relations joined,
e.g., if we look for triplets of bird observations, because the complexity is
determined only by the \emph{pairwise} join with the most predicates
that are not equalities.
\end{example}

\end{document}